\documentclass[11pt]{article}%
\usepackage{amsmath}
\usepackage{amsfonts}
\usepackage{amssymb}
\usepackage{fullpage}
\usepackage{graphicx}%
\providecommand{\U}[1]{\protect\rule{.1in}{.1in}}
\newtheorem{theorem}{Theorem}

\newtheorem{conjecture}[theorem]{Conjecture}
\newtheorem{corollary}[theorem]{Corollary}

\newtheorem{definition}[theorem]{Definition}

\newtheorem{lemma}[theorem]{Lemma}

\newtheorem{problem}[theorem]{Problem}
\newtheorem{proposition}[theorem]{Proposition}

\newenvironment{proof}[1][Proof]{\noindent\textbf{#1.} }{\ \rule{0.5em}{0.5em}}
\begin{document}

\title{AM with Multiple Merlins}
\author{Scott Aaronson\thanks{MIT. \ Email: aaronson@csail.mit.edu.\ \ Supported by
the National Science Foundation under Grant No. 0844626,\ a TIBCO Chair, and
an Alan T. Waterman award.}
\and Russell Impagliazzo\thanks{UCSD. \ Email: russell@cs.ucsd.edu. \ Supported by
the Simons Foundation, the Ellentuck Fund, the Friends of the Institute for
Advanced Study, and NSF grants DMS-0835373, CCF-121351, and CCF-0832797
subcontract no. 00001583.}
\and Dana Moshkovitz\thanks{MIT. \ Email: dmoshkov@mit.edu. \ Supported by the
National Science Foundation under Grant No. 1218547.}}
\date{}
\maketitle

\begin{abstract}
We introduce and study a new model of interactive proofs: $\mathsf{AM}\left(
k\right)  $, or Arthur-Merlin with $k$ non-communicating Merlins. \ Unlike
with the better-known $\mathsf{MIP}$, here the assumption is that each Merlin
receives an \textit{independent} random challenge from Arthur. \ One
motivation for this model (which we explore in detail) comes from the close
analogies between it and the quantum complexity class $\mathsf{QMA}\left(
k\right)  $, but the $\mathsf{AM}\left(  k\right)  $\ model is also natural in
its own right.

We illustrate the power of multiple Merlins by giving an $\mathsf{AM}\left(
2\right)  $\ protocol for \textsc{3Sat}, in which the Merlins' challenges and
responses consist of only $n^{1/2+o\left(  1\right)  }$ bits each. \ Our
protocol has the consequence that, assuming the Exponential Time Hypothesis
(ETH), any algorithm for approximating a dense CSP with a polynomial-size
alphabet must take $n^{\left(  \log n\right)  ^{1-o\left(  1\right)  }}%
$\ time. \ Algorithms nearly matching this lower bound are known, but their
running times had never been previously explained. \ Brandao and Harrow have
also recently used our \textsc{3Sat} protocol to show quasipolynomial hardness
for approximating the values of certain entangled games.

In the other direction, we give a simple quasipolynomial-time approximation
algorithm for free games, and use it to prove that, assuming the ETH, our
\textsc{3Sat}\ protocol is essentially optimal. \ More generally, we show that
multiple Merlins never provide more than a polynomial advantage over one: that
is, $\mathsf{AM}\left(  k\right)  =\mathsf{AM}$ for all
$k=\operatorname*{poly}\left(  n\right)  $.\ \ The key to this result is a
\textit{subsampling theorem for free games}, which follows from powerful
results by Alon et al.\ and Barak et al.\ on subsampling dense CSPs, and which
says that the value of any free game can be closely approximated by the value
of a logarithmic-sized random subgame.

\end{abstract}
\tableofcontents

\section{Introduction\label{INTRO}}

The $\mathsf{PCP}$ characterization of $\mathsf{NP}$ \cite{almss,arorasafra}, with the
resulting hardness of approximation results, is one of the great achievements
of computational complexity. \ Leading up to this work was another landmark
result, the 1991\ theorem of Babai, Fortnow, and Lund \cite{bfl} that
$\mathsf{MIP}=\mathsf{NEXP}$, where $\mathsf{MIP}$ is Multi-Prover Interactive
Proofs and $\mathsf{NEXP}$\ is Nondeterministic Exponential Time. \ Both of
these results can be paraphrased as characterizing the hardness of a certain
computational problem from game theory: estimating the value of a two-player
cooperative game with simultaneous moves. \ Such games are known in the
complexity community as \emph{two-prover}\textit{ games}, and in the quantum
information community as \emph{nonlocal games}. \ From now on, we will use the
term two-prover games.

\begin{definition}
[Two-Prover Games]A two-prover game $G$ consists of:

\begin{enumerate}
\item[(1)] finite question sets\ $X,Y$ and answer sets\ $A,B$,

\item[(2)] a probability distribution $\mathcal{D}$\ over question pairs
$\left(  x,y\right)  \in X\times Y$, and

\item[(3)] a verification function\ $V:X\times Y\times A\times B\rightarrow
\left[  0,1\right]  $.\footnote{In most of the actual games we will consider,
$V$\ will take values in $\left\{  0,1\right\}  $ only. \ However, the
possibility of real $V$ is needed for full generality.}
\end{enumerate}

The \textit{value} of the game, denoted $\omega\left(  G\right)  $, is the
maximum, over all pairs of response functions\ $a:X\rightarrow A$ and
$b:X\rightarrow B$, of%
\begin{equation}
\operatorname*{E}_{\left(  x,y\right)  \sim\mathcal{D}}\left[  V\left(
x,y,a\left(  x\right)  ,b\left(  y\right)  \right)  \right]  .
\end{equation}

\end{definition}

The interpretation is this: the game $G$\ involves a verifier/referee Arthur,
as well as two cooperating provers Merlin$_{1}$\ and Merlin$_{2}$, who can
agree on a strategy in advance but cannot communicate once the game starts.
\ First Arthur chooses a pair of questions $\left(  x,y\right)  $ from
$\mathcal{D}$, and sends $x$\ to Merlin$_{1}$ and $y$ to Merlin$_{2}$. \ The
Merlins then send back responses $a=a\left(  x\right)  $\ and $b=b\left(
y\right)  $\ respectively.\footnote{Because of convexity, we can assume
without loss of generality that both Merlins use deterministic strategies.}
\ Finally, Arthur declares the Merlins to have \textquotedblleft
won\textquotedblright\ with probability equal to $V\left(  x,y,a,b\right)  $.
\ Then $\omega\left(  G\right)  $\ is just the probability that the Merlins
win if they use an optimal strategy.

It is not hard to show that computing the exact value of a two-prover game is
$\mathsf{NP}$-hard. \ The PCP Theorem can be interpreted as saying that even
to \emph{approximate} the value to within an additive constant is also
$\mathsf{NP}$-hard. \ To make this precise, we can define the classes
$\mathsf{PCP}$ and $\mathsf{MIP}$ as those decision problems
\emph{polynomial-time reducible} to approximating the value of a two-prover
game. \ The difference between the classes is that for PCP's, the reduction
computes an \emph{explicit} description of the game, whereas for
$\mathsf{MIP}$, the description is \emph{implicit}.

To be more precise, we start with a decision problem $L$. \ Given an instance
$I$ of $L$, a reduction constructs a two-prover game $G_{I}$ with the
following properties:

\begin{itemize}
\item \textbf{(Completeness)} If $I \in L$ then $\omega\left(  G_{P}\right)
\geq2/3$.

\item \textbf{(Soundness)} If $I\not \in L$ then $\omega\left(  G_{P}\right)
\leq1/3$.

\item \textbf{(Efficiency)} In the \textquotedblleft
explicit\textquotedblright\ case, the sets $X,Y,A,B$ can be generated in time
polynomial in $n=\left\vert I\right\vert $, the distribution $\mathcal{D}%
$\ can be described in polynomial time as the uniform distribution over some
subset of $X\times Y$, and and the verification procedure $V(x,y,a,b)$ can be
generated in polynomial time as a table of size $\left\vert X\right\vert
\times\left\vert Y\right\vert \times\left\vert A\right\vert \times\left\vert
B\right\vert $. \ In the \textquotedblleft implicit\textquotedblright\ case,
$X,Y,A,B$ are sets of $\operatorname*{poly}\left(  n\right)  $-bit\ strings,
$\mathcal{D}$ can be described as a probabilistic polynomial-time sampling
procedure that returns a pair $(x,y)\in X\times Y$, and the verification
function $V(x,y,a,b)$ can be computed in polynomial time.
\end{itemize}

The class $\mathsf{PCP}$ then consists of all decision problems that can be
reduced explicitly to two-prover games, while $\mathsf{MIP}$ consists of all
decision problems that can be reduced implicitly to two-prover games. \ As
frequently happens, switching from explicit to implicit representations causes
us to \textquotedblleft jump up\textquotedblright\ in complexity by an
exponential. \ The dual theorems $\mathsf{PCP}=\mathsf{NP}$ and $\mathsf{MIP}%
=\mathsf{NEXP}$\ bear out this general pattern.

The hardness of approximating two-prover games can in turn be used to show
hardness of approximation for many constraint-satisfaction problems. \ Better
trade-offs in the parameters of the reduction and specific kinds of
verification procedure give tighter hardness of approximation results for a
wide variety of particular combinatorial optimization problems. \ So the study
of two-prover games did not end with the PCP Theorem.

\subsection{Restricting to Independent Questions\label{INDEP}}

In this paper, we consider the following restriction of two-prover games:

\begin{quotation}
\noindent \textit{What if we demand that Arthur's challenges to Merlin}$_{1}%
$\textit{\ and Merlin}$_{2}$\textit{\ be independent? \ In other words, what
if the distribution }$\mathcal{D}$\textit{\ is simply the uniform distribution
over }$X\times Y$\textit{?}\footnote{We could also let $\mathcal{D}$ be an
arbitrary product distribution, but we don't gain any interesting generality
that way: Arthur might as well just send Merlin$_{1}$ and Merlin$_{2}$\ the
uniform random bits he would've used to generate $x\in X$ and $y\in Y$
respectively.}
\end{quotation}

In the PCP literature, two-prover games where $\mathcal{D}$\ is uniform over
$X\times Y$\ are called \textquotedblleft free\textquotedblright\ games, and
have sometimes been studied as an easier-to-analyze special case of general
games \cite{brrrs,shaltiel}. \ Free games are also tightly connected to dense
instances of constraint satisfaction problems. \ In this paper, we consider
approximating the values of free games as an interesting computational problem
in its own right, and one that has not received explicit attention. \ As far
as we know, we are the first to study the complexity of this problem directly,
and to formulate complexity classes of problems reducible to free games.

In more detail, the restriction to free games gives us an analogue of
\textquotedblleft public-coin\textquotedblright\ protocols in the
single-prover interactive proof setting. \ This corresponds to the original
definition of the class $\mathsf{AM}$, so we use a version of $\mathsf{AM}$
notation. \ We consider $\mathsf{AM}\left(  2\right)  $, or two-prover
Arthur-Merlin: the class of all languages that admit two-prover, two-round
interactive proof systems, in which Arthur's challenges to the Merlins are
independent, uniformly-random $\operatorname*{poly}\left(  n\right)  $-bit
strings. \ In other words, $\mathsf{AM}(2)$ is the class of problems
\textit{implicitly} reducible to approximating the value of a free game.
\ Clearly
\begin{equation}
\mathsf{AM}\subseteq\mathsf{AM}\left(  2\right)  \subseteq\mathsf{MIP}%
=\mathsf{NEXP}.
\end{equation}
We want to know: \textit{what is the true power of }$\mathsf{AM}\left(
2\right)  $\textit{?} \ Is it more powerful than single-prover $\mathsf{AM}$?
\ Is it less powerful than $\mathsf{MIP}$? \ We will also be interested in the
complexity of approximating the values of \textit{explicit} free games.

As we'll discuss in Section \ref{QUANTUM}, an additional motivation to study
$\mathsf{AM}\left(  2\right)  $\ comes from difficult analogous questions
about \textit{quantum} multi-prover proof systems, and specifically about the
quantum complexity class $\mathsf{QMA}\left(  2\right)  $. \ Our results could
shed light on $\mathsf{QMA}\left(  2\right)  $, by showing how many questions
about it get resolved in a simpler \textquotedblleft classical model
situation.\textquotedblright

\section{Our Results\label{RESULTS}}

We have two main sets of results: upper bounds, showing that the value of a
free game can be approximated in quasipolynomial time and translating that
into complexity class containments; and hardness results, giving almost
matching lower bounds for this problem under the Exponential Time Hypothesis
(ETH). \ Thus, assuming only the ETH, we show both that free games are
exponentially easier than general two-prover games, and \textit{also} they
still remain nontrivial, out of reach for polynomial-time algorithms.

\subsection{Upper Bounds\label{UB}}

Let \textsc{FreeGame}$_{\varepsilon}$\ be the problem of approximating the
value of a free game to error $\pm\varepsilon$:

\begin{problem}
[\textsc{FreeGame}$_{\varepsilon}$]Given as input a description of a free game
$G=\left(  X,Y,A,B,V\right)  $, estimate $\omega\left(  G\right)  $\ to within
additive error $\pm\varepsilon$. \ (Here $n$, the input size, is $\left\vert
X\right\vert \left\vert Y\right\vert \left\vert A\right\vert \left\vert
B\right\vert $, and\ $\varepsilon$\ is an arbitrarily small constant if not
specified explicitly.)
\end{problem}

We give a quasipolynomial-time algorithm for \textsc{FreeGame}$_{\varepsilon}$:

\begin{theorem}
\label{freegame0}\textsc{FreeGame}$_{\varepsilon}$ is solvable in
deterministic time $n^{O(\varepsilon^{-2}\log n)}$.
\end{theorem}

While this is the first algorithm explicitly for \textsc{FreeGame}%
$_{\varepsilon}$, there is some directly-related algorithmic work. \ After
learning of our results (but before this paper was written), Brandao and
Harrow \cite{brandaoharrow:freegame} gave an algorithm for \textsc{FreeGame}%
$_{\varepsilon}$ with the same running time as ours, but using interestingly
different techniques. \ (Our algorithm is purely combinatorial, whereas theirs
uses linear programming relaxation.) \ Also, Barak et al.\ \cite{bhhs} gave a
quasipolynomial-time approximation algorithm for the related problem of
approximating the values of dense CSPs with polynomial-sized alphabets.

In the implicit setting, Theorem \ref{freegame0}\ implies that $\mathsf{AM}%
\left(  2\right)  \subseteq\mathsf{EXP}$, which improves on the trivial upper
bound of $\mathsf{NEXP}$. \ However, by building on the result of Barak et al.
\cite{bhhs} mentioned above, we are able to prove a stronger result, which
completely characterizes $\mathsf{AM}\left(  2\right)  $:

\begin{theorem}
\label{am2am0}$\mathsf{AM}\left(  2\right)  =\mathsf{AM}$.
\end{theorem}

We can even generalize Theorem \ref{am2am0} to handle any polynomial number of Merlins:

\begin{theorem}
$\mathsf{AM}\left(  k\right)  =\mathsf{AM}$ for all $k=\operatorname*{poly}%
\left(  n\right)  $.
\end{theorem}

Thus, in the complexity class setting, it is really the correlation between
queries that makes multiple provers more powerful than a single prover.

\subsection{Hardness Results\label{LB}}

Seeing just the above, one might conjecture that the values of free games are
approximable in polynomial time. \ But surprisingly, we give strong evidence
that this is \textit{not} the case.

To show the power of free games, we give a nontrivial reduction from
\textsc{3Sat} to \textsc{FreeGame}. \ Equivalently, we show that \textit{there
exists a nontrivial }$\mathsf{AM}\left(  2\right)  $\textit{\ protocol}: even
if Arthur's challenges are completely independent, two Merlins can be more
helpful to him than one Merlin. \ In particular, given a \textsc{3Sat}
instance $\varphi$, let the \textit{size} of $\varphi$ be the number of
variables plus the number of clauses. \ Then:

\begin{theorem}
\label{mainthm1}For some constant $\varepsilon>0$, there exists a reduction
running in time $2^{\widetilde{O}(\sqrt{n})}$\ that maps \textsc{3Sat}
instances of size $n$ to \textsc{FreeGame}$_{\varepsilon}$ instances of size
$2^{\widetilde{O}(\sqrt{n})}$ (where the $\widetilde{O}$ hides polylogarithmic factors).
\end{theorem}

In other words, there is a protocol whereby Arthur can check that a
\textsc{3Sat} instance of size $n$ is satisfiable, by exchanging only
$\widetilde{O}(\sqrt{n})$ bits with the Merlins---i.e., sending $\widetilde{O}%
(\sqrt{n})$-bit challenges and receiving $\widetilde{O}(\sqrt{n})$-bit
responses. \ The protocol has perfect completeness and a\ $1$
vs.\ $1-\varepsilon$\ completeness/soundness gap, for some fixed constant
$\varepsilon>0$. \ Since the first step we use is the PCP Theorem, by
composing our main protocol with various PCP constructions, we can get
reductions with different quantitative tradeoffs between reduction time,
completeness, soundness, and alphabet size.

One corollary of Theorem \ref{mainthm1}\ is that, if \textsc{FreeGame} is in
$\mathsf{P}$, then\ \textsc{3Sat} is in $\mathsf{TIME}(2^{\widetilde{O}%
(\sqrt{n})})$. \ Since \textsc{3Sat} is complete under quasilinear-time
reductions for $\mathsf{NTIME}(n)$, the same holds for any problem in
nondeterministic linear time. \ As a second corollary, we get a lower bound on
the time to approximate \textsc{FreeGame} assuming the ETH. \ This lower bound
almost matches the upper bounds described in Section \ref{UB}. \ To be more
precise, recall the \textit{Exponential Time Hypothesis} (ETH) of Impagliazzo
and Paturi \cite{ip:eth}:

\begin{conjecture}
[Exponential Time Hypothesis \cite{ip:eth}]Any deterministic algorithm for
\textsc{3Sat}\ requires $2^{\Omega\left(  n\right)  }$\ time. \ (There is also
the Randomized ETH, which says the same for bounded-error randomized algorithms.)
\end{conjecture}

Then we show the following:

\begin{corollary}
[Hardness of Free Games]\label{freegamehard}Assuming the (randomized) ETH, any
(randomized) algorithm for \textsc{FreeGame}$_{\varepsilon}$ requires
$n^{\widetilde{\Omega}(\varepsilon^{-1}\log n)}$\ time, for all $\varepsilon
\geq1/n$ bounded below some constant.
\end{corollary}

Again, by considering various PCP constructions, we get a variety of hardness
results for many interesting versions and ranges of parameters for the
\textsc{FreeGame} problem.

We can further reduce \textsc{FreeGame} to the problem of approximating dense
CSPs, where an arity $k$ CSP is considered dense if it contains constraints
for a constant fraction of all $k$-tuples of variables. \ We thus get the
following hardness result for dense CSPs.

\begin{corollary}
\label{ethcor}Assuming the ETH, the problem of approximating a \textit{dense
}$k$\textit{-CSP} (constraint satisfaction problem) with a polynomial-size
alphabet, to constant additive error, requires $n^{\widetilde{\Omega}\left(
\log n\right)  }$\ time, for any $k\geq2$.
\end{corollary}

Corollary \ref{ethcor} almost matches the upper bound of Barak et al.
\cite{bhhs}, explaining for the first time why Barak et al.\ were able to give
a quasipolynomial-time algorithm for approximating dense CSPs, but not a
polynomial-time one.

As another application of our hardness result\ for \textsc{FreeGame}, Brandao
and Harrow \cite{brandaoharrow:freegame}\ were recently able to use it to
prove that approximating the values of certain entangled games requires
$n^{\widetilde{\Omega}\left(  \log n\right)  }$\ time, assuming the
ETH.\footnote{See \cite{brandaoharrow:freegame}\ for the precise definition of
the entangled games they consider. \ Briefly, though, the games involve a
large number of provers, of whom two are selected at random to receive
challenges (the other provers are ignored).}

\section{Detailed Overview of Results\label{DETAILED}}

We now proceed to more detailed overview of our results and the techniques
used to prove them. \ Here, as in the technical part of the paper, we first
describe our hardness results for \textsc{FreeGame} (or equivalently,
$\mathsf{AM}\left(  2\right)  $\ protocols for \textsc{3Sat}), and then our
approximation algorithms (or equivalently, limitations of $\mathsf{AM}\left(
k\right)  $\ protocols).

\subsection{\textsc{3Sat} Protocol\label{3SATINT}}

The idea of our \textsc{3Sat} protocol is simple. \ First Arthur transforms
the \textsc{3Sat}\ instance $\varphi$\ into a PCP, so that it's either
satisfiable or far from satisfiable. \ For this to work, we need a
highly-efficient PCP theorem, which produces instances of near-linear size.
\ Fortunately, such PCP theorems are now known. \ Depending on the desired
parameters, we will use either the theorem of Dinur \cite{dinur} (which
produces \textsc{3Sat}\ instances of size $n\operatorname*{polylog}n$\ with a
small constant completeness/soundness gap), or that of Moshkovitz and Raz
\cite{mr} (which produces $2$-CSP\ instances of size $n\cdot2^{\left(  \log
n\right)  ^{1-\Omega(1)}}$\ with completeness/soundness gap arbitrarily close
to $1$).

Suppose for now that we use the PCP theorem of Dinur \cite{dinur}. \ Then
next, Arthur runs a variant of the so-called \textit{clause/variable game},
which we define below.

\begin{definition}
[Clause/Variable Game]Given a \textsc{3Sat} instance $\varphi$, consisting of
$n$ variables $x_{1},\ldots,x_{n}$\ and $m$\ clauses $C_{1},\ldots,C_{m}$, the
clause/variable game $G_{\varphi}$\ is defined as follows. \ Arthur chooses an
index $i\in\left[  m\right]  $\ uniformly at random, then chooses $j\in\left[
n\right]  $\ uniformly at random conditioned on $x_{j}$\ or $\urcorner x_{j}%
$\ appearing in $C_{i}$\ as a literal. \ He sends $i$ to Merlin$_{1}$\ and $j$
to Merlin$_{2}$. \ Arthur accepts if and only if

\begin{enumerate}
\item[(i)] Merlin$_{1}$ sends back a satisfying assignment to the variables in
$C_{i}$, and

\item[(ii)] Merlin$_{2}$ sends back a value for $x_{j}$\ that agrees with the
value sent by Merlin$_{1}$.
\end{enumerate}
\end{definition}

Let $\operatorname*{SAT}\left(  \varphi\right)  \in\left[  0,1\right]  $\ be
the maximum fraction of clauses of $\varphi$\ that can be simultaneously
satisfied. \ Then clearly the clause/variable game has \textit{perfect
completeness}: that is, if $\operatorname*{SAT}\left(  \varphi\right)
=1$\ then $\omega\left(  G_{\varphi}\right)  =1$. \ The following well-known
proposition shows that the game also has \textit{constant soundness}.

\begin{proposition}
\label{cvsound}If $\operatorname*{SAT}\left(  \varphi\right)  \leq
1-\varepsilon$, then $\omega\left(  G_{\varphi}\right)  \leq1-\varepsilon/3$.
\end{proposition}

\begin{proof}
Assume without loss of generality that Merlin$_{2}$\ answers according to a
particular assignment $x=\left(  x_{1},\ldots,x_{n}\right)  $. \ By
hypothesis, $x$\ violates the clause $C_{i}$\ with probability at least
$\varepsilon$\ over $i$. \ And if $x$ violates $C_{i}$, then regardless of
what Merlin$_{1}$\ does,\ Arthur rejects with probability at least
$1/3$---since Merlin$_{1}$'s assignment to $C_{i}$\ either violates $C_{i}$,
or else disagrees with $x$\ (which Arthur detects with probability at least
$1/3$ over the variable sent to Merlin$_{2}$).
\end{proof}

Also, given any two-prover game $G=\left(  X,Y,A,B,\mathcal{D},V\right)  $,
let $G^{k}$\ be the $k$\textit{-fold parallel repetition} of $G$: that is, the
game where Arthur

\begin{enumerate}
\item[(1)] draws $\left(  x_{1},y_{1}\right)  ,\ldots,\left(  x_{k}%
,y_{k}\right)  $ independently from $\mathcal{D}$,

\item[(2)] sends $x_{1},\ldots,x_{k}$\ to Merlin$_{1}$ and $y_{1},\ldots
,y_{k}$\ to Merlin$_{2}$,

\item[(3)] receives responses $a_{1},\ldots,a_{k}\in A$\ from Merlin$_{1}%
$\ and $b_{1},\ldots,b_{k}\in B$\ from Merlin$_{2}$, and then

\item[(4)] accepts with probability equal to $\prod_{i=1}^{k}V\left(
x_{i},y_{i},a_{i},b_{i}\right)  $.
\end{enumerate}

Then the famous Parallel Repetition Theorem asserts that $\omega(G^{k}%
)$\ decreases exponentially with $k$:

\begin{theorem}
[Parallel Repetition Theorem \cite{raz:prt,holenstein}]\label{prt}If
$\omega\left(  G\right)  \leq1-\varepsilon$, then%
\begin{equation}
\omega(G^{k})\leq\left(  1-\varepsilon^{3}\right)  ^{\Omega\left(
k/\log\left\vert A\right\vert \left\vert B\right\vert \right)  }.
\end{equation}

\end{theorem}

Unfortunately, neither the original clause/variable game $G_{\varphi}$, nor
its parallel repetition $G_{\varphi}^{k}$, work in the setting of
$\mathsf{AM}\left[  2\right]  $. \ For both games rely essentially on
\textit{correlation} between the clause(s) sent to Merlin$_{1}$\ and the
variable(s) sent to Merlin$_{2}$. \ To eliminate the need for correlation, we
use a new form of repetition that we call \textit{birthday repetition}.

\begin{definition}
[Birthday Repetition]\label{bdr}Let $G=\left(  X,Y,A,B,\mathcal{D},V\right)
$\ be a two-prover game with $V\in\left\{  0,1\right\}  $ (not necessarily
free). \ Assume $\mathcal{D}$\ is just the uniform distribution over some
subset $Z\subseteq X\times Y$. \ Then given positive integers $k\leq\left\vert
X\right\vert $\ and $\ell\leq\left\vert Y\right\vert $, the birthday
repetition $G^{k\times\ell}$\ is the free game defined as follows. \ Arthur
chooses subsets $S\subseteq X$\ and $T\subseteq Y$\ uniformly at random,
subject to $\left\vert S\right\vert =k$\ and $\left\vert T\right\vert =\ell$.
\ He sends $S$\ to Merlin$_{1}$ and asks for an assignment $a:S\rightarrow A$,
and sends $T$ to Merlin$_{2}$ and asks for an assignment $b:T\rightarrow B$.
\ Arthur accepts if and only if $V\left(  x,y,a\left(  x\right)  ,b\left(
y\right)  \right)  =1$\ for all $\left(  x,y\right)  \in S\times T$\ that
happen to lie in $Z$. \ (So in particular, if $\left(  S\times T\right)  \cap
Z$\ is empty, then Arthur always accepts.)
\end{definition}

Now consider the birthday repetition $G_{\varphi}^{k\times\ell}$\ of the
clause/variable game $G_{\varphi}$. \ In this game, Arthur chooses $k$\ random
clause indices $i_{1},\ldots,i_{k}$\ and sends them to Merlin$_{1}$, and
chooses $\ell$\ random variable indices $j_{1},\ldots,j_{\ell}$\ and sends
them to Merlin$_{2}$. \ He then sends $i_{1},\ldots,i_{k}$ to Merlin$_{1}$ and
asks for assignments to $C_{i_{1}},\ldots,C_{i_{k}}$, and sends $j_{1}%
,\ldots,j_{\ell}$\ to Merlin$_{2}$ and asks for assignments to $x_{j_{1}%
},\ldots,x_{j_{\ell}}$. \ Finally, Arthur accepts if and only if the
assignments to $C_{i_{1}},\ldots,C_{i_{k}}$\ satisfy those clauses,
\textit{and} are consistent with $x_{j_{1}},\ldots,x_{j_{\ell}}$\ on any
variables where they happen to intersect.

If $\varphi$ is satisfiable, then clearly $\omega\left(  G_{\varphi}%
^{k\times\ell}\right)  =1$. \ Our main result says that, if $\varphi$\ is far
from satisfiable and $k,\ell=\Omega\left(  \sqrt{n}\right)  $, then
$\omega\left(  G_{\varphi}^{k\times\ell}\right)  \leq1-\Omega\left(  1\right)
$. \ This result is \textquotedblleft intuitively plausible,\textquotedblright%
\ since if $k\ell=\Omega\left(  n\right)  $, then by the Birthday Paradox,
there's a constant probability that some $x_{j_{t}}$\ will occur as a literal
in some $C_{i_{s}}$, giving Arthur a chance to catch the Merlins in an
inconsistency if $\varphi$\ is far from satisfiable. \ But of course, any
soundness proof needs to account for the fact that Merlin$_{1}$\ sees the
entire list $C_{i_{1}},\ldots,C_{i_{k}}$, while Merlin$_{2}$\ sees the entire
list $x_{j_{1}},\ldots,x_{j_{\ell}}$! \ So it's conceivable that the Merlins
could cheat using some clever correlated strategy. \ We will rule that
possibility out, by showing that any cheating strategy for $G_{\varphi
}^{k\times\ell}$\ can be converted (with help from some combinatorial counting
arguments) into a cheating strategy for the original clause/variable game
$G_{\varphi}$.

One might worry that any proof of a \textquotedblleft Birthday Repetition
Theorem\textquotedblright\ would need to be at least as complicated as the
proof of the original Parallel Repetition Theorem. \ Fortunately, though, we
can get by with a relatively simple proof, for two reasons. \ First, we will
not prove that birthday repetition works for \textit{every} game $G$ or for
\textit{every} $k$ and $\ell$, for the simple reason that this is
false!\footnote{As a silly counterexample, let $G$ be the free game with
$X=Y=A=B=\left[  n\right]  $, where the Merlins lose if and only if $x=1$.
\ Then clearly $\omega\left(  G\right)  =1-1/n$\ and $\omega\left(
G^{k\times\ell}\right)  =1-k/n$, with no dependence on $\ell$. \ More
generally, it is not hard to see that $\omega\left(  G^{k\times\ell}\right)
\geq\max\left\{  \left\vert A\right\vert ^{-k},\left\vert B\right\vert
^{-\ell}\right\}  $\ for every game $G$\ with $\omega(G)>0$, since this is
achieved if one Merlin responds randomly, while the other Merlin
\textit{guesses} the first Merlin's responses and then responds optimally.
\ This implies the following result, for any game $G$. \ Let $\omega
(G^{1\times1})=1-\varepsilon$ (note that if $G$ is free, then $G^{1\times1}%
=G$, while otherwise $G^{1\times1}$\ is a\ \textquotedblleft
promise-free\textquotedblright\ version of $G$). \ Then the value
$\omega\left(  G^{k\times\ell}\right)  $\ can only decrease like
$\omega(G^{1\times1})^{\Omega(k\ell)}$\ so long as $k=O(\frac{1}{\varepsilon
}\log\left\vert B\right\vert )$ and $\ell=O(\frac{1}{\varepsilon}%
\log\left\vert A\right\vert )$.} \ Instead, our proof will use a special
property of the clause/variable game $G_{\varphi}$: namely, the fact that it
arises from a uniform constraint graph. \ Second, we are happy if we can
\textquotedblleft merely\textquotedblright\ construct a free game that
preserves\ the soundness of the original game $G$: amplifying $G$'s soundness
even further would be a bonus, but is not necessary. \ We leave it to future
work to determine the power of birthday repetition\ more generally.

\subsection{Approximation Algorithms for Free Games\label{ALGINT}}

Our second set of results aims at showing that a square-root savings in
communication, as achieved by our $\mathsf{AM}\left(  2\right)  $ protocol for
\textsc{3Sat}, is the \textit{best} that any such protocol can provide. \ More
formally, we prove the following set of four interrelated results:

\begin{enumerate}
\item[(1)] The \textsc{FreeGame}$_{\varepsilon}$ problem is solvable
deterministically in $\left(  \left\vert X\right\vert \left\vert A\right\vert
\right)  ^{O(\varepsilon^{-2}\log\left\vert Y\right\vert \left\vert
B\right\vert )}=n^{O(\varepsilon^{-2}\log n)}$ time. \ (There is also a
randomized algorithm that uses $\left\vert X\right\vert \cdot\left\vert
A\right\vert ^{O(\varepsilon^{-2}\log\left\vert Y\right\vert \left\vert
B\right\vert )}$\ time.)

\item[(2)] Any $\mathsf{AM}\left(  2\right)  $\ protocol involving $p\left(
n\right)  $\ bits of communication can be simulated in $2^{O(p\left(
n\right)  ^{2})}\operatorname*{poly}\left(  n\right)  $\ time
(deterministically, if Arthur's verification procedure is deterministic, and
probabilistically otherwise). \ So in particular, $\mathsf{AM}\left(
2\right)  \subseteq\mathsf{EXP}$, improving the trivial upper bound of
$\mathsf{NEXP}$. \ (As we point out, a closer analysis improves the upper
bound to $\mathsf{AM}\left(  2\right)  \subseteq\mathsf{AM}^{\mathsf{NP}}$.)

\item[(3)] Assuming the Randomized ETH, any constant-soundness $\mathsf{AM}%
\left(  2\right)  $\ protocol for \textsc{3Sat}\ must use $\Omega(\sqrt{n}%
)$\ communication. \ (In more detail, such a protocol must use $\Omega
(\sqrt{\varepsilon n})$\ communication if its completeness/soundness gap is
$1$\ vs.\ $1-\varepsilon$, and $\Omega(\sqrt{n\log1/\delta})$\ communication
if its gap is $1$\ vs.\ $\delta$. \ Also, if Arthur's verification procedure
is deterministic, then it suffices to assume the standard ETH.)

\item[(4)] $\mathsf{AM}\left(  2\right)  =\mathsf{AM}$. \ (Of course, this
supersedes our $\mathsf{AM}\left(  2\right)  \subseteq\mathsf{EXP}$ and
$\mathsf{AM}\left(  2\right)  \subseteq\mathsf{AM}^{\mathsf{NP}}$\ results.)
\end{enumerate}

In Section \ref{ALGSEC}, we provide a self-contained proof for result (1), and
then use (1) to deduce (2) and (3). \ The idea of our approximation algorithm
is to sample a small random subset $S\subset X$\ of the questions to
Merlin$_{1}$. \ We then brute-force search over all possible strategies
$\alpha:S\rightarrow A$\ for the questions in $S$. \ For each such strategy
$\alpha$, we find the optimal response $b_{\alpha}:Y\rightarrow B$\ of
Merlin$_{2}$ to that $\alpha$, and then the optimal response $a_{\alpha
}:X\rightarrow A$\ of Merlin$_{1}$\ to $b_{\alpha}$\ on his full question set
$X$. \ A simple probabilistic analysis then shows that, provided we take
$\left\vert S\right\vert =\Omega\left(  \varepsilon^{-2}\log\left\vert
Y\right\vert \left\vert B\right\vert \right)  $,\ at least one of these
\textquotedblleft induced\textquotedblright\ strategy pairs $\left(
a_{\alpha},b_{\alpha}\right)  $\ must achieve value within $\varepsilon$\ of
the optimal value $\omega\left(  G\right)  $. \ Similar ideas have been used
before in other approximation algorithms: for example, in that of Lipton,
Markakis, and Mehta \cite{lmm} for finding approximate Nash equilibria.

Once we have an $n^{O\left(  \varepsilon^{-2}\log n\right)  }$-time
approximation algorithm for \textsc{FreeGame}$_{\varepsilon}$, the containment
$\mathsf{AM}\left(  2\right)  \subseteq\mathsf{EXP}$\ follows almost
immediately. \ We also sketch an improvement to $\mathsf{AM}\left(  2\right)
\subseteq\mathsf{AM}^{\mathsf{NP}}$, which is obtained by modifying our
approximation algorithm so that it fits into the property-testing framework of
Goldreich, Goldwasser, and Ron \cite{ggr}. \ As for the optimality of our
\textsc{3Sat}\ protocol, we simply need to observe that, if we had
a\ protocol\ that used $o(\sqrt{n})$\ communication, then it would give rise
to a free game $G$\ of size $2^{o(\sqrt{n})}$, whose value $\omega\left(
G\right)  $\ we could estimate in $2^{o\left(  n\right)  }$\ time by using our
quasipolynomial-time approximation algorithm. \ But that would let us decide
\textsc{3Sat}\ in $2^{o\left(  n\right)  }$\ time, contradicting the
Exponential Time Hypothesis.

For result (4), we wish to go further, and show that any two-Merlin protocol
can be simulated using \textit{one} Merlin:\ that is, $\mathsf{AM}\left(
2\right)  =\mathsf{AM}$. \ Here we appeal to a powerful line of earlier work
on \textit{subsampling for dense CSPs}. \ Specifically, Alon et
al.\ \cite{avkk} showed in 2002 that, given any $k$-ary constraint
satisfaction problem $\varphi$\ over $n$ Boolean variables, one can estimate
the maximum number of constraints in $\varphi$ that can be simultaneously
satisfied, to within additive error $\pm\varepsilon\binom{n}{k}$, by simply
throwing away all the variables except for a random set $I$ of size
$\operatorname*{poly}\left(  1/\varepsilon\right)  $, and then using
brute-force search to find an optimal assignment to $\varphi_{I}$, the
restriction of $\varphi$\ to $I$.

To build intuition, it is easy to satisfy $\varphi_{I}$\ \textit{at least as
well} as we can satisfy $\varphi$, with high probability over $I$. \ To do so,
simply start with an optimal global assignment $x$\ for $\varphi$; then
restrict $x$ to the variables in $I$ and apply a Chernoff bound. \ The hard
part is to show that $\varphi_{I}$\ cannot be satisfied much \textit{better}
than the full instance $\varphi$\ was. \ Conversely, one needs to show that,
given a collection of \textquotedblleft local assignments,\textquotedblright%
\ involving just $\operatorname*{poly}\left(  1/\varepsilon\right)
$\ variables at a time, one can \textquotedblleft patch them
together\textquotedblright\ into a global assignment that is almost as good as
the local ones.

In later work, Barak et al.\ \cite{bhhs}\ proved a more general result, which
removed Alon et al.'s assumption that the alphabet is Boolean. \ Their result
lets us approximate the value of any dense $k$-CSP $\varphi$ over the finite
alphabet $\Sigma$ to within additive error $\pm\varepsilon\binom{n}{k}$, by
solving a random sub-instance\ on\ $\operatorname*{poly}\left(  1/\varepsilon
\right)  \cdot\log\left\vert \Sigma\right\vert $ variables.

To see the relevance of this work to free games, we simply need to observe
that \textsc{FreeGame}\ can be directly encoded as a dense CSP. \ Given a free
game $G=\left(  X,Y,A,B,V\right)  $, we can create variables $\left(  a\left(
x\right)  \right)  _{x\in X}$\ and $\left(  b\left(  y\right)  \right)  _{y\in
Y}$\ over the alphabets $A$\ and $B$ respectively, and then for all $\left(
x,y,a,b\right)  \in X\times Y\times A\times B$, add a number of constraints
setting $a\left(  x\right)  =a$\ and $b\left(  y\right)  =b$\ that is
proportional to $V\left(  x,y,a,b\right)  $. \ Once we do this, the result of
Barak et al.\ \cite{bhhs}\ implies a \textit{subsampling theorem for free
games}---saying that the value of any free game $G$ can be well-approximated
by the value of a logarithmic-sized random subgame. \ And this, in turn,
readily implies that $\mathsf{AM}\left(  2\right)  =\mathsf{AM}$. \ For given
any $\mathsf{AM}\left(  2\right)  $\ protocol, we can simulate the protocol in
$\mathsf{AM}$\ by having Arthur execute the following steps:

\begin{enumerate}
\item[(i)] Choose random subsets $S,T$\ of $\operatorname*{poly}\left(
n\right)  $\ questions to Merlin$_{1}$ and Merlin$_{2}$ respectively.

\item[(ii)] Ask a \textit{single} Merlin to send him responses to all
questions in $S$\ and $T$.

\item[(iii)] Check the responses, for all possible question pairs $\left(
x,y\right)  \in S\times T$.
\end{enumerate}

The soundness of this approach follows from the subsampling theorem, which
says that if Merlins had no winning strategy in the original $\mathsf{AM}%
\left(  2\right)  $ protocol, then with high probability, they have no winning
strategy even when restricted to the tiny subset of questions $S\times T$.

One might ask: if existing results on dense CSPs can be used to show that
$\mathsf{AM}\left(  2\right)  =\mathsf{AM}$, then why do we \textquotedblleft
reinvent the wheel,\textquotedblright\ and provide self-contained proofs for
weaker results such as $\mathsf{AM}\left(  2\right)  \subseteq\mathsf{EXP}$?
\ One answer is that the dense CSP results do not give good dependence on the
error. \ For example, those results imply that \textsc{FreeGame}%
$_{\varepsilon}$\ can be solved in $n^{O(\varepsilon^{-\Lambda}\log n)}%
$\ time\ for some large and unspecified constant $\Lambda$, but not that it
can be solved in $n^{O(\varepsilon^{-2}\log n)}$ time. \ And we actually care
about the dependence on $\varepsilon$, for at least two reasons. \ First, we
wish to make an analogy with a recent $n^{O(\varepsilon^{-2}\log n)}%
$\ algorithm for a problem in quantum information theory, due to Brandao,
Christandl, and Yard \cite{bcy} (for details see Section \ref{QUANTUM}). \ And
second, we wish to show that, assuming the ETH, the \textquotedblleft
obvious\textquotedblright\ $\mathsf{AM}\left(  2\right)  $\ protocol
for\ \textsc{3Sat}\ is optimal even in the very low-error and high-error
cases. \ The dense CSP results do not get us close to such a statement, but
our algorithm does.

More broadly, appealing to the dense CSP literature\ feels like overkill if we
just want to show (for example) that our \textsc{3Sat}\ protocol is optimal,
or that the values of free games can be approximated in quasipolynomial time.
\ If we \textit{can} prove those results in an elementary, self-contained way,
then it seems like we should---particularly because our proofs might help to
make certain striking techniques from the dense CSP world more accessible than
they would be otherwise.

\ Their algorithm also implies that $\mathsf{AM}\left(  2\right)
\subseteq\mathsf{EXP}$, and that our \textsc{3Sat}\ protocol is essentially
optimal assuming the ETH. \ On the other hand, it seems unlikely that their
algorithm can be used to get the containment $\mathsf{AM}\left(  2\right)
\subseteq\mathsf{AM}^{\mathsf{NP}}$, let alone $\mathsf{AM}\left(  2\right)
=\mathsf{AM}$.

\subsection{Generalizing to $k$ Merlins\label{GENERALK}}

One might wonder whether our limitation theorems for $\mathsf{AM}\left(
2\right)  $\ protocols could be evaded by simply adding more Merlins. \ So for
example, even if $\mathsf{AM}\left(  2\right)  $\ protocols for \textsc{3Sat}%
\ require $\Omega(\sqrt{n})$\ communication (assuming the ETH), could there be
an $\mathsf{AM}\left(  3\right)  $\ protocol that used $O(n^{1/3}%
)$\ communication, an $\mathsf{AM}\left(  10\right)  $\ protocol that used
$O(n^{1/10})$\ communication, and so forth? \ In Sections \ref{KMERLIN} and
\ref{KSUBSAMP}, we generalize our limitation theorems to the case of $k$
Merlins, in order to rule out that possibility. \ In particular, we give the
following extensions of our results from Section \ref{ALGINT}:

\begin{enumerate}
\item[(1')] There is a deterministic algorithm that, given as input a
$k$-player free game $G$ with question sets $Y_{1},\ldots,Y_{k}$\ and answer
sets $B_{1},\ldots,B_{k}$, approximates $\omega\left(  G\right)  $ to within
$\pm\varepsilon$\ in time
\begin{equation}
\exp\left(  \frac{k^{2}}{\varepsilon^{2}}\sum_{i<j}\log\left(  \left\vert
Y_{i}\right\vert \left\vert B_{i}\right\vert \right)  \cdot\log\left(
\left\vert Y_{j}\right\vert \left\vert B_{j}\right\vert \right)  \right)
=n^{O(\varepsilon^{-2}k^{2}\log n)},
\end{equation}
where $n=\left\vert Y_{1}\right\vert \left\vert B_{1}\right\vert
\cdots\left\vert Y_{k}\right\vert \left\vert B_{k}\right\vert $\ is the input
size. \ (There is also an alternative algorithm that runs in time
$n^{\varepsilon^{-O\left(  1\right)  }\log n}$, independently of $k$.)

\item[(2')] $\mathsf{AM}\left(  k\right)  \subseteq\mathsf{EXP}$ for all
$k=\operatorname*{poly}\left(  n\right)  $. \ (Indeed, any constant-soundness
$\mathsf{AM}\left(  k\right)  $\ protocol involving $p\left(  n\right)
$\ total bits of communication can be simulated in $2^{O(p\left(  n\right)
^{2})}\operatorname*{poly}\left(  n\right)  $\ randomized time, or
$2^{O(p\left(  n\right)  ^{2})}\operatorname*{poly}\left(  n\right)
$\ deterministic time if Arthur's verification procedure is deterministic.)

\item[(3')] Assuming the Randomized ETH, any constant-soundness $\mathsf{AM}%
\left(  k\right)  $\ protocol for \textsc{3Sat}\ must use $\Omega(\sqrt{n}%
)$\ total bits of communication, regardless of how large $k$ is. (If,
moreover, Arthur's verification procedure is deterministic, then it suffices
to assume the ordinary ETH.)

\item[(4')] $\mathsf{AM}\left(  k\right)  =\mathsf{AM}$ for all
$k=\operatorname*{poly}\left(  n\right)  $.
\end{enumerate}

We first prove (1'), and then derive (2') and (3') as consequences. \ For
(1'), the basic idea is to generalize our approximation algorithm for
$2$-player free games to $k$-player games, by calling the algorithm
recursively to \textquotedblleft peel off players one at a
time.\textquotedblright\ \ In other words, we reduce the approximation of a
$k$-player game to the approximation of a quasipolynomial number of $\left(
k-1\right)  $-player games, and continue recursing until we get down to $1$
player. \ When we do this, we need to control the buildup of error across all
$k$ levels of the recursion, and that is why we get a factor of $k^{2}$\ in
the exponent of the running time. \ Later, by using the subsampling machinery,
we will be able to go back and give an alternative algorithm whose running
time depends only on $n$, not on $k$. \ And that, in turn, will let us show
that assuming the ETH, any $\mathsf{AM}\left(  k\right)  $\ protocol for
\textsc{3Sat}\ must use\ $\Omega(\sqrt{n})$\ total bits of communication,
regardless of $k$. \ (Our first algorithm only implies a lower bound of
$k+\Omega(\sqrt{n}/k)=\Omega(n^{1/4})$\ on the total communication, assuming
the ETH.) \ The tradeoff is that the running time of the alternative algorithm
depends exponentially on $\varepsilon^{-\Lambda}$\ for some large constant
$\Lambda$, rather than on $\varepsilon^{-2}$.

For (4'), we need to show that the subsampling theorem of Barak et
al.\ \cite{bhhs}\ continues to give us what we want, so long as
$k=\operatorname*{poly}\left(  n\right)  $. \ This boils down to proving a
good \textit{subsampling theorem for }$k$\textit{-player free games}. \ That
is, given any $k$-player free game $G=\left(  Y_{1},\ldots,Y_{k},B_{1}%
,\ldots,B_{k},V\right)  $ of total size $n=\left\vert Y_{1}\right\vert
\left\vert B_{1}\right\vert \cdots\left\vert Y_{k}\right\vert \left\vert
B_{k}\right\vert $, we need to show that its value $\omega\left(  G\right)
$\ can be approximated to within additive error $\pm\varepsilon$, by
restricting attention to random subsets of questions $\left(  S_{i}\subset
Y_{i}\right)  _{i\in\left[  k\right]  }$, where each $S_{i}$ has size
$\varepsilon^{-O\left(  1\right)  }\log n$. \ A direct adaptation of our
argument from the $k=2$\ case turns out not to work here (it breaks down when
$k$ is greater than $O\left(  \log n\right)  $), but we give an alternative
encoding of $k$-player free games by $k$-CSPs that works for all
$k=\operatorname*{poly}\left(  n\right)  $.

\section{Quantum Motivation\label{QUANTUM}}

In studying $\mathsf{AM}\left(  2\right)  $, our original motivation was to
understand the quantum complexity class $\mathsf{QMA}\left(  2\right)  $
(i.e., two-prover Quantum Merlin-Arthur). \ So in this section, we provide
some background about $\mathsf{QMA}\left(  2\right)  $, and explain the
tantalizingly close analogy between it and $\mathsf{AM}\left(  2\right)  $.
\ Readers who don't care about quantum complexity theory can skip this section.

Recall that \textquotedblleft ordinary\textquotedblright\ $\mathsf{QMA}$ is
just the quantum analogue of $\mathsf{MA}$:

\begin{definition}
[Quantum Merlin-Arthur]$\mathsf{QMA}$\ is the class of languages
$L\subseteq\left\{  0,1\right\}  ^{\ast}$\ for which there exists a
polynomial-time quantum algorithm $Q$ such that, for all inputs $x\in\left\{
0,1\right\}  ^{n}$:

\begin{itemize}
\item If $x\in L$\ then there exists a quantum witness state $\left\vert
\phi\right\rangle $, on $\operatorname*{poly}\left(  n\right)  $\ qubits, such
that $Q\left(  x,\left\vert \phi\right\rangle \right)  $\ accepts with
probability at least $2/3$.

\item If $x\notin L$\ then $Q\left(  x,\left\vert \phi\right\rangle \right)
$\ accepts with probability at most $1/3$, for all purported witness states
$\left\vert \phi\right\rangle $.
\end{itemize}
\end{definition}

A lot is known about $\mathsf{QMA}$: for example, it\ has natural complete
promise problems, admits amplification, and is contained in $\mathsf{PP}$ (see
Aharonov and Naveh \cite{an} for a survey).

Now, $\mathsf{QMA}\left(  k\right)  $\ (introduced by Kobayashi et
al.\ \cite{kmy}) is just like $\mathsf{QMA}$, but with $k$ Merlins who are
assumed to be unentangled. \ Note that, if the Merlins were entangled, then
the joint state they sent to Arthur could be arbitrary---so from Arthur's
perspective, there might as well be only one Merlin.\footnote{For precisely
this reason, in the classical case we trivially have $\mathsf{MA}\left(
k\right)  =\mathsf{MA}$\ for all $k=\operatorname*{poly}\left(  n\right)  $.}
\ With $\mathsf{QMA}\left(  k\right)  $, the hope is that, ironically, Arthur
can exploit his knowledge that the messages are \textit{un}entangled to verify
statements that he otherwise could not. \ More formally:

\begin{definition}
[$k$-Prover Quantum Merlin-Arthur]$\mathsf{QMA}\left(  k\right)  $ is the
class of languages $L\subseteq\left\{  0,1\right\}  ^{\ast}$\ for which there
exists a polynomial-time quantum algorithm $Q$ such that, for all inputs
$x\in\left\{  0,1\right\}  ^{n}$:

\begin{itemize}
\item If $x\in L$, then there exist quantum witness states $\left\vert
\phi_{1}\right\rangle ,\ldots,\left\vert \phi_{k}\right\rangle $, each on
$\operatorname*{poly}\left(  n\right)  $\ qubits, such that $Q\left(
x,\left\vert \phi_{1}\right\rangle \otimes\cdots\otimes\left\vert \phi
_{k}\right\rangle \right)  $\ accepts with probability at least $2/3$.

\item If $x\notin L$\ then $Q\left(  x,\left\vert \phi_{1}\right\rangle
\otimes\cdots\otimes\left\vert \phi_{k}\right\rangle \right)  $\ accepts with
probability at most $1/3$ for all purported witness states $\left\vert
\phi_{1}\right\rangle ,\ldots,\left\vert \phi_{k}\right\rangle $.
\end{itemize}
\end{definition}

Compared to $\mathsf{QMA}$, strikingly little is known about $\mathsf{QMA}%
\left(  2\right)  $. \ Clearly%
\begin{equation}
\mathsf{QMA}\subseteq\mathsf{QMA}\left(  2\right)  \subseteq\mathsf{NEXP},
\end{equation}
but we do not know any better containments. \ We do not even have strong
evidence that $\mathsf{QMA}\left(  2\right)  \neq\mathsf{QMA}$, or at the
other extreme that $\mathsf{QMA}\left(  2\right)  \neq\mathsf{NEXP}$. \ Harrow
and Montanaro \cite{harrowmontanaro}\ showed that $\mathsf{QMA}\left(
2\right)  $\ allows exponential amplification of success probabilities, and
that $\mathsf{QMA}\left(  2\right)  =\mathsf{QMA}\left(  k\right)  $ for all
$k\geq3$; even these were surprisingly nontrivial results.

Of course, $\mathsf{QMA}\left(  2\right)  $\ would be of limited interest, if
we could never actually \textit{exploit} the promise of unentanglement to do
anything new. \ In 2007, however, Blier and Tapp \cite{bliertapp} gave a
$\mathsf{QMA}\left(  2\right)  $\ protocol for the $\mathsf{NP}$-complete
\textsc{3Coloring} problem, using two quantum witnesses with only $\log n$
qubits each. \ The catch was that Arthur has only a $1/\operatorname*{poly}%
\left(  n\right)  $\ probability of catching the Merlins if they cheat. \ Even
then, however, any one-prover $\mathsf{QMA}$\ protocol with the same
parameters would imply $\mathsf{NP}\subseteq\mathsf{BQP}$.

Independently, Aaronson et al.\ \cite{abdfs} gave a protocol to convince
Arthur that a \textsc{3Sat}\ instance of size $n$ is satisfiable, using
$\widetilde{O}(\sqrt{n})$ quantum witnesses with $\log n$\ qubits each.
\ Unlike Blier and Tapp's protocol, Aaronson et al.'s achieved
\textit{constant} soundness, and that is why it required more\ communication
($\widetilde{O}(\sqrt{n})$ rather than $\log n$). \ Shortly afterward,
Aaronson et al.'s protocol was improved by Harrow and Montanaro
\cite{harrowmontanaro}, who showed how to prove \textsc{3Sat}\ using two
quantum witnesses with $\widetilde{O}(\sqrt{n})$\ qubits each; and in a
different direction by Chen and Drucker \cite{chendrucker}, who showed how to
measure each of the $\widetilde{O}(\sqrt{n})$\ witnesses separately from the
others.\footnote{It is still not known whether one can combine the
Harrow-Montanaro and Chen-Drucker improvements, to get a \textsc{3Sat}%
\ protocol using two witnesses of $\widetilde{O}(\sqrt{n})$\ qubits\ each that
are measured separately from each other.}

Without going into too much detail, all of these $\widetilde{O}(\sqrt{n}%
)$-qubit protocols for \textsc{3Sat}\ ultimately rely on the Birthday Paradox.
\ In particular, they all involve Arthur measuring $k$ quantum registers with
$\log n$\ qubits each---and if we want constant soundness, then (roughly
speaking) we need a constant probability that two or more of Arthur's
measurements will reveal information about the same \textsc{3Sat} variable
$x_{j}$. \ And that is why we need $k=\Omega(\sqrt{n})$.

It is tempting to speculate that $\sqrt{n}$\ qubits represents some sort of
fundamental barrier for multi-prover $\mathsf{QMA}$ protocols: i.e., that
assuming we want constant soundness, we can save a quadratic factor in the
number of qubits needed to prove \textsc{3Sat}, but no more than that.
\ Certainly it would be astonishing if \textsc{3Sat}\ could be proved (with
constant soundness) using two unentangled witnesses with only
$\operatorname*{polylog}n$\ qubits each. \ In that case, \textquotedblleft
scaling up\textquotedblright\ by an exponential, we would presumably get that
$\mathsf{QMA}\left(  2\right)  =\mathsf{NEXP}$.

When one thinks about the above questions---or for that matter, almost
\textit{any} questions about $\mathsf{QMA}\left(  2\right)  $---one is
inevitably led to a computational problem that Harrow and Montanaro
\cite{harrowmontanaro} called the \textsc{Best Separable State} or
\textsc{BSS} problem.

\begin{problem}
[\textsc{BSS}$_{\varepsilon}$]Given as input a Hermitian matrix $A\in
\mathbb{C}^{n^{2}\times n^{2}}$, with eigenvalues in $\left[  0,1\right]  $,
approximate%
\begin{equation}
\lambda_{\operatorname*{sep}}\left(  A\right)  :=\max_{v,w\in\mathbb{C}%
^{n}:\left\Vert v\right\Vert =\left\Vert w\right\Vert =1}(v^{\dag}\otimes
w^{\dag})A\left(  v\otimes w\right)
\end{equation}
to additive error $\pm\varepsilon$. \ (Here $\varepsilon$\ is assumed to be an
arbitrarily small constant if not specified otherwise.)
\end{problem}

To build intuition, note that%
\begin{equation}
\lambda\left(  A\right)  :=\max_{u\in\mathbb{C}^{n^{2}}:\left\Vert
v\right\Vert =1}u^{\dag}Au
\end{equation}
is just the largest eigenvalue of $A$, which is easy to compute. \ Indeed, the
proof of $\mathsf{QMA}\subseteq\mathsf{PP}$\ works by reducing the simulation
of a\ $\mathsf{QMA}$\ protocol to the computation of $\lambda\left(  A\right)
$, for some exponentially-large Hermitian matrix $A$.

By contrast, \textsc{BSS}\ asks us to maximize $u^{\dag}Au$\ \textit{only over
unit vectors of the form} $u=v\otimes w$. \ That is why \textsc{BSS}\ models
the problem of maximizing the verifier's acceptance probability in a
$\mathsf{QMA}\left(  2\right)  $\ protocol, where the maximum is taken over
all separable witnesses, of the form $\left\vert \phi_{1}\right\rangle
\otimes\left\vert \phi_{2}\right\rangle $. \ From this standpoint, the reason
why $\mathsf{QMA}\left(  2\right)  $\ is so much harder to understand than
$\mathsf{QMA}$---but also why $\mathsf{QMA}\left(  2\right)  $ is\ potentially
more powerful---is that (as one can check) \textsc{BSS} is a non-convex
optimization problem, which lacks the clean linear-algebraic structure of
computing $\lambda\left(  A\right)  $.

Indeed, from the protocol of Blier and Tapp \cite{bliertapp} mentioned
earlier, it follows immediately that we can reduce \textsc{3Coloring}\ to the
problem of approximating $\lambda_{\operatorname*{sep}}\left(  A\right)  $ up
to additive error\ $\pm1/\operatorname*{poly}\left(  n\right)  $.
\ Furthermore, since the quantum witnesses in the Blier-Tapp protocol have
only $\log n$ qubits, the resulting matrix $A$ will have size $2^{O\left(
\log n\right)  }=\operatorname*{poly}\left(  n\right)  $. \ Thus:

\begin{theorem}
[Blier and Tapp \cite{bliertapp}]\label{bliertappthm}\textsc{BSS}%
$_{1/\operatorname*{poly}\left(  n\right)  }$ is $\mathsf{NP}$-hard.
\end{theorem}

One wants to know: is \textsc{BSS}$_{\varepsilon}$\ still a hard problem even
for \textit{constant} $\varepsilon$? \ Because it has constant soundness, the
protocol of Harrow and Montanaro \cite{harrowmontanaro}\ (building on Aaronson
et al.\ \cite{abdfs}) lets us reduce \textsc{3Sat}\ to the problem of
approximating $\lambda_{\operatorname*{sep}}\left(  A\right)  $\ up to
constant additive error. \ Now, since the quantum witnesses in the
Harrow-Montanaro protocol have $\widetilde{O}(\sqrt{n})$\ qubits, the
resulting matrix $A$ has size $2^{\widetilde{O}(\sqrt{n})}$, so we do not get
a polynomial-time reduction. \ We do, however, get something:

\begin{theorem}
\label{bsshardthm}If \textsc{BSS}\ is solvable in $t\left(  n\right)  $\ time,
then \textsc{3Sat}\ is solvable in $t(2^{\widetilde{O}(\sqrt{n})})$\ time.
\ So in particular, assuming the Exponential Time Hypothesis, \textsc{BSS}%
\ requires $n^{\Omega\left(  \log n\right)  }$\ deterministic time.
\ (Likewise, assuming the Randomized ETH, \textsc{BSS}\ requires
$n^{\Omega\left(  \log n\right)  }$\ randomized time.)
\end{theorem}

Could we go further than Theorems \ref{bliertappthm}\ and \ref{bsshardthm},
and prove that \textsc{BSS}$_{\varepsilon}$\ is $\mathsf{NP}$-hard\ even for
constant $\varepsilon$? \ Notice that if we could, then \textquotedblleft
scaling up by an exponential,\textquotedblright\ we could presumably also show
$\mathsf{QMA}\left(  2\right)  =\mathsf{NEXP}$! \ If, on the other hand, we
believe (as seems plausible) that $\mathsf{QMA}\left(  2\right)
\subseteq\mathsf{EXP}$, then we seem forced to believe that \textsc{BSS}\ is
solvable in $n^{\operatorname*{polylog}n}$ time, even if we have no idea what
the algorithm is.\footnote{Strictly speaking, neither of these implications is
a theorem. \ For example, even if \textsc{BSS}$_{\varepsilon}$\ turned out to
be $\mathsf{NP}$-hard\ for constant $\varepsilon$, it's possible that one
could exploit the special structure of the matrices arising from
polynomial-size quantum circuits to show that $\mathsf{QMA}\left(  2\right)
\subseteq\mathsf{EXP}$. \ In practice, however, a \textquotedblleft
reasonable\textquotedblright\ proof that \textsc{BSS}$_{\varepsilon}$\ is
$\mathsf{NP}$-hard\ would probably also imply $\mathsf{QMA}\left(  2\right)
=\mathsf{NEXP}$, and a \textquotedblleft reasonable\textquotedblright\ proof
of $\mathsf{QMA}\left(  2\right)  \subseteq\mathsf{EXP}$\ would probably
proceed by solving \textsc{BSS}$_{\varepsilon}$\ in quasipolynomial time.}

Raising the stakes even further, Barak et al.\ \cite{bbhksz} showed that
\textsc{BSS}\ is intimately related to other problems of great current
interest in complexity theory: namely, the \textsc{Unique Games},
\textsc{Small Set Expansion}, and \textsc{2-to-4 Norm} problems.\footnote{We
refer the reader to \cite{bbhksz}\ for the definitions of these problems, but
very briefly: \textsc{Unique Games}\ (\textsc{UG}) is the problem of deciding
whether $\omega(G)$\ is close to $1$ or close to $0$, given as input a
description of a two-prover game $G=\left(  X,Y,A,B,\mathcal{D},V\right)  $
with the special properties that $\left\vert A\right\vert =\left\vert
B\right\vert $, and that for every $\left(  x,y\right)  \in X\times Y$\ there
exists a permutation $\pi_{x,y}$\ such that $V(x,y,a,b)=1$\ if and only if
$b=\pi_{x,y}(a)$. \ \textsc{Small Set Expansion}\ (\textsc{SSE}) is the
problem of deciding whether a given graph $G$ is close to or far from an
expander graph, if we consider $G$'s expansion on \textquotedblleft
small\textquotedblright\ subsets of vertices only.\ \ \textsc{2-to-4
Norm}\ (\textsc{2-to-4}) is the problem, given as input an $n\times n$
matrix\ $A$, of approximating the maximum of $\left\Vert Av\right\Vert _{4}%
$\ over all vectors $v$\ such that $\left\Vert v\right\Vert _{2}=1$.} \ The
lattice of known reductions among these problems is as follows:%
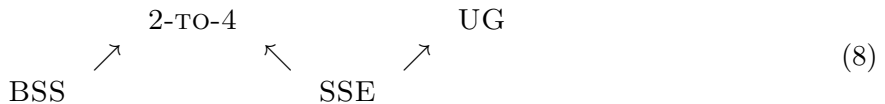
\begin{equation}%
\begin{array}
[c]{ccccccc}
&  & \text{\textsc{2-to-4}} &  &  &  & \text{\textsc{UG}}\\
& \nearrow &  & \nwarrow &  & \nearrow & \\
\text{\textsc{BSS}} &  &  &  & \text{\textsc{SSE}} &  &
\end{array}
\end{equation}

Now, assuming the ETH, Theorem \ref{bsshardthm} gives us $n^{\Omega\left(
\log n\right)  }$\ hardness for \textsc{BSS}---and as a direct consequence,
for \textsc{2-to-4 Norm}\ as well. \ That might not sound like much, but it's
a lot more than we currently know for either \textsc{Unique Games}\ or
\textsc{Small Set Expansion}! \ So the speculation arises that, if we fully
understood \textsc{BSS}, we might be able to apply some of the insights to
\textsc{UG}\ or \textsc{SSE}.

To lay our cards on the table, here is our conjecture about \textsc{BSS}:

\begin{conjecture}
\label{bssconj}\textsc{BSS}$_{\varepsilon}$\ is solvable in deterministic
$n^{O(\varepsilon^{-2}\log n)}$ time.
\end{conjecture}

If true, Conjecture \ref{bssconj}\ readily implies that $\mathsf{QMA}\left(
2\right)  \subseteq\mathsf{EXP}$. \ Since a $t\left(  n\right)  $%
-time\ algorithm for \textsc{BSS}\ can be combined with a $q\left(  n\right)
$-qubit $\mathsf{QMA}\left(  2\right)  $\ protocol for \textsc{3Sat}\ to get a
$t(2^{O\left(  q\left(  n\right)  \right)  })$-time\ algorithm for
\textsc{3Sat}, Conjecture \ref{bssconj} also implies that, assuming the ETH,
any $\mathsf{QMA}\left(  2\right)  $\ protocol for \textsc{3Sat}\ must use
$\Omega(\sqrt{n})$\ qubits.

There has been some progress toward a proof of Conjecture \ref{bssconj}. \ In
particular, Brandao, Christandl, and Yard \cite{bcy}\ gave an algorithm that
solves \textsc{BSS}$_{\varepsilon}$\ in $n^{O(\varepsilon^{-2}\log n)}$ time
if $\left\Vert A\right\Vert _{2}=O\left(  1\right)  $, or alternatively, if
$A$ represents a quantum measurement that can be implemented using LOCC (Local
Operations and Classical Communication). \ This implied, among other things,
that $\mathsf{QMA}_{\mathsf{LOCC}}\left(  k\right)  =\mathsf{QMA}$ for
$k=O\left(  1\right)  $, where $\mathsf{QMA}_{\mathsf{LOCC}}\left(  k\right)
$\ is the subclass of $\mathsf{QMA}\left(  k\right)  $\ in which Arthur is
restricted to LOCC measurements. \ Brandao et al.'s algorithm uses a technique
that quantum information researchers know as \textit{symmetric extension}, and
that theoretical computer scientists know as the \textit{Lasserre hierarchy}.
\ It is not known whether similar techniques could work for arbitrary
$\mathsf{QMA}\left(  k\right)  $\ protocols.

More recently, Brandao and\ Harrow \cite{brandaoharrow:freegame}\ showed that,
assuming the ETH, any so-called $\mathsf{BellQMA}\left(  k\right)  $ protocol
for \textsc{3Sat}---that is, any $\mathsf{QMA}\left(  k\right)  $\ protocol
where each of the $k$ witnesses are measured separately---must use
$n^{1/2-o\left(  1\right)  }$\ qubits. \ This lower bound is known to be
essentially tight, due to the protocol of Chen and Drucker \cite{chendrucker}.
\ The requirement that each witness be measured separately (with the
measurement outcomes then combined with classical postprocessing) is even more
stringent than the requirement of LOCC. \ Despite this, the result of Brandao
and\ Harrow \cite{brandaoharrow:freegame} did not follow from the earlier
result of Brandao, Christandl, and Yard \cite{bcy} that $\mathsf{QMA}%
_{\mathsf{LOCC}}\left(  k\right)  =\mathsf{QMA}$, because the latter works
only for constant $k$.

\subsection{Connection to Our Results\label{CONNECT}}

But what does any of the above have to do with $\mathsf{AM}\left(  2\right)
$? \ One way to view this paper's contribution is as follows: \textit{we prove
that a \textquotedblleft classical analogue\textquotedblright\ of Conjecture
\ref{bssconj}\ holds}. \ In more detail, we can think of $\mathsf{AM}\left(
2\right)  $\ as closely analogous in many ways to $\mathsf{QMA}\left(
2\right)  $. \ For both classes, the only obvious lower bound comes from
restricting to a single Merlin,\ while the only obvious\ upper bound is
$\mathsf{NEXP}$. \ For both classes, the difficulty with proving an
$\mathsf{EXP}$\ upper bound\ is the requirement that the Merlins can't
communicate, which gives rise to a non-convex optimization problem. \ For both
classes, there exists a protocol for \textsc{3Sat}\ that uses $\log
n$\ communication, but that has only a\ $1/\operatorname*{poly}\left(
n\right)  $\ probability of catching cheating Merlins. \ For both classes, we
can improve the \textsc{3Sat}\ protocol to have constant soundness, by using a
strong PCP theorem together with the Birthday Paradox---but if we do so, then
the communication cost increases from $\log n$\ to $\widetilde{O}(\sqrt{n})$.

Because the analogy runs so deep, it seems of interest to $\mathsf{QMA}\left(
2\right)  $\ researchers to know that:

\begin{enumerate}
\item[(1)] \textsc{FreeGame}$_{\varepsilon}$ \textit{is} solvable in
$n^{O(\varepsilon^{-2}\log n)}$\ time, as we conjecture that \textsc{BSS}%
$_{\varepsilon}$\ is.

\item[(2)] $\mathsf{AM}\left(  2\right)  $ \textit{is} contained in
$\mathsf{EXP}$, as we conjecture that $\mathsf{QMA}\left(  2\right)  $\ is.

\item[(3)] The $\widetilde{O}(\sqrt{n})$-communication $\mathsf{AM}\left(
2\right)  $\ protocol for \textsc{3Sat}\ \textit{is} essentially optimal
assuming the ETH, as we conjecture that the corresponding $\mathsf{QMA}\left(
2\right)  $\ protocol is.
\end{enumerate}

Of course, we also show in this paper that $\mathsf{AM}\left(  2\right)
=\mathsf{AM}$. \ So pushing the analogy between $\mathsf{AM}\left(  2\right)
$\ and $\mathsf{QMA}\left(  2\right)  $ all the way to the end\ would lead to
the conjecture that $\mathsf{QMA}\left(  2\right)  =\mathsf{QMA}$. \ We remain
agnostic about whether the analogy extends \textit{that} far!

\section{Preliminaries\label{PRELIM}}

Some notation: we use $\operatorname{E}$\ for expectation, $\left[  n\right]
$\ for $\left\{  1,\ldots,n\right\}  $,\ and\ $\binom{\left[  n\right]  }{k}$
for the set of subsets of $\left[  n\right]  $\ of size $k$. \ In addition to
the notation $\widetilde{O}(f\left(  n\right)  )$\ for $O(f\left(  n\right)
\operatorname*{polylog}f\left(  n\right)  )$, we also use $\widetilde{\Omega
}(f\left(  n\right)  )$\ for $\Omega(f\left(  n\right)
/\operatorname*{polylog}f\left(  n\right)  )$. \ All logs are base $2$ unless
specified otherwise.

Sections \ref{INTRO}\ and \ref{RESULTS} have already defined many of the
concepts we will need, including two-prover games, free games, the
clause/variable game, the birthday repetition, and the \textsc{FreeGame}%
\ problem. \ For completeness, though, we now give the general definition of
$k$-player free games.

\begin{definition}
[$k$-Player Free Games]A $k$-player free game $G$ consists of:

\begin{enumerate}
\item[(1)] finite question sets\ $Y_{1},\ldots,Y_{k}$ and answer
sets\ $B_{1},\ldots,B_{k}$, and

\item[(2)] a verification function\ $V:Y_{1}\times\cdots\times Y_{k}\times
B_{1}\times\cdots\times B_{k}\rightarrow\left[  0,1\right]  $.
\end{enumerate}

The \textit{value} of the game, denoted $\omega\left(  G\right)  $, is the
maximum, over all tuples of response functions\ $\left(  b_{i}:Y_{i}%
\rightarrow B_{i}\right)  _{i\in\left[  k\right]  }$ , of%
\begin{equation}
\operatorname*{E}_{y_{1}\in Y_{1},\ldots,y_{k}\in Y_{k}}\left[  V\left(
y_{1},\ldots,y_{k},b_{1}\left(  y_{1}\right)  ,\ldots,b_{k}\left(
y_{k}\right)  \right)  \right]  .
\end{equation}

\end{definition}

Directly related to $k$-player free games is the complexity class
$\mathsf{AM}\left(  k\right)  $, which we now formally
define.\footnote{$\mathsf{AM}\left(  k\right)  $ should not be confused with
$\mathsf{AM}\left[  k\right]  $, which means $\mathsf{AM}$\ with a single
Merlin but $k$ \textit{rounds} of communication. \ A classic result of Babai
and Moran \cite{babaimoran}\ says that $\mathsf{AM}\left[  k\right]
=\mathsf{AM}\left[  2\right]  =\mathsf{AM}$\ for all constants $k\geq2$.
\ When $k=\operatorname*{poly}\left(  n\right)  $, by contrast, such a
collapse is not believed to happen, since $\mathsf{AM}\left[
\operatorname*{poly}\right]  =\mathsf{IP}=\mathsf{PSPACE}$.}

\begin{definition}
[$k$-Prover Arthur-Merlin]\label{amkdef}Let $k$\ be a positive integer. \ Then
$\mathsf{AM}\left(  k\right)  $ is the class of languages $L\subseteq\left\{
0,1\right\}  ^{\ast}$\ for which there exists a probabilistic polynomial-time
verifier $V$ such that for all $n$\ and all inputs $x\in\left\{  0,1\right\}
^{n}$:

\begin{itemize}
\item \textbf{(Completeness)} If $x\in L$, then there exist functions
$b_{1},\ldots,b_{k}:\left\{  0,1\right\}  ^{\operatorname*{poly}\left(
n\right)  }\rightarrow\left\{  0,1\right\}  ^{\operatorname*{poly}\left(
n\right)  }$, depending on $x$, such that%
\begin{equation}
\Pr_{y_{1},\ldots,y_{k}\in_{R}\left\{  0,1\right\}  ^{\operatorname*{poly}%
\left(  n\right)  }}\left[  V\left(  x,y_{1},\ldots,y_{k},b_{1}\left(
y_{1}\right)  ,\ldots,b_{k}\left(  y_{k}\right)  \right)  ~\text{accepts}%
\right]  \geq\frac{2}{3}.
\end{equation}

\item \textbf{(Soundness)} If $x\notin L$, then for all such functions
$b_{1},\ldots,b_{k}$,%
\begin{equation}
\Pr_{y_{1},\ldots,y_{k}\in_{R}\left\{  0,1\right\}  ^{\operatorname*{poly}%
\left(  n\right)  }}\left[  V\left(  x,y_{1},\ldots,y_{k},b_{1}\left(
y_{1}\right)  ,\ldots,b_{k}\left(  y_{k}\right)  \right)  ~\text{accepts}%
\right]  \leq\frac{1}{3}.
\end{equation}

\end{itemize}
\end{definition}

Clearly $\mathsf{AM}\left(  1\right)  =\mathsf{AM}$ and $\mathsf{AM}\left(
k\right)  \subseteq\mathsf{AM}\left(  k+1\right)  $ for all $k$. \ We also
have $\mathsf{AM}\left(  k\right)  \subseteq\mathsf{MIP}\left(  k\right)  $,
thereby giving the crude upper bound $\mathsf{AM}\left(  k\right)
\subseteq\mathsf{NEXP}$ (later we will do much better).

We can easily generalize the definition of $\mathsf{AM}\left(  k\right)  $\ to
$\mathsf{AM}\left(  k\left(  n\right)  \right)  $, for any growth rate
$k\left(  n\right)  =O\left(  \operatorname*{poly}\left(  n\right)  \right)
$. \ Also, let $\mathsf{AM}_{p\left(  n\right)  }\left(  k\right)  $ be the
variant of $\mathsf{AM}\left(  k\right)  $ where all messages (both the
$y_{i}$'s and the $b_{i}$'s) are constrained to be $p\left(  n\right)  $ bits long.

Given any probabilistic complexity class $\mathcal{C}$, one of the first
questions we can ask is whether $\mathcal{C}$\ admits amplification of success
probabilities---or equivalently, whether $\mathcal{C}$\ is robust under
changing its error parameters (such as $1/3$\ and $2/3$). \ At least for
$\mathsf{AM}\left(  2\right)  $, we are fortunate that a positive answer
follows from known results. \ In particular, building on the Parallel
Repetition Theorem (Theorem \ref{prt}), Rao \cite{rao}\ proved a useful
concentration bound for the parallel repetitions of two-prover games:

\begin{theorem}
[Rao's Concentration Theorem \cite{rao}]\label{raothm}For all $\delta>0$\ and
all two-prover games $G=\left(  X,Y,A,B,\mathcal{D},V\right)  $, if
Merlin$_{1}$\ and Merlin$_{2}$ play the parallel repeated version $G^{N}$,
then they can win more than a $\omega\left(  G\right)  +\delta$\ fraction of
the games with probability at most%
\begin{equation}
2\left(  1-\frac{\delta/2}{\omega\left(  G\right)  +3\delta/4}\right)
^{\Omega\left(  \frac{\delta^{2}N}{\log\left\vert A\right\vert \left\vert
B\right\vert -\log\left(  \omega\left(  G\right)  +\delta/4\right)  }\right)
}%
\end{equation}

\end{theorem}

Theorem \ref{raothm} implies that \textquotedblleft
amplification\ works\textquotedblright\ for $\mathsf{AM}\left(  2\right)  $ protocols:

\begin{proposition}
\label{ampcor}In the definition of $\mathsf{AM}\left(  2\right)  $, replacing
the constants $\left(  1/3,2/3\right)  $\ by $\left(  a,b\right)  $\ for any
constants $0<a<b<1$, or indeed by $\left(  2^{-p\left(  n\right)
},1-2^{-p\left(  n\right)  }\right)  $\ or $\left(  1/2-1/p\left(  n\right)
,1/2+1/p\left(  n\right)  \right)  $\ for any polynomial $p$, gives rise to
the same complexity class.
\end{proposition}

\begin{proof}
Suppose, for example, that we want to amplify $\left(  1/3,2/3\right)  $\ to
$\left(  2^{-p\left(  n\right)  },1-2^{-p\left(  n\right)  }\right)  $; the
other cases are analogous. \ Given a language $L\in\mathsf{AM}\left(
2\right)  $ and an input $x\in\left\{  0,1\right\}  ^{n}$, the $\mathsf{AM}%
\left(  2\right)  $\ protocol for checking whether $x\in L$\ can be
represented as a free game $G=\left(  X,Y,A,B,V\right)  $, where
$X=Y=A=B=\left\{  0,1\right\}  ^{q\left(  n\right)  }$ for some polynomial
$q$. \ We have $\omega\left(  G\right)  \geq2/3$\ if $x\in L$, and
$\omega\left(  G\right)  \leq1/3$\ if $x\notin L$. \ Now let $G_{1/2}^{N}$\ be
the game where the Merlins\ play $N$ parallel instances of the original game
$G$, and they \textquotedblleft win\textquotedblright\ if and only if they win
on at least $N/2$\ instances. \ If $\omega\left(  G\right)  \geq2/3$, then
clearly $\omega(G_{1/2}^{N})\geq1-2^{-\Omega\left(  N\right)  }$---since if
the Merlins just play their optimal strategy for $G$ on each of the $N$
instances separately, then the number that they win will be concentrated
around $\omega\left(  G\right)  N\geq2N/3$\ by a standard Chernoff bound. \ On
the other hand, if $\omega\left(  G\right)  \leq1/3$, then Theorem
\ref{raothm}\ implies that $\omega(G_{1/2}^{N})\leq2^{-\Omega\left(
N/q\left(  n\right)  \right)  }$. \ So, by simply choosing $N\gg p\left(
n\right)  q\left(  n\right)  $ to be a suitably large polynomial, we can
ensure that $\omega(G_{1/2}^{N})\geq1-2^{-p\left(  n\right)  }$\ if $x\in
L$\ while $\omega(G_{1/2}^{N})\leq2^{-p\left(  n\right)  }$\ if $x\notin L$.
\end{proof}

Note that Proposition \ref{ampcor} can blow up the communication cost by a
polynomial factor, because of the dependence of $N$ on $q\left(  n\right)  $
(which derives from the $1/\log\left\vert A\right\vert \left\vert B\right\vert
$\ factor in the exponent from Theorem \ref{raothm}). \ For this reason,
Proposition \ref{ampcor} doesn't directly imply any useful amplification for
our $\widetilde{O}(\sqrt{n})$-communication protocol for \textsc{3Sat}. \ See
Section \ref{LOWERROR}\ for further discussion of this issue, and for our best
current results for the low-error case.

Rao (personal communication) believes that it would be straightforward to
generalize Theorem \ref{raothm}\ to games with $k\geq3$\ players, as long as
the games are free.\footnote{By contrast, for \textit{general} games with
$k\geq3$\ players, even proving a \textquotedblleft standard\textquotedblright%
\ parallel repetition theorem is a notorious open problem.} \ If so, then we
would also obtain an amplification theorem for $\mathsf{AM}\left(  k\right)
$, for all $k=\operatorname*{poly}\left(  n\right)  $. \ However, this
generalization has not yet been worked out explicitly.

One last remark: a classic result\ about \textquotedblleft
ordinary\textquotedblright\ $\mathsf{AM}$\ (see \cite{fgmsz}) states that any
$\mathsf{AM}$\ protocol can be made to have \textit{perfect completeness}.
\ In other words, the condition $x\in L\Rightarrow\Pr\left[  V\text{
accepts}\right]  \geq2/3$\ can be strengthened to $x\in L\Rightarrow\Pr\left[
V\text{ accepts}\right]  =1$ without loss of generality. \ Another classic
result \cite{gs} states that any $\mathsf{AM}$\ protocol can be made
\textit{public-coin}, meaning that any random bits generated by Arthur are
immediately shared with Merlin. \ In terms of games, the public-coin property
would imply in particular that Arthur's verification function was
deterministic: that is, $V\left(  x,y,a,b\right)  \in\left\{  0,1\right\}
$\ for all $x,y,a,b$.

Thus, one might wonder whether any $\mathsf{AM}\left(  k\right)  $ protocol
can be made perfect-completeness and public-coin as well. \ Ultimately,
affirmative answers to these questions will follow from our result that
$\mathsf{AM}\left(  k\right)  =\mathsf{AM}$, which works regardless of whether
the original $\mathsf{AM}\left(  k\right)  $\ protocol had perfect
completeness or was public-coin. \ But it would be interesting to find
\textit{direct} proofs of these properties for $\mathsf{AM}\left(  k\right)
$. \ (It would also be interesting to find a direct proof that $\mathsf{AM}%
\left(  k\right)  =\mathsf{AM}\left(  2\right)  $ for all $k>2$, rather
than\ deducing this as a consequence of $\mathsf{AM}\left(  k\right)
=\mathsf{AM}$.)

\section{Analysis of the Birthday Game\label{3SAT}}

Our goal in this section is to prove Theorem \ref{mainthm1}: informally, that
$\mathsf{AM}\left(  2\right)  $\ protocols for \textsc{3Sat}\ can achieve
nearly a quadratic savings in communication over the \textquotedblleft
na\"{\i}ve\textquotedblright\ bound of $n$ bits. \ The section is organized as
follows. \ First, in Section \ref{BASIC3SAT}, we give a \textquotedblleft
basic\textquotedblright\ protocol\ with a $1$\ vs.\ $1-\epsilon$%
\ completeness/soundness gap (for some fixed $\epsilon>0$) and $\widetilde{O}%
(\sqrt{n})$\ communication cost. \ The protocol is based on the birthday
repetition already discussed in Section \ref{3SATINT}; for concreteness, we
initially implement the idea using Dinur's PCP Theorem and the clause/variable
game. \ Next, in Section \ref{HIGHERROR}, we study the high-error case,
showing how a more refined analysis leads to a protocol with a $1$%
\ vs.\ $1-\varepsilon$\ completeness/soundness gap and $O(\sqrt{\varepsilon
n}\operatorname*{polylog}n)$\ communication cost. \ Then in Section
\ref{LOWERROR}, we switch to the low-error case, using the PCP Theorem of
Moshkovitz and Raz \cite{mr}\ to obtain an $\mathsf{AM}\left(  2\right)
$\ protocol for \textsc{3Sat}\ with a $1$\ vs.\ $\delta$%
\ completeness/soundness gap and $n^{1/2+o\left(  1\right)  }%
\operatorname*{poly}\left(  1/\delta\right)  $\ communication cost. \ Finally,
in Section \ref{CC}, we give the implication $\mathsf{NTIME}\left[  n\right]
\subseteq\mathsf{AM}_{n^{1/2+o\left(  1\right)  }}\left(  2\right)  $\ and
show that this implication is nonrelativizing.

\subsection{The Basic Result\label{BASIC3SAT}}

The first step is to state a variant of the PCP Theorem that is strong enough
for our purposes.

\begin{theorem}
[PCP Theorem, Dinur's Version \cite{dinur}]\label{pcpthm}Given a
\textsc{3Sat}\ instance $\varphi$\ of size $n$, it is possible in
$\operatorname*{poly}\left(  n\right)  $\ time to produce a new \textsc{3Sat}%
\ instance $\phi$, of size $n\operatorname*{polylog}n$, such that:

\begin{itemize}
\item \textbf{(Completeness)} If $\operatorname*{SAT}\left(  \varphi\right)
=1$\ then $\operatorname*{SAT}\left(  \phi\right)  =1$.

\item \textbf{(Soundness)} If $\operatorname*{SAT}\left(  \varphi\right)
<1$\ then $\operatorname*{SAT}\left(  \phi\right)  <1-\epsilon$, for some
constant $0<\epsilon<1/8$.

\item \textbf{(Balance)} Every clause of $\phi$\ involves exactly $3$
variables, and every variable of $\phi$\ appears in exactly $d$ clauses, for
some constant $d$.\footnote{It is known that we can assume this balance
condition without loss of generality.}
\end{itemize}
\end{theorem}

The reason why, for now, we use Dinur's version of the PCP Theorem is that it
produces instances of size $n\operatorname*{polylog}n$. \ Later, in Section
\ref{LOWERROR}, we will switch over to the PCP Theorem of Moshkovitz and Raz
\cite{mr}, which produces instances of the slightly larger size $n\cdot
2^{\left(  \log n\right)  ^{1-\Delta}}=n^{1+o\left(  1\right)  }$\ (for some
constant $\Delta>0$) but achieves sub-constant error. \ Were we willing to
accept a protocol with $\sqrt{n}2^{\left(  \log n\right)  ^{1-\Delta}}%
$\ communication, we could have used the Moshkovitz-Raz version from the
start, but we will try to keep the communication cost down to $\sqrt
{n}\operatorname*{polylog}n$\ for as long as we can.

Let the \textsc{3Sat}\ instance $\phi$\ produced by Theorem \ref{pcpthm}\ have
$N$ variables $x_{1},\ldots,x_{N}$\ and $M$ clauses $C_{1},\ldots,C_{M}$.
\ Also, let $G_{\phi}$\ be the clause/variable game for $\phi$, as defined in
Section \ref{3SATINT}. \ Then combining Theorem \ref{pcpthm} with Proposition
\ref{cvsound}\ yields the following corollary.

\begin{corollary}
\label{phicor}If $\phi$\ is unsatisfiable, then $\omega\left(  G_{\phi
}\right)  <1-\epsilon/3$.
\end{corollary}

Next, given positive integers $k$\ and $\ell$, let $G_{\phi}^{k\times\ell}%
$\ be the birthday repetition of $\phi$, also defined in Section
\ref{3SATINT}. \ Then to prove Theorem \ref{mainthm1}, it suffices to show
that $\omega(G_{\phi}^{k\times\ell})$ is bounded away from $1$, assuming that
$\phi$\ is unsatisfiable and that $k\ell=\Omega\left(  N\right)  $.

Our strategy for upper-bounding $\omega(G_{\phi}^{k\times\ell})$\ will be to
relate it to $\omega\left(  G_{\phi}\right)  $, which we already know is
bounded away from $1$. \ More concretely:

\begin{theorem}
\label{gbdthm}For all $k\in\left[  M\right]  $ and $\ell\in\left[  N\right]
$,%
\begin{equation}
\omega\left(  G_{\phi}\right)  \geq\omega(G_{\phi}^{k\times\ell})-O\left(
\sqrt{\frac{N}{k\ell}}\right)  .
\end{equation}

\end{theorem}

So in particular, by choosing $k=\ell=c\sqrt{N}$, where $c$\ is some
sufficiently large constant, we can ensure (say) $\omega(G_{\phi}^{k\times
\ell})\leq\omega\left(  G_{\phi}\right)  +0.01$.

Let $\mathcal{U}$\ be the uniform distribution over all input pairs%
\begin{equation}
\left(  I,J\right)  \in\binom{\left[  M\right]  }{k}\times\binom{\left[
N\right]  }{\ell},
\end{equation}
and let $V_{BD}$\ be Arthur's verification function in $G_{\phi}^{k\times\ell
}$. \ To prove Theorem \ref{gbdthm}, we consider an arbitrary cheating
strategy for Merlin$_{1}$\ and Merlin$_{2}$\ in the birthday game:%
\begin{equation}
a:\binom{\left[  M\right]  }{k}\rightarrow\left(  \left\{  0,1\right\}
^{3}\right)  ^{k},~~~~~b:\binom{\left[  N\right]  }{\ell}\rightarrow\left\{
0,1\right\}  ^{\ell}.
\end{equation}
Let $p$ be the success probability of that cheating strategy: that is,%
\begin{equation}
p=\operatorname*{E}_{\left(  I,J\right)  \sim\mathcal{U}}\left[  V_{BD}\left(
I,J,a\left(  I\right)  ,b\left(  J\right)  \right)  \right]  .
\end{equation}

Using $a$ and $b$, our task is to construct a cheating strategy for the
original clause/variable game $G_{\phi}$, which succeeds with probability at
least $p-O(\sqrt{N/k\ell})$. \ That strategy will be the \textquotedblleft
natural\textquotedblright\ one: namely, given as input a clause index
$i\in\left[  M\right]  $, Merlin$_{1}$ first chooses a subset $\left\{
i_{1},\ldots,i_{k-1}\right\}  $ uniformly at random from $\binom{\left[
M\right]  -\left\{  i\right\}  }{k-1}$, and sets $I:=\left\{  i,i_{1}%
,\ldots,i_{k-1}\right\}  $. \ (Crucially, $I$ is a set, so if its elements
were listed in some canonical way---for example, in order---$i$ would
generally be somewhere in the middle, and would not be particularly
conspicuous!) \ Merlin$_{1}$ then computes $a\left(  I\right)  \in\left(
\left\{  0,1\right\}  ^{3}\right)  ^{k}$, and sends Arthur the restriction of
$a\left(  I\right)  $\ to the index $i$. \ Likewise, given as input a variable
index $j\in\left[  N\right]  $, Merlin$_{2}$ first chooses a subset $\left\{
j_{1},\ldots,j_{\ell-1}\right\}  $ uniformly at random from $\binom{\left[
N\right]  -\left\{  j\right\}  }{\ell-1}$, and sets $J:=\left\{
j,j_{1},\ldots,j_{\ell-1}\right\}  $. \ He then computes $b\left(  J\right)
\in\left\{  0,1\right\}  ^{\ell}$, and sends Arthur the restriction of
$b\left(  J\right)  $\ to the index $j$. \ Of course, the resulting strategy
is randomized, but we can convert it to an equally-good deterministic strategy
using convexity.

Let $\mathcal{D}$\ be the probability distribution over $\left(  I,J\right)
$\ pairs induced by the cheating strategy above, if we average over all valid
inputs $\left(  i,j\right)  $\ to the original clause/variable game. \ Then
let $q$ be the Merlins' success probability in the birthday game, if they use
their same cheating strategy $\left(  a,b\right)  $, but for $\left(
I,J\right)  $\ pairs drawn from $\mathcal{D}$, rather than from the uniform
distribution $\mathcal{U}$:%
\begin{equation}
q=\operatorname*{E}_{\left(  I,J\right)  \sim\mathcal{D}}\left[  V_{BD}\left(
I,J,a\left(  I\right)  ,b\left(  J\right)  \right)  \right]  .
\end{equation}
Clearly the Merlins' success probability in the clause/variable game is at
least\ $q$, since any time they win $G_{\phi}^{k\times\ell}$, they also win
its restriction to $G_{\phi}$. \ Therefore, to prove Theorem \ref{gbdthm}, it
suffices to prove that $q\geq p-O(\sqrt{N/k\ell})$. \ And to do \textit{that},
it in turn suffices to show that $\mathcal{D}$ is close in variation
distance\ to the uniform distribution $\mathcal{U}$, since%
\begin{equation}
\left\vert \operatorname*{E}_{\mathcal{D}}\left[  Z\right]  -\operatorname*{E}%
_{\mathcal{U}}\left[  Z\right]  \right\vert \leq\left\Vert \mathcal{D}%
-\mathcal{U}\right\Vert
\end{equation}
for any $\left[  0,1\right]  $\ random variable $Z$. \ We upper-bound
$\left\Vert \mathcal{D}-\mathcal{U}\right\Vert $ in the following lemma.

\begin{lemma}
\label{vardist}$\left\Vert \mathcal{D}-\mathcal{U}\right\Vert =O\left(
\sqrt{\frac{N}{k\ell}}\right)  .$
\end{lemma}

\begin{proof}
Let $A=\left(  a_{ij}\right)  \in\left\{  0,1\right\}  ^{M\times N}$\ be the
incidence matrix of the \textsc{3Sat}\ instance $\phi$. \ That is, set
$a_{ij}:=1$\ if the clause $C_{i}$\ involves the variable $x_{j}$, and
$a_{ij}:=0$\ otherwise. \ Note that, by the balance condition, we must have
$\sum_{ij}a_{ij}=3M=dN$\ (where $d$ is the number of clauses that each
variable appears in), and%
\begin{equation}
\sum_{j\in\left[  N\right]  }a_{ij}=3~~\forall i,~~~~~\sum_{i\in\left[
M\right]  }a_{ij}=d~~\forall j.
\end{equation}
Given any $I\subseteq\left[  M\right]  $\ and $J\subseteq\left[  N\right]  $,
define%
\begin{equation}
S_{IJ}:=\sum_{i\in I,j\in J}a_{ij},
\end{equation}
and observe that%
\begin{equation}
\operatorname*{E}_{\left\vert I\right\vert =k,\left\vert J\right\vert =\ell
}\left[  S_{IJ}\right]  =\frac{3k\ell}{N}.
\end{equation}

In the clause/variable game $G_{\phi}$, Arthur chooses his input $\left(
i,j\right)  \in\left[  M\right]  \times\left[  N\right]  $\ uniformly at
random subject to $a_{ij}=1$. \ Now consider an $\left(  I,J\right)  $\ drawn
from $\mathcal{D}$. \ By symmetry, $\left(  I,J\right)  $\ is equally likely
to have been formed starting from any input $\left(  i,j\right)  \in I\times
J$\ such that $a_{ij}=1$. \ This means that $\Pr_{\mathcal{D}}\left[  \left(
I,J\right)  \right]  $ must simply be proportional to $S_{IJ}$, and
normalization gives us the rest:%
\begin{equation}
\Pr_{\mathcal{D}}\left[  \left(  I,J\right)  \right]  =\Pr_{\mathcal{U}%
}\left[  \left(  I,J\right)  \right]  \cdot\frac{S_{IJ}}{3k\ell/N}.
\end{equation}
Thus,%
\begin{align}
\left\Vert \mathcal{D}-\mathcal{U}\right\Vert  &  =\frac{1}{2}\sum_{\left\vert
I\right\vert =k,\left\vert J\right\vert =\ell}\left\vert \Pr_{\mathcal{D}%
}\left[  \left(  I,J\right)  \right]  -\Pr_{\mathcal{U}}\left[  \left(
I,J\right)  \right]  \right\vert \\
&  =\frac{1}{2}\operatorname*{E}_{\left\vert I\right\vert =k,\left\vert
J\right\vert =\ell}\left[  \left\vert \frac{S_{IJ}}{3k\ell/N}-1\right\vert
\right] \\
&  \leq\frac{1}{2}\sqrt{\operatorname*{E}_{\left\vert I\right\vert
=k,\left\vert J\right\vert =\ell}\left[  \left(  \frac{S_{IJ}}{3k\ell
/N}-1\right)  ^{2}\right]  }\label{csi}\\
&  =\frac{1}{2}\sqrt{\frac{\operatorname*{E}_{\left\vert I\right\vert
=k,\left\vert J\right\vert =\ell}\left[  S_{IJ}^{2}\right]  }{\left(
3k\ell/N\right)  ^{2}}-1}%
\end{align}
where line (\ref{csi}) used Cauchy-Schwarz. \ Now,%
\begin{align}
\operatorname*{E}_{\left\vert I\right\vert =k,\left\vert J\right\vert =\ell
}\left[  S_{IJ}^{2}\right]   &  =\operatorname*{E}_{\left\vert I\right\vert
=k,\left\vert J\right\vert =\ell}\left[  \sum_{i,i^{\prime}\in I,~~j,j^{\prime
}\in J}a_{ij}a_{i^{\prime}j^{\prime}}\right] \\
&  =\sum_{i,i^{\prime}\in\left[  M\right]  ,~~j,j^{\prime}\in\left[  N\right]
}a_{ij}a_{i^{\prime}j^{\prime}}\Pr_{\left\vert I\right\vert =k,\left\vert
J\right\vert =\ell}\left[  i,i^{\prime}\in I,~~j,j^{\prime}\in J\right]  .
\end{align}
Here it is convenient to divide the sum into four cases: the case
$i=i^{\prime}$\ and $j=j^{\prime}$, the case $i=i^{\prime}$\ but $j\neq
j^{\prime}$, the case $j=j^{\prime}$\ but $i\neq i^{\prime}$, and the case
$i\neq i^{\prime}$\ and $j=j^{\prime}$. \ These cases give us respectively:%
\begin{align}
\sum_{i\in\left[  M\right]  ,~~j\in\left[  N\right]  }a_{ij}\Pr_{\left\vert
I\right\vert =k,\left\vert J\right\vert =\ell}\left[  i,i^{\prime}\in
I,~~j,j^{\prime}\in J\right]   &  \leq3M\frac{k\ell}{MN},\\
\sum_{i\in\left[  M\right]  ,~~j\neq j^{\prime}\in\left[  N\right]  }%
a_{ij}a_{ij^{\prime}}\Pr_{\left\vert I\right\vert =k,\left\vert J\right\vert
=\ell}\left[  i,i^{\prime}\in I,~~j,j^{\prime}\in J\right]   &  \leq
6M\frac{k\ell\left(  \ell-1\right)  }{MN\left(  N-1\right)  },\\
\sum_{i\neq i^{\prime}\in\left[  M\right]  ,~~j\in\left[  N\right]  }%
a_{ij}a_{i^{\prime}j}\Pr_{\left\vert I\right\vert =k,\left\vert J\right\vert
=\ell}\left[  i,i^{\prime}\in I,~~j,j^{\prime}\in J\right]   &  \leq d\left(
d-1\right)  N\frac{k\left(  k-1\right)  \ell}{M\left(  M-1\right)  N},\\
\sum_{i\neq i^{\prime}\in\left[  M\right]  ,~~j\neq j^{\prime}\in\left[
N\right]  }a_{ij}a_{i^{\prime}j^{\prime}}\Pr_{\left\vert I\right\vert
=k,\left\vert J\right\vert =\ell}\left[  i,i^{\prime}\in I,~~j,j^{\prime}\in
J\right]   &  \leq\left(  3M\right)  ^{2}\frac{k\left(  k-1\right)
\ell\left(  \ell-1\right)  }{M\left(  M-1\right)  N\left(  N-1\right)  }.
\end{align}
Hence%
\begin{align}
&  \operatorname*{E}_{\left\vert I\right\vert =k,\left\vert J\right\vert
=\ell}\left[  S_{IJ}^{2}\right] \\
&  \leq3M\frac{k\ell}{MN}+6M\frac{k\ell\left(  \ell-1\right)  }{MN\left(
N-1\right)  }+d\left(  d-1\right)  N\frac{k\left(  k-1\right)  \ell}{M\left(
M-1\right)  N}+\left(  3M\right)  ^{2}\frac{k\left(  k-1\right)  \ell\left(
\ell-1\right)  }{M\left(  M-1\right)  N\left(  N-1\right)  }\\
&  =\left(  \frac{3k\ell}{N}\right)  ^{2}\left(  O\left(  \frac{N}{k\ell
}\right)  +O\left(  \frac{1}{k}\right)  +O\left(  \frac{1}{\ell}\right)
+1\right) \\
&  =\left(  \frac{3k\ell}{N}\right)  ^{2}\left(  1+O\left(  \frac{N}{k\ell
}\right)  \right)  ,
\end{align}
where we treated $d$ as a constant. \ Therefore%
\begin{equation}
\left\Vert \mathcal{D}-\mathcal{U}\right\Vert \leq\frac{1}{2}\sqrt
{\frac{\operatorname*{E}_{\left\vert I\right\vert =k,\left\vert J\right\vert
=\ell}\left[  S_{IJ}^{2}\right]  }{\left(  3k\ell/N\right)  ^{2}}-1}=O\left(
\sqrt{\frac{N}{k\ell}}\right)  .
\end{equation}

\end{proof}

This completes the proof of Theorem \ref{gbdthm}---showing that if $\phi$\ is
unsatisfiable, then%
\begin{equation}
\omega(G_{\phi}^{k\times\ell})\leq\omega\left(  G_{\phi}\right)  +O\left(
\sqrt{\frac{N}{k\ell}}\right)  \leq1-\Omega\left(  1\right)  ,
\end{equation}
provided we set $k=\ell=c\sqrt{N}$\ for a sufficiently large constant $c$.
\ Theorem \ref{gbdthm}, in turn, gives us the following corollary.

\begin{corollary}
\label{firstcor}There exists an $\mathsf{AM}\left(  2\right)  $\ protocol for
\textsc{3Sat}\ that uses $\widetilde{O}(\sqrt{n})$\ communication, and that
has a $1$\ vs.\ $1-\epsilon$\ completeness/soundness gap for some constant
$\epsilon>0$.
\end{corollary}

\begin{proof}
Given a \textsc{3Sat}\ instance $\varphi$\ of size $n$, we apply Theorem
\ref{pcpthm}\ to get a PCP $\phi$\ of size $N=n\operatorname*{polylog}n$. \ We
then consider the birthday game $G_{\phi}^{k\times k}$, where $k=c\sqrt{N}%
$\ for some large constant $c$. \ Clearly, if $\varphi$\ is satisfiable then
$\omega(G_{\phi}^{k\times k})=1$, while if $\varphi$\ is unsatisfiable then
$\omega(G_{\phi}^{k\times k})\leq1-\epsilon$ for some constant $\epsilon>0$.
\ The only further observation we need to make is that Arthur can apply his
verification function $V_{BD}$\ in time polynomial in $n$.
\end{proof}

Of course, one way to state our $\mathsf{AM}\left(  2\right)  $\ protocol is
as a \textit{reduction}: starting with a \textsc{3Sat}\ instance $\varphi$\ of
size $n$, we produce a free game $G$\ of size $2^{\widetilde{O}\left(
\sqrt{n}\right)  }$ in $2^{\widetilde{O}\left(  \sqrt{n}\right)  }$\ time,
such that $\omega\left(  G\right)  =1$\ if $\varphi$\ is satisfiable and
$\omega\left(  G\right)  \leq1-\epsilon$\ if $\varphi$\ is unsatisfiable.
\ This immediately implies that, assuming the Exponential Time Hypothesis,
there must be some constant $\epsilon>0$\ such that the \textsc{FreeGame}%
$_{\varepsilon}$\ problem requires $n^{\widetilde{\Omega}\left(  \log
n\right)  }$ time for all $\varepsilon\leq\epsilon$.

However, we would like to do better than that, and also understand how the
complexity of \textsc{FreeGame}$_{\varepsilon}$\ depends on the error
$\varepsilon=\varepsilon(n)$. \ Unfortunately, our previous analysis was
deficient in two ways: one that becomes relevant when $\varepsilon$\ is very
small, and another that becomes relevant when $\varepsilon$\ is large. \ The
first deficiency is that, while we showed that the distributions $\mathcal{D}$
and $\mathcal{U}$\ had variation distance $O(\sqrt{N/k\ell})$, that bound
gives nothing if $k,\ell\ll\sqrt{N}$, which is the relevant situation for
small $\varepsilon$. \ And this prevents us from showing that, if
$\varepsilon=o\left(  1\right)  $, then \textsc{FreeGame}$_{\varepsilon}%
$\ requires $n^{\widetilde{\Omega}(\varepsilon^{-1}\log n)}$\ time assuming
the ETH. \ The second deficiency is that, because of our reliance on the
clause/variable game, we were unable to prove \textit{anything} when
$\varepsilon$\ was greater than some small, fixed constant $\epsilon$. \ This
is particularly inconvenient, since it prevents us from saying that we have an
\textquotedblleft$\mathsf{AM}\left(  2\right)  $\ protocol,\textquotedblright%
\ if $\mathsf{AM}\left(  2\right)  $\ is defined with the conventional
completeness/soundness gap of $2/3$\ vs.\ $1/3$. \ The next two subsections
will remedy these deficiencies.

\subsection{\label{HIGHERROR}The High-Error Case}

Our goal, in this subsection, is to show that if $\varepsilon=o\left(
1\right)  $, then deciding whether $\omega\left(  G\right)  =1$\ or
$\omega\left(  G\right)  \leq1-\varepsilon$\ for a given free game
$G$\ requires $n^{\widetilde{\Omega}(\varepsilon^{-1}\log n)}$\ time assuming
the ETH. \ (Later, Theorem \ref{improvedalg}\ will give an algorithm that
nearly achieves this lower bound.) \ To prove the $\varepsilon$-dependent
hardness result, we first need a simple combinatorial lemma, which can be seen
as a generalization of the Birthday Paradox to regular bipartite graphs.

\begin{lemma}
\label{bipartitelem}Consider a bipartite graph, with $M$ left-vertices each
having degree $c$, and $N$ right-vertices each having degree $d$. \ Choose $k$
left-vertices and $\ell$\ right-vertices uniformly at random, and let $H$ be
the induced subgraph that they form. \ Then%
\begin{equation}
\Pr\left[  H\text{ contains an edge}\right]  \geq\frac{ck\ell}{N}\left(
1-\frac{c^{2}k^{2}}{N}-\frac{ck\ell}{N}\right)  .
\end{equation}

\end{lemma}

\begin{proof}
Given a left-vertex $v\in\left[  M\right]  $, let $\mathcal{N}\left(
v\right)  \subseteq\left[  N\right]  $\ be the set of right-neighbors of $v$;
thus $\left\vert \mathcal{N}\left(  v\right)  \right\vert =c$\ for all $v$.
\ Then by regularity, for any fixed $w\in\left[  N\right]  $\ we have%
\begin{equation}
\Pr_{v\in\left[  M\right]  }\left[  w\in\mathcal{N}\left(  v\right)  \right]
=\frac{c}{N}%
\end{equation}
and%
\begin{equation}
\Pr_{v,v^{\prime}\in\left[  M\right]  ~:~v\neq v^{\prime}}\left[
w\in\mathcal{N}\left(  v\right)  \cap\mathcal{N}\left(  v^{\prime}\right)
\right]  \leq\left(  \frac{c}{N}\right)  ^{2}.
\end{equation}
Now let $A$ be the set of left-vertices in $H$ (thus $\left\vert A\right\vert
=k$), and let $E$\ denote the event that there exist two vertices
$v,v^{\prime}\in A$\ with a common neighbor. \ Then by the union bound,%
\begin{align}
\Pr\left[  E\right]   &  \leq\binom{k}{2}\sum_{w\in\left[  N\right]  }%
\Pr_{v,v^{\prime}\in\left[  M\right]  ~:~v\neq v^{\prime}}\left[
w\in\mathcal{N}\left(  v\right)  \cap\mathcal{N}\left(  v^{\prime}\right)
\right] \\
&  \leq\binom{k}{2}\cdot N\left(  \frac{c}{N}\right)  ^{2}\\
&  \leq\frac{c^{2}k^{2}}{N}.
\end{align}
Furthermore, if $E$ fails, then the left-vertices in $H$ have $ck$\ distinct
neighbors. \ So by the Bonferroni inequality, which states (as a special case)
that%
\begin{equation}
\left(  1-\varepsilon\right)  ^{n}\leq1-\varepsilon n+\left(  \varepsilon
n\right)  ^{2},
\end{equation}
we have%
\begin{equation}
\Pr\left[  H\text{ contains no edge
$\vert$
}\overline{E}\right]  \leq\left(  1-\frac{ck}{N}\right)  ^{\ell}\leq
1-\frac{ck\ell}{N}+\left(  \frac{ck\ell}{N}\right)  ^{2}.
\end{equation}
Hence%
\begin{align}
\Pr\left[  H\text{ contains an edge}\right]   &  \geq\left(  1-\frac
{c^{2}k^{2}}{N}\right)  \left(  \frac{ck\ell}{N}-\left(  \frac{ck\ell}%
{N}\right)  ^{2}\right) \\
&  \geq\frac{ck\ell}{N}\left(  1-\frac{c^{2}k^{2}}{N}-\frac{ck\ell}{N}\right)
.
\end{align}

\end{proof}

We can now prove a more refined upper bound on $\omega(G_{\phi}^{k\times\ell
})$, the success probability in the birthday game, in the case where $k$ and
$\ell$\ are small and $\omega(G_{\phi})$\ is bounded away from $1$.

\begin{lemma}
\label{smallkl}Suppose that $\omega(G_{\phi})\leq1-\epsilon$\ and $k,\ell
\leq\sqrt{\epsilon N}/4$\ for some absolute constant $\epsilon>0$. \ Then%
\begin{equation}
\omega(G_{\phi}^{k\times\ell})\leq1-\Omega\left(  \frac{k\ell}{N}\right)  .
\end{equation}

\end{lemma}

\begin{proof}
Reusing notation from Section \ref{BASIC3SAT}\ (and in particular, from the
proof of Lemma \ref{vardist}), we have%
\begin{align}
\omega(G_{\phi}^{k\times\ell})  &  \leq\operatorname*{E}_{\mathcal{U}}\left[
V_{BD}\right] \\
&  =\sum_{I,J}\Pr_{\mathcal{U}}\left[  I,J\right]  \cdot V_{BD}\left(
I,J,a\left(  I\right)  ,b\left(  J\right)  \right) \\
&  \leq\Pr_{\mathcal{U}}\left[  S_{IJ}=0\right]  +\sum_{I,J:S_{IJ}\geq1}%
\Pr_{\mathcal{U}}\left[  I,J\right]  \cdot V_{BD}\left(  I,J,a\left(
I\right)  ,b\left(  J\right)  \right) \\
&  =\Pr_{\mathcal{U}}\left[  S_{IJ}=0\right]  +\sum_{I,J:S_{IJ}\geq1}%
\Pr_{\mathcal{D}}\left[  \left(  I,J\right)  \right]  \frac{3k\ell/N}{S_{IJ}%
}\cdot V_{BD}\left(  I,J,a\left(  I\right)  ,b\left(  J\right)  \right) \\
&  \leq\Pr_{\mathcal{U}}\left[  S_{IJ}=0\right]  +\frac{3k\ell}{N}%
\sum_{I,J:S_{IJ}\geq1}\Pr_{\mathcal{D}}\left[  \left(  I,J\right)  \right]
\cdot V_{BD}\left(  I,J,a\left(  I\right)  ,b\left(  J\right)  \right) \\
&  =\Pr_{\mathcal{U}}\left[  S_{IJ}=0\right]  +\frac{3k\ell}{N}%
\operatorname*{E}_{\mathcal{D}}\left[  V_{BD}\right] \\
&  \leq\Pr_{\mathcal{U}}\left[  S_{IJ}=0\right]  +\frac{3k\ell}{N}%
\omega(G_{\phi})\\
&  \leq\left(  1-\frac{3k\ell}{N}\left(  1-\frac{9k^{2}}{N}-\frac{3k\ell}%
{N}\right)  \right)  +\frac{3k\ell}{N}\left(  1-\epsilon\right) \label{sij0}\\
&  =1-\frac{3k\ell}{N}\left(  \epsilon-\frac{9k^{2}}{N}-\frac{3k\ell}%
{N}\right) \\
&  =1-\Omega\left(  \frac{k\ell}{N}\right)  , \label{last}%
\end{align}
where line (\ref{sij0}) used Lemma \ref{bipartitelem}, and line (\ref{last})
used the assumption that $k,\ell\leq\sqrt{\epsilon N}/4$.
\end{proof}

Lemma \ref{smallkl} has the following corollary, which gives a counterpart of
Theorem \ref{mainthm1}\ for $\mathsf{AM}\left[  2\right]  $\ protocols with
less than $\sqrt{n}$\ communication.

\begin{corollary}
\label{am2higherr}For all $\varepsilon>0$, there exists an $\mathsf{AM}\left(
2\right)  $\ protocol for \textsc{3Sat}\ instances of size $n$\ which uses
$O(\sqrt{\varepsilon n}\operatorname*{polylog}n)$ bits of communication, and
which has a $1$\ vs.\ $1-\varepsilon$\ completeness/soundness gap.
\end{corollary}

We also get the desired hardness result for \textsc{FreeGame}.

\begin{theorem}
\label{freegamecor2}Assuming the ETH, there exists a constant $\Delta>0$\ such
that \textsc{FreeGame}$_{\varepsilon}$\ requires $n^{\widetilde{\Omega
}(\varepsilon^{-1}\log n)}$\ deterministic time, for all $\varepsilon
\in\left[  1/n,\Delta\right]  $. \ (Likewise, \textsc{FreeGame}$_{\varepsilon
}$\ requires $n^{\widetilde{\Omega}(\varepsilon^{-1}\log n)}$\ randomized time
assuming the Randomized ETH.)
\end{theorem}

\begin{proof}
Set $\Delta:=\epsilon/16$, where $\epsilon$\ is the constant from Lemma
\ref{smallkl}. \ Fix a function $\varepsilon=\varepsilon\left(  M\right)
\in\left[  1/M,\Delta\right]  $, and suppose that \textsc{FreeGame}%
$_{\varepsilon}$\ instances of size $M$\ were solvable in time%
\begin{equation}
M^{o\left(  \frac{\varepsilon^{-1}\log M}{\operatorname*{polylog}%
(\varepsilon^{-1}\log M)}\right)  }.
\end{equation}
We need to show how, using that, we could decide a \textsc{3Sat}\ instance
$\varphi$\ of size $n$\ in\ time $2^{o\left(  n\right)  }$, thereby violating
the ETH. \ The first step is to convert $\varphi$\ into a PCP $\phi$\ of size
$N=n\operatorname*{polylog}n$. \ Next, we generate the birthday repetition
$G_{\phi}^{k\times k}$, where $k=\sqrt{\varepsilon N}$. \ (Here we use the
assumption $\varepsilon\geq1/n$\ to ensure that $k\geq1$, and we use the
assumption $\varepsilon\leq\epsilon/16$\ to ensure that $k\leq\sqrt{\epsilon
N}/4$.) \ Note that the sizes of $G_{\phi}^{k\times k}$'s question and answer
sets are $M=2^{k\log N}=N^{k}$.

If $\phi$\ is satisfiable then $\omega(G_{\phi}^{k\times k})=1$, while by
Lemma \ref{smallkl}, if $\phi$\ is unsatisfiable then%
\begin{equation}
\omega(G_{\phi}^{k\times k})\leq1-\Omega\left(  \frac{k^{2}}{N}\right)
<1-2\varepsilon.
\end{equation}
So by approximating $\omega(G_{\phi}^{k\times k})$\ to within $\pm\varepsilon
$, we can distinguish these cases and thereby decide whether $\varphi$\ was
satisfiable. \ Using our hypothesized algorithm for \textsc{FreeGame}%
$_{\varepsilon}$, this takes time%
\begin{align}
\exp\left(  o\left(  \frac{\varepsilon^{-1}\log^{2}M}{\operatorname*{polylog}%
(\varepsilon^{-1}\log M)}\right)  \right)   &  =\exp\left(  o\left(
\frac{\left(  k^{2}/N\right)  ^{-1}\cdot k^{2}\log^{2}N}%
{\operatorname*{polylog}(\left(  k^{2}/N\right)  ^{-1}\log N^{k})}\right)
\right) \\
&  =\exp\left(  o\left(  \frac{N\log^{2}N}{\left(  \log\left(  N/k^{2}\right)
+\log k+\log\log N\right)  ^{R}}\right)  \right)
\end{align}
for some constant $R$. \ Note that if $k$ is large, then $\log k$\ is
$\Omega\left(  \log N\right)  $, while if $k$\ is small, then $\log\left(
N/k^{2}\right)  $\ is $\Omega\left(  \log N\right)  $. \ Therefore, provided
$R$\ is large enough, the denominator will contain enough factors of $\log
N$\ to clear all the $\log N$\ factors in the numerator, and our algorithm
will have running time $\exp\left(  o\left(  n\right)  \right)  $, giving the
desired violation of the ETH. \ This reduction produces a deterministic
algorithm if the \textsc{FreeGame}\ algorithm was deterministic, or randomized
if the \textsc{FreeGame}\ algorithm was randomized.
\end{proof}

We conjecture that the bound of Theorem \ref{freegamecor2}\ could be improved
to $n^{\widetilde{\Omega}(\varepsilon^{-2}\log n)}$, by considering free games
$G$\ with $\omega\left(  G\right)  \approx1/2$\ rather than $\omega\left(
G\right)  \approx1$. \ This is a problem that we leave to future work.

\subsection{The Low-Error Case\label{LOWERROR}}

There is one obvious question that we haven't yet addressed: can we give an
$\mathsf{AM}\left(  2\right)  $\ protocol for \textsc{3Sat}\ with
\textit{near-perfect} soundness? \ Or equivalently, given a free game $G$ and
some\ tiny $\delta>0$, can we show that (assuming the ETH) there is no
polynomial-time algorithm even to decide whether $\omega\left(  G\right)
=1$\ or $\omega\left(  G\right)  <\delta$? \ In this section we show that,
using high-powered PCP machinery, we can indeed do this, although the result
we get is probably not optimal.

One's first idea would be to apply ordinary parallel repetition to the
birthday game---i.e., to consider $(G_{\phi}^{k\times\ell})^{m}$ for some
$m>1$. \ Alas, this fails to work for an interesting reason. \ Namely, in the
statement of the Parallel Repetition Theorem (Theorem \ref{prt}), there is a
$1/\log\left\vert A\right\vert \left\vert B\right\vert $ factor in the
exponent, which is known to be necessary in general by a result of Feige and
Verbitsky \cite{feigeverbitsky}. \ That factor immediately pushes the running
time of our putative \textsc{3Sat} algorithm above $2^{O\left(  n\right)  }$,
preventing a contradiction with the ETH.

Note that Rao \cite{rao}\ proved that, for the special case of
\textit{projection games}, one can dramatically improve the Parallel
Repetition Theorem, to show that $\omega\left(  G^{t}\right)  \leq
\omega\left(  G\right)  ^{\Omega\left(  t\right)  }$\ with no dependence on
$\log\left\vert A\right\vert \left\vert B\right\vert $. \ Here a projection
game is a two-prover game $G=\left(  X,Y,A,B,\mathcal{D},V\right)  $\ (not
necessarily free)\ with $V\in\left\{  0,1\right\}  $\ such that, for every
$\left(  x,y\right)  \in X\times Y$\ in the support of $\mathcal{D}$\ and
every $a\in A$, there\ is a unique $b\in B$\ such that $V\left(
x,y,a,b\right)  =1$. \ Unfortunately, while the clause/variable game itself is
a projection game, its birthday repetition is not.

Recently, Shaltiel \cite{shaltiel}\ proved a \textquotedblleft
derandomized\textquotedblright\ version of the Parallel Repetition Theorem for
the special case of free games. \ In particular, given a free game $G=\left(
X,Y,A,B,V\right)  $ with $V\in\left\{  0,1\right\}  $, Shaltiel constructs a
new free game $G_{t}=\left(  X_{t},Y_{t},A_{t},B_{t},V_{t}\right)  $, which
satisfies $\omega\left(  G_{t}\right)  \leq\omega\left(  G\right)  ^{t}$, as
well as\ $\omega\left(  G_{t}\right)  =1$\ whenever $\omega\left(  G\right)
=1$. \ Furthermore, the question sets $X_{t}$\ and $Y_{t}$\ in Shaltiel's game
have size at most $\left(  \left\vert X\right\vert \left\vert Y\right\vert
\left\vert A\right\vert \left\vert B\right\vert \right)  ^{O\left(  t\right)
}$, which is perfect for our application. \ Unfortunately, the \textit{answer}
sets $A_{t}$\ and $B_{t}$\ have size $\exp(\left(  t\log\left\vert
X\right\vert \left\vert Y\right\vert \left\vert A\right\vert \left\vert
B\right\vert \right)  ^{C})$\ for some large constant $C$, and this once again
prevents the desired contradiction with the ETH.

Finally, if we try to apply parallel repetition to the clause/variable game
\textit{before} applying birthday repetition, then the situation is even
worse. \ For even one or two rounds of parallel repetition will blow up the
question sets $X$\ and $Y$\ so that they no longer have size $n^{1+o\left(
1\right)  }$, meaning that we no longer have any hope of finding a collision
using $n^{1/2+o\left(  1\right)  }$\ rounds of birthday repetition.

Currently, then, the best approach we know to the low-error case is simply to
choose a PCP that \textit{already} has low error, and then ensure that
birthday repetition does not increase its error much further. \ In particular,
rather than Theorem \ref{pcpthm}, we can start with the following result of
Moshkovitz and Raz \cite{mr}:

\begin{theorem}
[PCP Theorem, Moshkovitz-Raz Version \cite{mr}]\label{mrthm}Given a
\textsc{3Sat}\ instance $\varphi$\ of size $n$ as well as $\delta>0$, it is
possible in $\operatorname*{poly}\left(  n\right)  $\ time to produce a
$2$-CSP instance $\phi$, with $n^{1+o\left(  1\right)  }\operatorname*{poly}%
\left(  1/\delta\right)  $\ variables and constraints, and over an alphabet
$\Sigma$\ of size $\left\vert \Sigma\right\vert \leq2^{\operatorname*{poly}%
\left(  1/\delta\right)  }$, such that:

\begin{itemize}
\item \textbf{(Completeness)} If $\operatorname*{SAT}\left(  \varphi\right)
=1$\ then $\operatorname*{SAT}\left(  \phi\right)  =1$.

\item \textbf{(Soundness)} If $\operatorname*{SAT}\left(  \varphi\right)
<1$\ then $\operatorname*{SAT}\left(  \phi\right)  <\delta$.

\item \textbf{(Balance)} The constraint graph of $\phi$\ is bipartite, and
every variable appears in exactly $d$ constraints, for some
$d=\operatorname*{poly}\left(  1/\delta\right)  $.
\end{itemize}
\end{theorem}

Since Theorem \ref{mrthm}\ outputs a $2$-CSP $\phi$, we do not even need to
consider the clause/variable game. \ Rather, $\phi$\ directly gives rise to a
two-prover game $H_{\phi}$, in which Arthur chooses a constraint $C$ of $\phi
$\ uniformly at random, sends one of $C$'s variables to Merlin$_{1}$ and the
other to Merlin$_{2}$, gets back their values, and accepts if and only if the
values satisfy $C$. \ Clearly, if $\operatorname*{SAT}\left(  \varphi\right)
=1$\ then $\omega(H_{\phi})=1$, while if $\operatorname*{SAT}\left(
\varphi\right)  <1$\ then $\omega(H_{\phi})<\delta$.

Now let $N=n^{1+o\left(  1\right)  }\operatorname*{poly}\left(  1/\delta
\right)  $\ be the number of variables in $\phi$, let $k,\ell\in\left[
N\right]  $, and consider the birthday repetition $H_{\phi}^{k\times\ell}$.
\ Observe that, in the variation distance argument from Section
\ref{BASIC3SAT}, the only special property of $G_{\phi}$\ that we used was the
\textit{regularity} of the constraint graph, and that property also holds for
$H_{\phi}$. \ For this reason, we immediately get the following counterpart of
Theorem \ref{gbdthm}:

\begin{theorem}
\label{gbdthm2}For all $k,\ell\in\left[  N\right]  $,%
\begin{equation}
\omega\left(  H_{\phi}\right)  \geq\omega(H_{\phi}^{k\times\ell})-O\left(
\sqrt{\frac{N}{k\ell}}\right)  .
\end{equation}

\end{theorem}

So in particular, suppose we set $k:=\sqrt{N}/\delta$. \ Then if
$\operatorname*{SAT}\left(  \varphi\right)  <1$, we find that%
\begin{equation}
\omega(H_{\phi}^{k\times k})\leq\omega\left(  H_{\phi}\right)  +O(\sqrt
{N}/k)=O\left(  \delta\right)  . \label{delta}%
\end{equation}
Of course, if $\operatorname*{SAT}\left(  \varphi\right)  =1$\ then
$\omega(H_{\phi}^{k\times k})=1$. \ This gives us the following corollary:

\begin{corollary}
\label{am2lowerr}For all $\delta>0$, there exists an $\mathsf{AM}\left(
2\right)  $\ protocol for \textsc{3Sat}\ instances of size $n$\ which uses
$n^{1/2+o\left(  1\right)  }\operatorname*{poly}\left(  1/\delta\right)  $
bits of communication, and which has a $1$\ vs.\ $\delta$%
\ completeness/soundness gap.
\end{corollary}

We also get the following hardness result for distinguishing $\omega\left(
G\right)  =1$\ from $\omega\left(  G\right)  <\delta$:

\begin{theorem}
\label{freegamecor3}Assuming the ETH, any deterministic algorithm to decide
whether $\omega\left(  G\right)  =1$\ or $\omega\left(  G\right)  <\delta$,
given as input a description of a free game $G$ of size $n$, requires
$n^{\operatorname*{poly}\left(  \delta\right)  \cdot\left(  \log n\right)
^{1-o\left(  1\right)  }}$\ time. \ (Likewise, any randomized algorithm
requires $n^{\operatorname*{poly}\left(  \delta\right)  \cdot\left(  \log
n\right)  ^{1-o\left(  1\right)  }}$\ time\ assuming the Randomized ETH.)
\end{theorem}

\begin{proof}
Given a free game $G$\ of size $M$, suppose we could decide whether
$\omega\left(  G\right)  =1$\ or $\omega\left(  G\right)  <\delta$\ in time
$M^{p\left(  \delta\right)  \cdot\left(  \log M\right)  ^{1-\eta}}$, for some
constant $\eta>0$\ and sufficiently large polynomial $p$. \ We need to show
how, using that, we could decide a \textsc{3Sat}\ instance $\varphi$\ of size
$n$\ in\ time $2^{o\left(  n\right)  }$, thereby violating the ETH. \ The
first step is to convert $\varphi$\ into a $2$-CSP $\phi$\ with
$N=n^{1+o\left(  1\right)  }\operatorname*{poly}\left(  1/\delta\right)  $
variables, using Theorem \ref{mrthm}. \ Observe that the game $H_{\phi
}=\left(  X,Y,A,B,V\right)  $\ satisfies $\left\vert X\right\vert =\left\vert
Y\right\vert =N$\ and $\left\vert A\right\vert =\left\vert B\right\vert
=\left\vert \Sigma\right\vert =2^{\operatorname*{poly}\left(  1/\delta\right)
}$.

Next, we generate the birthday repetition $H_{\phi}^{k\times k}$, where
$k:=c\sqrt{N}/\delta$ for some suitable constant $c$. \ Note that $H_{\phi
}^{k\times k}$ has question sets of size%
\begin{equation}
N^{k}=\exp\left(  \frac{c\sqrt{N}\log N}{\delta}\right)  =\exp\left(
\frac{n^{1/2+o\left(  1\right)  }}{\operatorname*{poly}\left(  \delta\right)
}\right)
\end{equation}
and answer sets of size%
\begin{equation}
2^{k\operatorname*{poly}\left(  1/\delta\right)  }=\exp\left(  \frac
{n^{1/2+o\left(  1\right)  }}{\operatorname*{poly}\left(  \delta\right)
}\right)  .
\end{equation}
Thus, we set $M:=\exp\left(  n^{1/2+o\left(  1\right)  }/\operatorname*{poly}%
\left(  \delta\right)  \right)  $.

If $\phi$\ is satisfiable then $\omega(H_{\phi}^{k\times k})=1$, while if
$\phi$\ is unsatisfiable then $\omega(H_{\phi}^{k\times k})\leq\delta$\ by
equation (\ref{delta}), provided the constant $c$ was large enough. \ So by
distinguishing these cases, we can decide whether $\varphi$\ was satisfiable.
\ Using our hypothesized algorithm, this takes time%
\begin{equation}
\exp\left(  p\left(  \delta\right)  \cdot\log^{2-\eta}M\right)  =\exp\left(
p\left(  \delta\right)  \left(  \frac{n^{1/2+o\left(  1\right)  }%
}{\operatorname*{poly}\left(  \delta\right)  }\right)  ^{2-\eta}\right)
=\exp\left(  n^{1-\eta/2+o\left(  1\right)  }\right)  ,
\end{equation}
provided the polynomial $p$ was large enough. \ This gives us our desired
violation of the ETH.
\end{proof}

We conjecture that, assuming the ETH, distinguishing $\omega\left(  G\right)
=1$\ from $\omega\left(  G\right)  <\delta$ for a free game $G$ should require
$n^{\Omega\left(  \frac{\log n}{\log1/\delta}\right)  }$\ time for all
$\delta\leq1/2$, matching an upper bound that we will give in Theorem
\ref{improvedalg}. \ A first step toward proving this conjecture would be to
improve\ Theorem \ref{mrthm} (the result of Moshkovitz and Raz), so that it
gave alphabet size $\left\vert \Sigma\right\vert \leq\operatorname*{poly}%
\left(  1/\delta\right)  $\ rather than $\left\vert \Sigma\right\vert
\leq2^{\operatorname*{poly}\left(  1/\delta\right)  }$. \ This is a well-known
open problem. \ However, even if that problem were solved, one would
\textit{also} need a more refined analysis of birthday repetition, to
eliminate the dependence of $k$\ on $1/\delta$\ in the proof of Theorem
\ref{freegamecor3}.

\subsection{Complexity Consequences\label{CC}}

Setting $\delta:=1/3$, Theorem \ref{gbdthm2} finally puts us in a position to
say that%
\begin{equation}
\text{\textsc{3Sat}}\in\mathsf{AM}_{n^{1/2+o\left(  1\right)  }}\left(
2\right)  ,
\end{equation}
where $\mathsf{AM}_{n^{1/2+o\left(  1\right)  }}\left(  2\right)  $\ is
defined with a $2/3$\ vs.\ $1/3$\ completeness/soundness gap, as in Definition
\ref{amkdef}. \ If we further combine this\ with a tight Cook-Levin Theorem
(see, e.g., Tourlakis \cite{tourlakis}), showing that every language
$L\in\mathsf{NTIME}\left[  n\right]  $\ can be efficiently reduced to a set of
\textsc{3Sat}\ instances of size $n\operatorname*{polylog}n$, then we get the
following corollary:

\begin{corollary}
\label{ntimecor}$\mathsf{NTIME}\left[  n\right]  \subseteq\mathsf{AM}%
_{n^{1/2+o\left(  1\right)  }}\left(  2\right)  .$
\end{corollary}

Let us observe that Corollary \ref{ntimecor}\ is non-relativizing.

\begin{proposition}
\label{oraclesep}There exists an oracle $A$ relative to which $\mathsf{NTIME}%
^{A}\left[  n\right]  \not \subset \mathsf{AM}_{n^{1/2+o\left(  1\right)  }%
}^{A}\left(  2\right)  $.
\end{proposition}

\begin{proof}
[Proof Sketch]For each $n$, the oracle $A$\ encodes a Boolean function
$A_{n}:\left\{  0,1\right\}  ^{n}\rightarrow\left\{  0,1\right\}  $, which is
either identically $0$ or else $1$ on exactly one input. \ Let $L_{A}$\ be the
unary language defined by $0^{n}\in L_{A}$\ if there\ exists an $x\in\left\{
0,1\right\}  ^{n}$\ such that $A_{n}\left(  x\right)  =1$,\ and $0^{n}\notin
L_{A}$\ otherwise. \ Then certainly $L_{A}\in\mathsf{NTIME}^{A}\left[
n\right]  $ for all $A$. \ On the other hand, using standard diagonalization
techniques, it is not hard to construct $A$ in such a way that $L_{A}%
\notin\mathsf{AM}_{n^{1/2+o\left(  1\right)  }}^{A}\left(  2\right)  $---or
even $L_{A}\notin\mathsf{AM}_{n/4}^{A}\left(  2\right)  $. \ Intuitively, if
the Merlins send only $n/4$\ bits each to Arthur (so $n/2$\ bits in total),
then regardless of how those bits depend on their random challenges, with high
probability Arthur will still need to query $A$\ on at least $2^{n/2}$\ inputs
to confirm that $0^{n}\in L_{A}$. \ We omit the details, which are similar to
those in the paper of Fortnow and Sipser \cite{fs}.
\end{proof}

We leave as an open problem whether Corollary \ref{ntimecor} is algebrizing in
the sense of Aaronson and Wigderson \cite{awig}.

\section{\label{LIMITS}Limitations of Multi-Prover $\mathsf{AM}$}

Our goal in this section is to prove that our \textsc{3Sat}\ protocol is
essentially optimal assuming the ETH, that $\mathsf{AM}\left(  k\right)
\subseteq\mathsf{EXP}$\ for all $k=\operatorname*{poly}\left(  n\right)  $,
and that $\mathsf{AM}\left(  k\right)  =\mathsf{AM}$\ for all $k=O\left(  \log
n\right)  $.

The section is organized as follows. \ First, in Section \ref{ALGSEC}, we give
a quasipolynomial-time algorithm for estimating the value of a $2$-player free
game. \ This algorithm implies that $\mathsf{AM}\left(  2\right)
\subseteq\mathsf{EXP}$, and even (with some more work) that $\mathsf{AM}%
\left(  2\right)  \subseteq\mathsf{AM}^{\mathsf{NP}}$. \ The algorithm also
implies that, if there exists an $\mathsf{AM}\left(  2\right)  $\ protocol for
\textsc{3Sat}\ using $o(\sqrt{n})$\ communication, then \textsc{3Sat}\ is
solvable in $2^{o\left(  n\right)  }$\ time. \ In Section \ref{SUBSAMPLING},
we go further, using a result of Barak et al.\ \cite{bhhs} about subsampling
dense CSPs to show that the value of any free game can be approximated by the
value of a logarithmic-sized random subgame, and as a consequence, that
$\mathsf{AM}\left(  2\right)  =\mathsf{AM}$. \ Finally, in Section
\ref{KMERLIN}, we generalize these results from $\mathsf{AM}\left(  2\right)
$\ to $\mathsf{AM}\left(  k\right)  $\ for all $k=\operatorname*{poly}\left(
n\right)  $.

\subsection{\label{ALGSEC}The Basic Approximation Algorithm}

We now explain how to approximate the value of a free game in quasipolynomial time.

\begin{theorem}
\label{freegamealg}\textsc{FreeGame}$_{\varepsilon}$ is solvable in time
$n^{O(\varepsilon^{-2}\log n)}$. \ In more detail, given as input a
description of a free game $G=\left(  X,Y,A,B,V\right)  $, there exists a
randomized algorithm running in time $\left\vert X\right\vert \cdot\left\vert
A\right\vert ^{O(\varepsilon^{-2}\log\left\vert Y\right\vert \left\vert
B\right\vert )}$, which estimates $\omega\left(  G\right)  $\ to within
additive error $\pm\varepsilon$, with at least $2/3$\ success probability.
\ There also exists a deterministic algorithm running in time $\left(
\left\vert X\right\vert \left\vert A\right\vert \right)  ^{O(\varepsilon
^{-2}\log\left\vert Y\right\vert \left\vert B\right\vert )}$.
\end{theorem}

\begin{proof}
The randomized estimation algorithm, call it \texttt{REst}, works as follows.
\ First \texttt{REst} chooses a subset of questions $S\subseteq X$\ uniformly
at random, subject to $\left\vert S\right\vert =\kappa$\ where%
\begin{equation}
\kappa:=\left\lceil \frac{\ln\left(  6\left\vert Y\right\vert \left\vert
B\right\vert \right)  }{\varepsilon^{2}}\right\rceil .
\end{equation}
Next \texttt{REst} loops over all $\left\vert A\right\vert ^{\kappa}%
$\ possible settings $\alpha:S\rightarrow A$\ of the answers to the $\kappa
$\ questions in $S$. \ For each such $\alpha$, \texttt{REst} does the following:

\begin{enumerate}
\item[(1)] It computes Merlin$_{2}$'s \textquotedblleft optimal
response\textquotedblright\ $b_{\alpha}:Y\rightarrow B$ to $\alpha$, supposing
that Merlin$_{1}$\ was only asked questions in $S$. \ For each question $y\in
Y$, in other words, \texttt{Est} finds a response $b_{\alpha}\left(  y\right)
\in B$\ that maximizes%
\begin{equation}
\operatorname*{E}_{x\in S}\left[  V\left(  x,y,\alpha\left(  x\right)
,b_{\alpha}\left(  y\right)  \right)  \right]
\end{equation}
(breaking ties arbitrarily).

\item[(2)] It computes Merlin$_{1}$'s \textquotedblleft optimal
response\textquotedblright\ $a_{\alpha}:X\rightarrow A$\ to $b_{\alpha}$.
\ For each $x\in X$, in other words, \texttt{REst} finds an $a_{\alpha}\left(
x\right)  \in A$\ that maximizes%
\begin{equation}
\operatorname*{E}_{y\in Y}\left[  V\left(  x,y,a_{\alpha}\left(  x\right)
,b_{\alpha}\left(  y\right)  \right)  \right]  .
\end{equation}

\item[(3)] It computes the \textquotedblleft value\textquotedblright\ obtained
from the setting $\alpha$, as follows:%
\begin{equation}
W_{\alpha}:=\operatorname*{E}_{x\in X,y\in Y}\left[  V\left(  x,y,a_{\alpha
}\left(  x\right)  ,b_{\alpha}\left(  y\right)  \right)  \right]
\end{equation}

\end{enumerate}

Finally, \texttt{REst}\ computes%
\begin{equation}
W:=\max_{\alpha}W_{\alpha},
\end{equation}
and outputs $W+\varepsilon$\ as its estimate for $\omega\left(  G\right)  $.

Clearly \texttt{REst}\ runs in time%
\begin{equation}
O\left(  \left\vert A\right\vert ^{\kappa}\left(  \left\vert Y\right\vert
\left\vert B\right\vert \kappa+\left\vert X\right\vert \left\vert A\right\vert
\left\vert Y\right\vert +\left\vert X\right\vert \left\vert Y\right\vert
\right)  \right)  =\left\vert X\right\vert \cdot\left\vert A\right\vert
^{O(\varepsilon^{-2}\log\left\vert Y\right\vert \left\vert B\right\vert )}.
\label{runningtime}%
\end{equation}
To prove correctness, we need to argue that, with high probability over the
choice of $S$, we have%
\begin{equation}
\left\vert \left(  W+\varepsilon\right)  -\omega\left(  G\right)  \right\vert
\leq\varepsilon.
\end{equation}
First observe that $W_{\alpha}\leq\omega\left(  G\right)  $\ for every
$\alpha$. \ Therefore $W\leq\omega\left(  G\right)  $\ as well, and
$W+\varepsilon\leq\omega\left(  G\right)  +\varepsilon$. \ So it suffices to
prove the other direction: that $W\geq\omega\left(  G\right)  -2\varepsilon
$\ with at least $2/3$\ probability over $S$.

Let $a^{\ast}:X\rightarrow A$\ be an \textit{optimal} strategy for
Merlin$_{1}$ in the game $G$: that is, a strategy that, when combined with an
optimal response $b^{\ast}:Y\rightarrow B$\ by Merlin$_{2}$, achieves the
value\ $\omega\left(  G\right)  $. \ Also, fix a particular question $y\in
Y$\ and answer $b\in B$. \ Then since the function $V$\ is $\left[
0,1\right]  $-valued, Hoeffding's inequality (which also holds in the case of
sampling without replacement) gives us%
\begin{equation}
\Pr_{S\subseteq X,\left\vert S\right\vert =\kappa}\left[  \left\vert
\operatorname*{E}_{x\in S}\left[  V\left(  x,y,a^{\ast}\left(  x\right)
,b\right)  \right]  -\operatorname*{E}_{x\in X}\left[  V\left(  x,y,a^{\ast
}\left(  x\right)  ,b\right)  \right]  \right\vert >\varepsilon\right]
<2e^{-\varepsilon^{2}\kappa}.
\end{equation}
So by the union bound, if we choose $S$\ randomly, then we have%
\begin{equation}
\left\vert \operatorname*{E}_{x\in S}\left[  V\left(  x,y,a^{\ast}\left(
x\right)  ,b\right)  \right]  -\operatorname*{E}_{x\in X}\left[  V\left(
x,y,a^{\ast}\left(  x\right)  ,b\right)  \right]  \right\vert \leq
\varepsilon\label{cond}%
\end{equation}
for \textit{every} $y\in Y$ and $b\in B$\ simultaneously, with probability at
least%
\begin{equation}
1-2e^{-\varepsilon^{2}\kappa}\left\vert Y\right\vert \left\vert B\right\vert
\geq\frac{2}{3}%
\end{equation}
over $S$. \ So suppose the inequality (\ref{cond}) holds. \ Let $\alpha^{\ast
}:S\rightarrow A$\ be the restriction of the optimal strategy $a^{\ast}$\ to
the set $S$, and let $b_{\alpha^{\ast}}:Y\rightarrow B$\ be an optimal
response to $\alpha^{\ast}$. \ Then%
\begin{align}
W  &  \geq W_{\alpha^{\ast}}\\
&  =\max_{a:X\rightarrow A}\operatorname*{E}_{x\in X,y\in Y}\left[  V\left(
x,y,a\left(  x\right)  ,b_{\alpha^{\ast}}\left(  y\right)  \right)  \right] \\
&  \geq\operatorname*{E}_{x\in X,y\in Y}\left[  V\left(  x,y,a^{\ast}\left(
x\right)  ,b_{\alpha^{\ast}}\left(  y\right)  \right)  \right] \\
&  \geq\operatorname*{E}_{x\in S,y\in Y}\left[  V\left(  x,y,a^{\ast}\left(
x\right)  ,b_{\alpha^{\ast}}\left(  y\right)  \right)  \right]  -\varepsilon
\label{s1}\\
&  \geq\operatorname*{E}_{x\in S,y\in Y}\left[  V\left(  x,y,a^{\ast}\left(
x\right)  ,b^{\ast}\left(  y\right)  \right)  \right]  -\varepsilon\\
&  \geq\operatorname*{E}_{x\in X,y\in Y}\left[  V\left(  x,y,a^{\ast}\left(
x\right)  ,b^{\ast}\left(  y\right)  \right)  \right]  -2\varepsilon
\label{s2}\\
&  =\omega\left(  G\right)  -2\varepsilon,
\end{align}
where lines (\ref{s1}) and (\ref{s2}) used inequality (\ref{cond}).

Finally, to get a deterministic estimation algorithm---call it \texttt{Est}%
---we simply need to loop over all possible $S\subseteq X$\ with $\left\vert
S\right\vert =\kappa$, rather than choosing $S$\ randomly. \ We then output
the maximum of $W+\varepsilon$\ over all $S$\ as our estimate for
$\omega\left(  G\right)  $. \ This yields a running time of%
\begin{equation}
\left\vert X\right\vert ^{\kappa}\cdot\left\vert X\right\vert \left\vert
A\right\vert ^{O(\varepsilon^{-2}\log\left\vert Y\right\vert \left\vert
B\right\vert )}=\left(  \left\vert X\right\vert \left\vert A\right\vert
\right)  ^{O(\varepsilon^{-2}\log\left\vert Y\right\vert \left\vert
B\right\vert )}.
\end{equation}

\end{proof}

Let us point out some simple modifications to the algorithm \texttt{Est}\ from
Theorem \ref{freegamealg}, which can improve its running time of
$n^{O(\varepsilon^{-2}\log n)}$\ in certain cases.

\begin{theorem}
\label{improvedalg}Given a free game $G$ of size $n$, there is a deterministic
algorithm running in $n^{O(\varepsilon^{-1}\log n)}$\ time to decide whether
$\omega\left(  G\right)  =1$ or $\omega\left(  G\right)  \leq1-\varepsilon$
(promised that one of those is the case), and there is a deterministic
algorithm running in $n^{O\left(  1+\frac{\log n}{\log1/\delta}\right)  }%
$\ time to decide whether $\omega\left(  G\right)  =1$ or $\omega\left(
G\right)  <\delta$.
\end{theorem}

\begin{proof}
In both cases, the key observation is that when running \texttt{Est}, we no
longer need to estimate the quantity $\operatorname*{E}_{x\in X}\left[
V\left(  x,y,a^{\ast}\left(  x\right)  ,b\right)  \right]  $\ to within
$\pm\varepsilon$, and to pay the $1/\varepsilon^{2}$\ price that comes from
Hoeffding's inequality for doing so. \ Instead, for each $y\in Y$ and $b\in
B$, we simply need to know whether $\operatorname*{E}_{x\in X}\left[  V\left(
x,y,a^{\ast}\left(  x\right)  ,b\right)  \right]  $\ is $1$ or less than $1$.
\ Or equivalently, does there exist a \textquotedblleft bad\textquotedblright%
\ $x\in X$---one such that $V\left(  x,y,a^{\ast}\left(  x\right)  ,b\right)
<1$? \ Moreover, we are promised that, if such a bad $x$ \textit{does} exist,
then at least an $\varepsilon$\ or a $1-\delta$\ fraction (respectively) of
all $x$'s are bad. \ Thus, when choosing the subset of questions $S\subseteq
X$, it suffices to ensure that, with nonzero probability over $S$, we succeed
in sampling one of the bad $x$'s for every $y\in Y$ and $b\in B$. \ By the
union bound, this means that it suffices if, respectively,%
\begin{equation}
\left(  1-\varepsilon\right)  ^{\kappa}<\frac{1}{3\left\vert Y\right\vert
\left\vert B\right\vert }~~\ ~~\text{or \ \ \ \ }\delta^{\kappa}<\frac
{1}{3\left\vert Y\right\vert \left\vert B\right\vert }%
\end{equation}
where $\kappa=\left\vert S\right\vert $. \ Solving, we get respectively
$\kappa=O(\varepsilon^{-1}\log\left\vert Y\right\vert \left\vert B\right\vert
)$\ or $\kappa=O\left(  1+\frac{\log\left\vert Y\right\vert \left\vert
B\right\vert }{\log1/\delta}\right)  $. \ Now we just need to plug the lower
values of $\kappa$\ into equation (\ref{runningtime}) from the proof of
Theorem \ref{freegamealg} to get the improved running times.
\end{proof}

The proof of Theorem \ref{freegamealg} has a curious property. \ Namely, we
showed that the value $\omega\left(  G\right)  $\ can in some sense be
well-approximated by restricting attention to a random subset of questions
$S\subseteq X$\ of logarithmic size. \ However, if $G_{S}$ is the subgame
obtained from $G$ by restricting $X$\ to $S$, then the proof did \textit{not}
imply that $\omega(G_{S})\approx\omega\left(  G\right)  $! \ Using Hoeffding's
inequality, one can easily show that $\omega(G_{S})\geq\omega\left(  G\right)
-\varepsilon$\ with high probability over the choice of $S$. \ The difficulty
comes from the other direction---ironically, the \textquotedblleft
trivial\textquotedblright\ direction in the proof of Theorem \ref{freegamealg}%
. \ To get that $W_{\alpha}\leq\omega\left(  G\right)  $, we implicitly used
the fact that $W_{\alpha}$\ was the value of a strategy pair $\left(
a_{\alpha},b_{\alpha}\right)  $ for the \textit{whole} game $G$, not merely
for the subgame $G_{S}$. \ Therefore, nothing we said implies that
$\omega(G_{S})\leq\omega\left(  G\right)  $, or even $\omega(G_{S})\leq
\omega\left(  G\right)  +\varepsilon$. \ And this makes intuitive sense: if
the Merlins know that Merlin$_{1}$'s question $x$\ will be restricted to a set
of logarithmic size, then how do we know they can't exploit that knowledge to
win with greater probability? \ As we'll discuss in Section \ref{SUBSAMPLING},
it turns out that one \textit{can} prove the stronger result that
$\omega(G_{S})\leq\omega\left(  G\right)  +\varepsilon$\ with high probability
over $S$---and this, in turn, lets one prove that $\mathsf{AM}\left(
2\right)  =\mathsf{AM}$. \ But more work (in particular, that of Alon et
al.\ \cite{avkk}\ and Barak et al.\ \cite{bhhs}) is needed.

Before we discuss that, let us point out some simple corollaries of Theorems
\ref{freegamealg} and \ref{improvedalg}.

\begin{corollary}
\label{amexpcor}$\mathsf{AM}\left(  2\right)  \subseteq\mathsf{EXP}$. \ (In
more detail, we can simulate any $\mathsf{AM}\left(  2\right)  $\ protocol
that uses $p\left(  n\right)  $\ communication and $r\left(  n\right)
=\operatorname*{poly}\left(  n\right)  $ auxiliary randomness in
$2^{O(p\left(  n\right)  ^{2})+r(n)}\operatorname*{poly}\left(  n\right)  $
deterministic time, or $2^{O(p\left(  n\right)  ^{2})}\operatorname*{poly}%
\left(  n\right)  $\ randomized time.)
\end{corollary}

\begin{proof}
Let $L\in\mathsf{AM}\left(  2\right)  $. \ Then given an input $x\in\left\{
0,1\right\}  ^{n}$, the $\mathsf{AM}\left(  2\right)  $\ protocol for checking
whether $x\in L$\ can be represented as a free game $G=\left(
X,Y,A,B,V\right)  $, where $X=Y=A=B=\left\{  0,1\right\}  ^{p\left(  n\right)
}$, and where Arthur's verification function $V$ is computable in randomized
$\operatorname*{poly}\left(  n\right)  $\ time using $r\left(  n\right)
$\ random bits. \ Now by Theorem \ref{freegamealg}, we can estimate
$\omega\left(  G\right)  $\ to additive error (say) $\pm1/10$, using $\left(
\left\vert X\right\vert \left\vert A\right\vert \right)  ^{O\left(
\log\left\vert Y\right\vert \left\vert B\right\vert \right)  }=2^{O(p\left(
n\right)  ^{2})}$\ deterministically-chosen evaluations of $V$. \ Furthermore,
each of these $V$ evaluations can be performed in $\operatorname*{poly}\left(
n\right)  $\ steps by a randomized algorithm (including the time needed for
amplification to exponentially-small error probability), or in $2^{r(n)}%
\operatorname*{poly}\left(  n\right)  $\ steps by a deterministic algorithm.
\ Finally, estimating $\omega\left(  G\right)  $\ lets us decide whether
$\omega\left(  G\right)  \geq2/3$\ or $\omega\left(  G\right)  \leq1/3$, and
hence whether $x\in L$.
\end{proof}

Corollary \ref{amexpcor}\ (and Theorem \ref{improvedalg}) have the following
further consequence:

\begin{corollary}
\label{3satcor}If \textsc{3Sat{}{}~}$\in\mathsf{AM}_{p(n)}\left(  2\right)  $,
then \textsc{3Sat{}{}~}$\in\mathsf{TIME}[2^{O(p\left(  n\right)  ^{2}%
)}\operatorname*{poly}\left(  n\right)  ]$. \ So in particular, assuming the
Randomized ETH, any $\mathsf{AM}\left(  2\right)  $\ protocol for
\textsc{3Sat{}{}}\ with a $1$\ vs.\ $1-\varepsilon$\ completeness/soundness
gap must use $\Omega(\sqrt{\varepsilon n})$\ communication. \ Likewise,
assuming the Randomized ETH, any protocol with a $1$\ vs.\ $\delta$\ gap must
use $\Omega(\sqrt{n\log1/\delta})$\ communication provided $\delta\geq2^{-n}$.
\ (If, moreover, Arthur's verification procedure is deterministic, then it
suffices to assume the ordinary ETH.)
\end{corollary}

Also, a closer examination of the proof of Theorem \ref{freegamealg}\ yields a
better upper bound on $\mathsf{AM}\left(  2\right)  $ than $\mathsf{EXP}$.

\begin{theorem}
\label{amnpcor}$\mathsf{AM}\left(  2\right)  \subseteq\mathsf{AM}%
^{\mathsf{NP}}$.
\end{theorem}

\begin{proof}
[Proof Sketch]We only sketch the proof, since in any case this result will be
superseded by the later result that $\mathsf{AM}\left(  2\right)
=\mathsf{AM}$.

In the algorithm \texttt{REst}\ from Theorem \ref{freegamealg}, the first step
has the form of an $\mathsf{AM}$\ protocol. \ That is, following Corollary
\ref{amexpcor}, let $G=\left(  X,Y,A,B,V\right)  $ be the free game associated
to the $\mathsf{AM}\left(  2\right)  $\ protocol we want to simulate, with
$X=Y=A=B=\left\{  0,1\right\}  ^{\operatorname*{poly}\left(  n\right)  }$.
\ Then in our $\mathsf{AM}^{\mathsf{NP}}$\ simulation, we can have Arthur
first choose a subset $S\subseteq X$\ of size $\kappa=\operatorname*{poly}%
\left(  n\right)  $\ uniformly at random and send $S$ to Merlin. \ Next, using
$\kappa\log\left\vert A\right\vert =\operatorname*{poly}\left(  n\right)
$\ bits, Merlin can send back a complete description of a response function
$\alpha:S\rightarrow A$\ that is claimed to achieve (say) $W_{\alpha}\geq2/3$.
\ The question is how to implement the rest of the algorithm---or
equivalently, how Arthur can verify that $W_{\alpha}$\ is large using an
$\mathsf{NP}$\ oracle.

Here the key idea is to use the property-testing paradigm of Goldreich,
Goldwasser, and Ron \cite{ggr}. \ As it stands, the inner loop of
\texttt{REst}\ requires\ first computing Merlin$_{2}$'s optimal
response\ $b_{\alpha}:Y\rightarrow B$ to $\alpha$, then computing Merlin$_{1}%
$'s optimal response\ $a_{\alpha}:X\rightarrow A$ to $b_{\alpha}$, and finally
computing the value $W_{\alpha}$\ achieved by the pair $\left(  a_{\alpha
},b_{\alpha}\right)  $. \ In our case, all three of these steps would operate
on $2^{p\left(  n\right)  }$-sized objects and require $2^{O\left(  p\left(
n\right)  \right)  }$\ time.

However, by using a GGR-like approach, we can replace all three of these
exponential-time computations by polynomial-time random sampling combined with
$\mathsf{NP}$\ oracle calls. \ In more detail: given $\alpha$, Arthur first
chooses a subset $T\subseteq Y$\ of size $\ell=\operatorname*{poly}\left(
n\right)  $\ uniformly at random. \ He then uses his $\mathsf{NP}$\ oracle\ to
find a response function $\beta:T\rightarrow B$\ that maximizes%
\begin{equation}
\operatorname*{E}_{x\in S,y\in T}\left[  V\left(  x,y,\alpha\left(  x\right)
,\beta\left(  y\right)  \right)  \right]  .
\end{equation}
Next, Arthur chooses \textit{another} subset $U\subseteq X$\ of size
$m=\operatorname*{poly}\left(  n\right)  $\ uniformly at random, and again
uses his $\mathsf{NP}$\ oracle\ to find a response function $\gamma
:U\rightarrow A$\ that maximizes%
\begin{equation}
W_{\gamma}:=\operatorname*{E}_{x\in U,y\in T}\left[  V\left(  x,y,\gamma
\left(  x\right)  ,\beta\left(  y\right)  \right)  \right]  .
\end{equation}
Finally, Arthur accepts if and only if $\max_{\gamma}W_{\gamma}\geq1/2$.

If $\omega\left(  G\right)  \geq2/3$, then certainly Merlin can cause Arthur
to accept with high probability in this protocol---for example, by sending
Arthur $\alpha=\alpha^{\ast}$, the restriction of the globally optimal
strategy $a^{\ast}:X\rightarrow A$ to the subset $S$. \ As we saw in the proof
of Theorem \ref{freegamealg}, the Hoeffding inequality and union bound ensure
that $\alpha^{\ast}$\ \textquotedblleft induces\textquotedblright\ responses
$\beta:T\rightarrow B$\ by Merlin$_{2}$ that are close to the best possible
responses in the full game $G$. \ Furthermore, even if $T$\ is only an
$O(\frac{1}{\varepsilon^{2}}\log\left(  \left\vert X\right\vert \left\vert
A\right\vert \right)  )$-sized subset of the full set $Y$, a \textit{second}
application of the Hoeffding inequality and union bound ensure that $\beta$,
in turn, induces responses $\gamma:U\rightarrow A$\ by Merlin$_{1}$\ that are
close to the best possible responses. \ So with high probability over the
choices of $S$, $T$, and $U$, the optimal response functions $\beta
:T\rightarrow B$\ and $\gamma:U\rightarrow A$\ will achieve a value of
$W_{\gamma}$\ close to $\omega\left(  G\right)  $.

As usual, the more interesting part is soundness: if $\omega\left(  G\right)
\leq1/3$, then why can Merlin \textit{not} cause Arthur to accept with high
probability? \ The basic answer is that Merlin has to provide $\alpha
$\ without knowing $T$\ or $U$ (which Arthur will only choose later), and
without being able to control $\beta$\ or $\gamma$\ (which are both just
solutions to maximization problems, obtained using the $\mathsf{NP}$\ oracle).
\ As a consequence, one can show that, if $\max_{\gamma}W_{\gamma}\geq
1/2$\ with high probability over $S$, $T$, and $U$, then one can construct a
\textit{global} strategy pair\ $a:X\rightarrow A$, $b:Y\rightarrow B$ that
achieves value close to $1/2$. \ We omit the details, which closely follow
those in the correctness proofs for property-testing algorithms due to
Goldreich, Goldwasser, and Ron \cite{ggr}.
\end{proof}

\subsection{Subsampling for Free Games and $\mathsf{AM}\left(  2\right)
=\mathsf{AM}$\label{SUBSAMPLING}}

In this section, we wish to go further than $\mathsf{AM}\left(  2\right)
\subseteq\mathsf{EXP}$\ or\ $\mathsf{AM}\left(  2\right)  \subseteq
\mathsf{AM}^{\mathsf{NP}}$, and prove that actually $\mathsf{AM}\left(
2\right)  =\mathsf{AM}$. \ For this purpose, given a free game $G=\left(
X,Y,A,B,V\right)  $, we need to show not merely that a near-optimal pair of
strategies for $G$ can be \textquotedblleft induced\textquotedblright\ by
examining a small random subgame $G_{S}$, but that $\omega\left(
G_{S}\right)  $\ \textit{itself} gives a good approximation to $\omega\left(
G\right)  $, with high probability over $S$. \ As we explained in Section
\ref{ALGSEC}, it is easy to see that $\operatorname{E}_{S}\left[
\omega\left(  G_{S}\right)  \right]  \geq\omega\left(  G\right)  $, since we
can start with optimal strategies for $G$ and then restrict them to $G_{S}$.
\ The hard part is to prove the other direction, that $\operatorname{E}%
_{S}\left[  \omega\left(  G_{S}\right)  \right]  \leq\omega\left(  G\right)
+\varepsilon$.

Here it is convenient to appeal to a powerful recent result of Barak et
al.\ \cite{bhhs}, which shows that \textit{any} dense CSP over a finite
alphabet $\Sigma$ can be \textquotedblleft subsampled,\textquotedblright%
\ generalizing an earlier subsampling result for the Boolean case by Alon et
al.\ \cite{avkk}.

\begin{theorem}
[Subsampling of Dense CSPs\ \cite{bhhs}]\label{bhhsthm}Let $\varphi$ be a
$k$-CSP, involving $n$ variables $X=\left(  x_{1},\ldots,x_{n}\right)  $\ over
the finite alphabet $\Sigma$. \ Suppose that $\varphi$\ has \textquotedblleft
density\textquotedblright\ $\alpha$, in the following sense: for every
collection $Y\subset X$ of $k-1$\ variables, $\varphi$ contains a $\left[
0,1\right]  $-valued constraint $C$\ involving the variables $Y\cup\left\{
x_{i}\right\}  $, for at least an $\alpha$\ fraction of the remaining
variables $x_{i}\in X\setminus Y$. \ Let $\operatorname*{SAT}\left(
\varphi\right)  \in\left[  0,1\right]  $ be the value of $\varphi$; that is,
the maximum of $\operatorname*{E}_{C\in\varphi}\left[  C\left(  X\right)
\right]  $\ over all $X\in\Sigma^{n}$. \ Also, given a subset of variable
indices $I\subseteq\left[  n\right]  $, let $\varphi_{I}$\ be the restriction
of $\varphi$\ to the variables in $I$\ (and to those constraints that only
involve $I$ variables). \ Then provided we choose $I$ uniformly at random
subject to $\left\vert I\right\vert \geq\frac{\log\left\vert \Sigma\right\vert
}{\alpha\varepsilon^{\Lambda}}$ for some suitable constant $\Lambda$, we have%
\begin{equation}
\operatorname*{E}_{I}\left[  \operatorname*{SAT}\left(  \varphi_{I}\right)
\right]  \leq\operatorname*{SAT}\left(  \varphi\right)  +\varepsilon.
\end{equation}

\end{theorem}

As a side note, if the alphabet size $\left\vert \Sigma\right\vert $\ is
constant, and if one does not care about the dependence of $\left\vert
I\right\vert $\ on $\varepsilon$, then a version of Theorem \ref{bhhsthm}%
\ follows almost immediately from the Szemer\'{e}di Regularity Lemma, in its
many-colored variant (see for example \cite[Theorem 1.18]{komlossimon}).
\ However, this is of limited relevance to us, since in our case $\left\vert
\Sigma\right\vert =\operatorname*{poly}\left(  n\right)  $.

We now use Theorem \ref{bhhsthm} to deduce an analogous subsampling theorem
for free games.

\begin{theorem}
[Subsampling of Free Games]\label{subfg}Given a free game $G=\left(
X,Y,A,B,V\right)  $ and $\varepsilon>0$, let $\kappa:=2\varepsilon^{-\Lambda
}\log\left(  \left\vert A\right\vert +\left\vert B\right\vert \right)  $\ (for
some suitable constant $\Lambda$), and assume $\kappa\leq\left\vert
X\right\vert $. \ Choose a subset $S\subseteq X$ of Merlin$_{1}$ questions
uniformly at random subject to $\left\vert S\right\vert =\kappa$, and let
$G_{S}$\ be the subgame of $G$ with Merlin$_{1}$'s questions restricted to
$S$. \ Then%
\begin{equation}
\operatorname*{E}_{S}\left[  \omega\left(  G_{S}\right)  \right]  \leq
\omega\left(  G\right)  +\varepsilon.
\end{equation}

\end{theorem}

\begin{proof}
We define a $2$-CSP $\varphi$\ as follows. \ Let $X^{\prime}:=X\times R_{1}%
$\ and $Y^{\prime}:=Y\times R_{2}$, where $R_{1}$\ and $R_{2}$\ are finite
sets chosen to ensure that $\left\vert X\right\vert \left\vert R_{1}%
\right\vert =\left\vert Y\right\vert \left\vert R_{2}\right\vert $. \ We think
of $X^{\prime}$\ and $Y^{\prime}$\ as \textquotedblleft
augmented\textquotedblright\ versions of $X$ and $Y$ respectively, obtained by
duplicating variables. \ Then $\varphi$\ will have a variable set consisting
of $a\left(  x,r_{1}\right)  $\ for all $\left(  x,r_{1}\right)  \in
X^{\prime}$\ and $b\left(  y,r_{2}\right)  $\ for all $\left(  y,r_{2}\right)
\in Y^{\prime}$,\ and alphabet $\Sigma:=A\cup B$. \ For each $\left(
x,r_{1}\right)  \in X^{\prime}$ and $\left(  y,r_{2}\right)  \in Y^{\prime}$,
we will add a $\left[  0,1\right]  $-valued constraint $C$\ between $a\left(
x,r_{1}\right)  $\ and $b\left(  y,r_{2}\right)  $, which behaves as follows:

\begin{itemize}
\item If $a\in A$ and $b\in B$, then $C\left(  a,b\right)  =V\left(
x,y,a,b\right)  $.

\item If $a\notin A$\ or $b\notin B$, then $C\left(  a,b\right)  =0$.
\end{itemize}

In this way, we ensure the following two properties:

\begin{enumerate}
\item[(1)] $\operatorname*{SAT}\left(  \varphi\right)  =\omega\left(
G\right)  $. \ Indeed, from any strategy pair $\left(  a,b\right)  $\ that
achieves value $\omega$\ in $G$,\ we can construct an assignment to $\varphi
$\ with value at least $\omega$, and conversely. \ (To see the converse, note
that by convexity, if some $a\left(  x,r_{1}\right)  $\ takes on more than one
value as we range over $r_{1}\in R_{1}$, then there must be a single value
$a\left(  x,r_{1}\right)  =a\left(  x\right)  $ that does at least as well as
the mean, and likewise for $b\left(  y,r_{2}\right)  $.)

\item[(2)] $\varphi$ has density $\alpha=1/2$\ in the sense of Theorem
\ref{bhhsthm}. \ For it includes a constraint between every $a$-variable\ and
every $b$-variable\ (although no constraints relating two $a$-variables or two
$b$-variables), and the numbers of $a$-variables\ and $b$-variables\ are equal.
\end{enumerate}

Now suppose we choose a subset $I$\ of the variables of $\varphi$\ uniformly
at random, subject to $\left\vert I\right\vert =\kappa$\ where $\kappa
=2\varepsilon^{-\Lambda}\log\left(  \left\vert A\right\vert +\left\vert
B\right\vert \right)  $. \ Then by Theorem \ref{bhhsthm}, we have%
\begin{equation}
\operatorname*{E}_{I}\left[  \operatorname*{SAT}\left(  \varphi_{I}\right)
\right]  \leq\operatorname*{SAT}\left(  \varphi\right)  +\varepsilon
=\omega\left(  G\right)  +\varepsilon.
\end{equation}
To complete the proof, we need to show that%
\begin{equation}
\operatorname*{E}_{S}\left[  \omega\left(  G_{S}\right)  \right]
\leq\operatorname*{E}_{I}\left[  \operatorname*{SAT}\left(  \varphi
_{I}\right)  \right]  .
\end{equation}
We will do so by appealing to the following general principle. \ Suppose we
were to draw $I$ by some random process in which we started with the set of
all variables in $\varphi$, then repeatedly discarded variables until we were
left with a uniformly-random subset $I$ of size $\kappa$. \ Suppose, further,
that at any point in this process, the distribution over constraints between
remaining variables remained uniform over the set of \textit{all} constraints
$C\in\varphi$. \ Let $J$\ be a subset of variables obtained by stopping such a
process at any intermediate point. \ Then we must have%
\begin{equation}
\operatorname*{E}_{J}\left[  \operatorname*{SAT}\left(  \varphi_{J}\right)
\right]  \leq\operatorname*{E}_{I}\left[  \operatorname*{SAT}\left(
\varphi_{I}\right)  \right]  .
\end{equation}
The reason is simply that, if we had a collection of partial assignments to
the $\varphi_{J}$'s that achieved expected value $\omega$, then restricting
those assignments to the $\varphi_{I}$'s\ would also achieve expected value
$\omega$, by linearity of expectation.

So in particular, suppose we form $J$ by choosing discarding all but $\kappa
$\ variables of the form $a\left(  x,r_{1}\right)  $, (while keeping all
variables of the form $b\left(  y,r_{2}\right)  $). \ Then since the
distribution over constraints remains uniform for this $J$, and since a
sequence of further discardings could produce a uniformly-random subset
$I$\ with $\left\vert I\right\vert =\kappa$, we have%
\begin{equation}
\operatorname*{E}_{J}\left[  \operatorname*{SAT}\left(  \varphi_{J}\right)
\right]  \leq\operatorname*{E}_{I}\left[  \operatorname*{SAT}\left(
\varphi_{I}\right)  \right]  \leq\omega\left(  G\right)  +\varepsilon.
\end{equation}
But $\operatorname*{E}_{J}\left[  \operatorname*{SAT}\left(  \varphi
_{J}\right)  \right]  $\ is simply $\operatorname*{E}_{S}\left[  \omega\left(
G_{S}\right)  \right]  $, where $S\subseteq X$\ is a uniformly-random subset
of Merlin$_{1}$\ questions of size $\kappa$. \ This completes the proof.
\end{proof}

Theorem \ref{subfg} has the following easy corollary, which removes the
\textquotedblleft asymmetry\textquotedblright\ between the two Merlins.

\begin{corollary}
\label{subfgcor}Given a free game $G=\left(  X,Y,A,B,V\right)  $ and
$\varepsilon>0$, let $\kappa:=2\varepsilon^{-\Lambda}\log\left(  \left\vert
A\right\vert +\left\vert B\right\vert \right)  $, and assume $\kappa\leq
\min\left\{  \left\vert X\right\vert ,\left\vert Y\right\vert \right\}  $.
\ Choose $S\subseteq X$ and $T\subseteq Y$\ uniformly at random and
independently, subject to $\left\vert S\right\vert =\left\vert T\right\vert
=\kappa$. \ Also, let $G_{S,T}$\ be the subgame of $G$ with Merlin$_{1}$'s
questions restricted to $S$ and Merlin$_{2}$'s restricted to $T$. \ Then%
\begin{equation}
\operatorname*{E}_{S,T}\left[  \omega\left(  G_{S,T}\right)  \right]
\leq\omega\left(  G\right)  +2\varepsilon.
\end{equation}

\end{corollary}

\begin{proof}
We simply need to apply Theorem \ref{subfg}\ twice in succession, once to
reduce Merlin$_{1}$'s question set, and then a second time to reduce
Merlin$_{2}$'s. \ The result follows by linearity of expectation.
\end{proof}

Using Corollary \ref{subfgcor}, we now prove that $\mathsf{AM}\left(
2\right)  =\mathsf{AM}$.

\begin{theorem}
\label{am2am}$\mathsf{AM}\left(  2\right)  =\mathsf{AM}$.
\end{theorem}

\begin{proof}
Let $L\in\mathsf{AM}\left(  2\right)  $. \ Then just like in Corollary
\ref{amexpcor}, an $\mathsf{AM}\left(  2\right)  $\ protocol for checking
whether a string is in $L$\ can be represented as a free game $G=\left(
X,Y,A,B,V\right)  $, where $X=Y=A=B=\left\{  0,1\right\}  ^{p\left(  n\right)
}$ for some polynomial $p$, and $V$\ is computable in randomized
$\operatorname*{poly}\left(  n\right)  $\ time.

Let $\varepsilon:=1/24$\ and $\kappa:=2\varepsilon^{-\Lambda}\left(  p\left(
n\right)  +1\right)  $, and suppose we choose $S\subseteq X$ and $T\subseteq
Y$ uniformly at random subject to $\left\vert S\right\vert =\left\vert
T\right\vert =\kappa$.\ \ Then by Corollary \ref{subfgcor},%
\begin{equation}
\operatorname*{E}_{S,T}\left[  \omega\left(  G_{S,T}\right)  \right]
\leq\omega\left(  G\right)  +\frac{1}{12}. \label{satsat}%
\end{equation}
But this immediately gives us our $\mathsf{AM}$\ simulation, as follows.
\ First Arthur chooses $S,T\subseteq\left\{  0,1\right\}  ^{p\left(  n\right)
}$ uniformly at random, subject to $\left\vert S\right\vert =\left\vert
T\right\vert =\kappa$\ as above. \ He then sends $S$ and $T$ to Merlin, using
$2\kappa\cdot p\left(  n\right)  =O(p\left(  n\right)  ^{2})$ bits. \ Next
Merlin replies with a pair of strategies $a:S\rightarrow A$\ and
$b:T\rightarrow B$, which again takes $O(p\left(  n\right)  ^{2})$\ bits.
\ Let%
\begin{equation}
\omega_{S,T}:=\operatorname*{E}_{x\in S,y\in T}\left[  V\left(  x,y,a\left(
x\right)  ,b\left(  y\right)  \right)  \right]
\end{equation}
be the subsampled success probability; notice that $\omega_{S,T}\leq
\omega\left(  G_{S,T}\right)  $ for all $S,T$. \ Then finally, if $V$ is
deterministic, then Arthur simply computes $\omega_{S,T}$\ and accepts if and
only if $\omega_{S,T}\geq1/2$. \ If $V$ is randomized, then Arthur instead
computes an estimate $\widetilde{\omega}_{S,T}$\ such that%
\begin{equation}
\Pr\left[  \left\vert \widetilde{\omega}_{S,T}-\omega_{S,T}\right\vert
>0.01\right]  \leq\exp\left(  -\kappa\right)  ,
\end{equation}
and accepts if and only if $\widetilde{\omega}_{S,T}\geq0.51$.

We claim, first, that this protocol has completeness error at most
$\exp\left(  -\kappa\right)  $. \ For we can always consider an optimal pair
of strategies $a:X\rightarrow A$\ and $b:Y\rightarrow B$\ for the full
protocol, which achieve value%
\begin{equation}
\omega:=\operatorname*{E}_{x\in X,y\in Y}\left[  V\left(  x,y,a\left(
x\right)  ,b\left(  y\right)  \right)  \right]  =\operatorname*{E}%
_{S,T}\left[  \omega_{S,T}\right]  \geq\frac{2}{3}%
\end{equation}
by assumption. \ Then a standard Chernoff bound implies that $\omega_{S,T}%
\geq0.52$, and hence $\widetilde{\omega}_{S,T}\geq0.51$, with at least
$1-\exp\left(  -\kappa\right)  $\ probability over the choice of $S$ and $T$.

We next upper-bound the soundness error. \ Suppose $\omega\left(  G\right)
\leq1/3$; then by equation (\ref{satsat}),%
\begin{equation}
\operatorname*{E}_{S,T}\left[  \omega_{S,T}\right]  \leq\operatorname*{E}%
_{S,T}\left[  \omega\left(  G_{S,T}\right)  \right]  \leq\omega\left(
G\right)  +\frac{1}{12}\leq\frac{5}{12}.
\end{equation}
So by Markov's inequality,%
\begin{equation}
\Pr_{S,T}\left[  \omega_{S,T}\geq\frac{1}{2}\right]  \leq\frac{5/12}%
{1/2}=\frac{5}{6},
\end{equation}
and hence%
\begin{equation}
\Pr\left[  \widetilde{\omega}_{S,T}\geq0.51\right]  \leq\frac{5}{6}%
+\exp\left(  -\kappa\right)
\end{equation}
as well. \ So Arthur rejects with constant probability. \ Of course, we can
amplify the completeness/soundness gap further by repeating the protocol.
\end{proof}

\subsection{The $k$-Merlin Case\label{KMERLIN}}

In this section, we generalize our results from $\mathsf{AM}\left(  2\right)
$\ to $\mathsf{AM}\left(  k\right)  $ for larger $k$. \ The first step is to
generalize Theorem \ref{freegamealg}, to obtain a nontrivial approximation
algorithm for $k$-player free games.

\begin{theorem}
\label{kmerlinalg}Let $G$ be a $k$-player free game, with question sets
$Y_{1},\ldots,Y_{k}$\ and answer sets $B_{1},\ldots,B_{k}$ (assume $\left\vert
B_{i}\right\vert \geq2$\ for all $i\in\left[  k\right]  $). \ There exists a
deterministic algorithm that approximates $\omega\left(  G\right)  $\ to
within additive error $\pm\varepsilon$, in time%
\begin{equation}
\exp\left(  \frac{k^{2}}{\varepsilon^{2}}\sum_{i<j}\log\left(  \left\vert
Y_{i}\right\vert \left\vert B_{i}\right\vert \right)  \cdot\log\left(
\left\vert Y_{j}\right\vert \left\vert B_{j}\right\vert \right)  \right)
=n^{O(\varepsilon^{-2}k^{2}\log n)},
\end{equation}
where $n=\left\vert Y_{1}\right\vert \left\vert B_{1}\right\vert
\cdots\left\vert Y_{k}\right\vert \left\vert B_{k}\right\vert $\ is the input size.
\end{theorem}

\begin{proof}
The basic idea is to use a recursive generalization, call it \texttt{Est}%
$_{k}$, of the (deterministic) approximation algorithm \texttt{Est}\ from
Theorem \ref{freegamealg}. \ The recursive version will \textquotedblleft peel
off the Merlins one at a time.\textquotedblright\ \ That is, given a
description of a $k$-player free game $G$ as input, \texttt{Est}$_{k}$\ will
reduce the estimation of $\omega\left(  G\right)  $\ to the estimation of
$\omega(G^{\prime})$, for a quasipolynomial number of $\left(  k-1\right)
$-player free games $G^{\prime}$, each one involving Merlin$_{1}$\ through
Merlin$_{k-1}$ only (Merlin$_{k}$'s behavior having already been fixed).
\ Each $\omega(G^{\prime})$\ will in turn be estimated by calling
\texttt{Est}$_{k-1}$, and so on until $k=1$, at which point we can just do a
straightforward maximization.

In more detail, let $\delta:=\varepsilon/k$. \ Then for each $\ell\in\left\{
2,\ldots,k\right\}  $, let%
\begin{equation}
\kappa_{\ell}:=\frac{C}{\delta^{2}}\sum_{i=1}^{\ell-1}\log\left(  \left\vert
Y_{i}\right\vert \left\vert B_{i}\right\vert \right)  ,
\end{equation}
for some suitable constant $C$. \ Then \texttt{Est}$_{k}$\ loops over all
$\binom{\left\vert Y_{k}\right\vert }{\kappa_{k}}$ subsets of questions
$S_{k}\subseteq Y_{k}$ such that $\left\vert S_{k}\right\vert =\kappa_{k}%
$,\ as well as all $\left\vert B_{k}\right\vert ^{\kappa_{k}}$\ possible
settings $\alpha_{k}:S_{k}\rightarrow B_{k}$\ of the answers to the
$\kappa_{k}$\ questions in $S_{k}$. \ For each such pair $P=\left(
S_{k},\alpha_{k}\right)  $, we define a $\left(  k-1\right)  $-player\ subgame
$G_{P}$, which is played by Merlin$_{1}$ through Merlin$_{k-1}$, and which has
question sets $Y_{1},\ldots,Y_{k-1}$\ and answer sets $B_{1},\ldots,B_{k-1}$.
\ The verification function of $G_{P}$\ is defined as follows:%
\begin{equation}
V_{P}\left(  y_{1},\ldots,y_{k-1},b_{1},\ldots,b_{k-1}\right)
:=\operatorname*{E}_{y_{k}\in S_{k}}\left[  V\left(  y_{1},\ldots,y_{k}%
,b_{1},\ldots,b_{k-1},\alpha_{k}\left(  y_{k}\right)  \right)  \right]  .
\end{equation}
In other words, $G_{P}$\ is the same game as $G$, except that we assume that
Merlin$_{k}$ is only asked questions $y_{k}\in S_{k}$, and that he responds to
each with $\alpha_{k}\left(  y_{k}\right)  $.

Now, for each $P$, the algorithm \texttt{Est}$_{k}$\ does the following:

\begin{enumerate}
\item[(1)] If $k\geq3$, then it calls \texttt{Est}$_{k-1}$\ recursively, in
order to find approximately optimal strategies $\left(  b_{P,i}:Y_{i}%
\rightarrow B_{i}\right)  _{i\in\left[  k-1\right]  }$\ for Merlin$_{1}$
through Merlin$_{k-1}$ in $G_{P}$. \ Here \textquotedblleft approximately
optimal\textquotedblright\ means achieving value at least $\omega
(G_{P})-\delta$. \ Of course, when $k=2$, the algorithm can simply compute
Merlin$_{1}$'s exactly-optimal response $b_{P,1}:Y_{1}\rightarrow B_{1}$\ by
maximizing%
\begin{equation}
\operatorname*{E}_{y_{2}\in S_{2}}\left[  V\left(  y_{1},y_{2},b_{P,1}\left(
y_{1}\right)  ,\alpha_{2}\left(  y_{2}\right)  \right)  \right]
\end{equation}
for each $y_{1}\in Y_{1}$\ separately, just like in the two-player algorithm
\texttt{Est}.

\item[(2)] Given the responses $b_{P,1},\ldots,b_{P,k-1}$ of Merlin$_{1}$
through Merlin$_{k-1}$, the algorithm computes Merlin$_{k}$'s best response
$b_{P,k}:Y_{k}\rightarrow B_{k}$\ on the full set $Y_{k}$\ by maximizing%
\begin{equation}
\operatorname*{E}_{y_{1}\in Y_{1},\ldots,y_{k-1}\in Y_{k-1}}\left[  V\left(
y_{1},\ldots,y_{k},b_{P,1}\left(  y_{1}\right)  ,\ldots,b_{P,k}\left(
y_{k}\right)  \right)  \right]
\end{equation}
for each $y_{k}\in Y_{k}$\ separately. \ It then lets%
\begin{equation}
W_{P}:=\operatorname*{E}_{y_{1}\in Y_{1},\ldots,y_{k}\in Y_{k}}\left[
V\left(  y_{1},\ldots,y_{k},b_{P,1}\left(  y_{1}\right)  ,\ldots
,b_{P,k}\left(  y_{k}\right)  \right)  \right]
\end{equation}
be the value of the $k$-tuple of strategies induced by $P$.
\end{enumerate}

Finally, \texttt{Est}$_{k}$\ outputs $W:=\max_{P}W_{P}$\ as its estimate for
$\omega\left(  G\right)  $. \ Note that, in addition to $W$, the algorithm
also outputs a strategy $k$-tuple $\left(  b_{P,1},\ldots,b_{P,k}\right)  $
that achieves value $W$.

Let $T\left(  \ell\right)  $\ be the number of evaluations of the
\textquotedblleft original\textquotedblright\ verification function $V\left(
y_{1},\ldots,y_{k},b_{1},\ldots,b_{k}\right)  $\ that \texttt{Est}$_{\ell}$
needs to make, when it's called on an $\ell$-player game involving
Merlin$_{1}$\ through Merlin$_{\ell}$. \ Then we have the following recurrence
relation:%
\begin{equation}
T\left(  \ell\right)  \leq\binom{\left\vert Y_{\ell}\right\vert }{\kappa
_{\ell}}\left\vert B_{\ell}\right\vert ^{\kappa_{\ell}}\left(  T\left(
\ell-1\right)  +\left\vert Y_{1}\right\vert \cdots\left\vert Y_{\ell
-1}\right\vert \cdot\left\vert Y_{\ell}\right\vert \left\vert B_{\ell
}\right\vert \cdot\kappa_{\ell+1}\cdots\kappa_{k}\right)  , \label{tl}%
\end{equation}
with base case $T\left(  1\right)  =\left\vert Y_{1}\right\vert \left\vert
B_{1}\right\vert \cdot\kappa_{2}\cdots\kappa_{k}$. \ (The reason for the
factor of $\kappa_{\ell+1}\cdots\kappa_{k}$\ is that, just to compute $V$ for
a game involving Merlin$_{1}$ through Merlin$_{\ell}$, one needs to take an
expectation over all $y_{\ell+1}\in S_{\ell+1},\ldots,y_{k}\in S_{k}$.) \ Now,
it is not hard to see that the%
\begin{equation}
\left\vert Y_{1}\right\vert \cdots\left\vert Y_{\ell-1}\right\vert
\cdot\left\vert Y_{\ell}\right\vert \left\vert B_{\ell}\right\vert \cdot
\kappa_{\ell+1}\cdots\kappa_{k}%
\end{equation}
terms all get absorbed by asymptotically larger terms. \ Asymptotically, then,%
\begin{align}
T\left(  k\right)   &  \leq\left(  \left\vert Y_{k}\right\vert \left\vert
B_{k}\right\vert \right)  ^{\kappa_{k}}T\left(  k-1\right) \\
&  =\exp\left(  \log\left(  \left\vert Y_{k}\right\vert \left\vert
B_{k}\right\vert \right)  \cdot\frac{C}{\delta^{2}}\sum_{i=1}^{k-1}\log\left(
\left\vert Y_{i}\right\vert \left\vert B_{i}\right\vert \right)  \right)
\cdot T\left(  k-1\right) \\
&  =\exp\left(  \frac{k^{2}}{\varepsilon^{2}}\sum_{i<j}\log\left(  \left\vert
Y_{i}\right\vert \left\vert B_{i}\right\vert \right)  \cdot\log\left(
\left\vert Y_{j}\right\vert \left\vert B_{j}\right\vert \right)  \right)  .
\end{align}
Since the running time is dominated by evaluations of $V$ (each of which takes
constant time), this also gives the asymptotic running time.

The proof of correctness for \texttt{Est}$_{k}$ follows the same general
outline as the proof of the correctness for \texttt{Est}. \ Once again, since
each $W_{P}$\ is the value achieved by some actual $k$-tuple of strategies
$b_{P,1},\ldots,b_{P,k}$ in the full game $G$, it is clear that $W_{P}%
\leq\omega\left(  G\right)  $ for all $P$. \ The nontrivial part is to show
that $W_{P}\geq\omega\left(  G\right)  -\varepsilon$\ for \textit{some}
$P=\left(  S_{k},\alpha_{k}\right)  $.

We will prove this claim by induction on $\ell$. \ That is, suppose by
induction that, for every $\left(  \ell-1\right)  $-player game $G_{P}%
$\ played by Merlin$_{1}$ through Merlin$_{\ell-1}$, the algorithm
\texttt{Est}$_{\ell-1}$\ finds an $\left(  \ell-1\right)  $-tuple of
strategies that achieve a value at least $\omega\left(  G_{P}\right)
-\epsilon$. \ We will show that this implies that, for every $\ell$-player
game $G_{Q}$\ played by Merlin$_{1}$ through Merlin$_{\ell}$, the algorithm
\texttt{Est}$_{\ell}$\ achieves a value at least $\omega\left(  G_{Q}\right)
-\epsilon-\delta$. \ Since $\delta=\varepsilon/k$, clearly this suffices to
prove the claim.

Let $G_{Q}$\ be the $\ell$-player game defined by the tuple $Q=\left(
S_{\ell+1},\ldots,S_{k},\alpha_{\ell+1},\ldots,\alpha_{k}\right)  $. \ Then
$G_{Q}$\ has the verification function%
\begin{equation}
V_{Q}\left(  y_{1},\ldots,y_{\ell},b_{1},\ldots,b_{\ell}\right)
:=\operatorname*{E}_{y_{\ell+1}\in S_{\ell+1},\ldots,y_{k}\in S_{k}}\left[
V\left(  y_{1},\ldots,y_{k},b_{1},\ldots,b_{\ell},\alpha_{\ell+1}\left(
y_{\ell+1}\right)  ,\ldots,\alpha_{k}\left(  y_{k}\right)  \right)  \right]  .
\end{equation}
By definition, there exists an $\ell$-tuple of strategies $\left(  b_{i}%
^{\ast}:Y_{i}\rightarrow B_{i}\right)  _{i\in\left[  \ell\right]  }$\ such
that%
\begin{equation}
\operatorname*{E}_{y_{1}\in Y_{1},\ldots,y_{\ell}\in S_{\ell}}\left[
V_{Q}\left(  y_{1},\ldots,y_{\ell},b_{1}^{\ast}\left(  y_{1}\right)
,\ldots,b_{\ell}^{\ast}\left(  y_{\ell}\right)  \right)  \right]
=\omega\left(  G_{Q}\right)  .
\end{equation}
Given a subset $S_{\ell}\subseteq Y_{\ell}$ with $\left\vert S_{\ell
}\right\vert =\kappa_{\ell}$, call $S_{\ell}$\ \textquotedblleft
good\textquotedblright\ if it has the property that%
\begin{equation}
\left\vert \operatorname*{E}_{y_{\ell}\in S_{\ell}}\left[  V_{Q}\left(
y_{1},\ldots,y_{\ell},b_{1},\ldots,b_{\ell-1},b_{\ell}^{\ast}\left(  y_{\ell
}\right)  \right)  \right]  -\operatorname*{E}_{y_{\ell}\in Y_{\ell}}\left[
V_{Q}\left(  y_{1},\ldots,y_{\ell},b_{1},\ldots,b_{\ell-1},b_{\ell}^{\ast
}\left(  y_{\ell}\right)  \right)  \right]  \right\vert \leq\frac{\delta}{2}%
\end{equation}
for \textit{every} $\left(  \ell-1\right)  $-tuple of questions $\left(
y_{1},\ldots,y_{\ell-1}\right)  \in Y_{1}\times\cdots\times Y_{\ell-1}$\ and
answers $\left(  b_{1},\ldots,b_{\ell-1}\right)  \in B_{1}\times\cdots\times
B_{\ell-1}$. \ Then a straightforward application of the Hoeffding inequality
and union bound shows that the fraction of $S_{\ell}$'s that are good is at
least%
\begin{equation}
1-2e^{-\delta^{2}\kappa_{\ell}}\left\vert Y_{1}\right\vert \left\vert
B_{1}\right\vert \cdots\left\vert Y_{\ell-1}\right\vert \left\vert B_{\ell
-1}\right\vert =1-2\exp\left(  -C\sum_{i=1}^{\ell-1}\log\left(  \left\vert
Y_{i}\right\vert \left\vert B_{i}\right\vert \right)  \right)  \left\vert
Y_{1}\right\vert \left\vert B_{1}\right\vert \cdots\left\vert Y_{\ell
-1}\right\vert \left\vert B_{\ell-1}\right\vert \geq\frac{2}{3}%
\end{equation}
for suitable $C$. \ Thus, certainly there \textit{exists} a good $S_{\ell}$,
and \texttt{Est}$_{\ell}$\ will find one when it loops over all possibilities.
\ Fix a good $S_{\ell}$\ in what follows.

Let $G_{P}$\ be the $\left(  \ell-1\right)  $-player game played by
Merlin$_{1}$ through Merlin$_{\ell-1}$, which is obtained from $G_{Q}$\ by
restricting Merlin$_{\ell}$'s\ question set to $S_{\ell}$, \textit{and}
restricting Merlin$_{\ell}$'s\ strategy to $b_{\ell}^{\ast}$. \ Then notice
that $S_{\ell}$\ being good\ has the following two consequences:

\begin{enumerate}
\item[(i)] We can achieve value at least $\omega\left(  G_{Q}\right)
-\delta/2$\ in $G_{P}$, by simply starting with $b_{1}^{\ast},\ldots,b_{\ell
}^{\ast}$\ and then restricting $b_{\ell}^{\ast}$\ to $S_{\ell}$.

\item[(ii)] Any time we find strategies $b_{1},\ldots,b_{\ell-1}$\ that
achieve value at least $W$ in\ $G_{P}$, we have also found strategies that
achieve value at least $W-\delta/2$\ in $G_{Q}$: we simply need to fix
Merlin$_{\ell}$'s\ strategy to be $b_{\ell}^{\ast}$.
\end{enumerate}

Combining facts (i) and (ii), we find that, \textit{if} \texttt{Est}$_{\ell
-1}$ can achieve value at least $\omega\left(  G_{P}\right)  -\epsilon$\ in
$G_{P}$, then \texttt{Est}$_{\ell}$ can achieve value at least $\omega\left(
G_{Q}\right)  -\epsilon-\delta$\ in $G_{Q}$. \ Intuitively, this is because
the errors build up linearly: we incur an error of $\delta/2$\ when switching
from $G_{Q}$\ to $G_{P}$, then an error of $\epsilon$ (by hypothesis) when
running \texttt{Est}$_{\ell-1}$\ to find strategies for $G_{P}$, and finally
another error of\ $\delta/2$\ when switching from $G_{P}$\ back to $G_{Q}$.
\ This completes the induction, and hence the proof that $W\geq\omega\left(
G\right)  -\varepsilon$.
\end{proof}

Just like in the $k=2$\ case, we can modify the algorithm \texttt{Est}$_{k}%
$\ so that it chooses the sets $S$\ uniformly at random, rather than looping
over all possible $S$'s. \ By doing so, we can get a randomized algorithm that
approximates $\omega\left(  G\right)  $\ to within additive error
$\pm\varepsilon$ in the slightly better running time%
\begin{equation}
\left\vert Y_{k}\right\vert \cdot\exp\left(  \frac{k^{2}}{\varepsilon^{2}}%
\sum_{i<j}\log\left(  \left\vert Y_{i}\right\vert \left\vert B_{i}\right\vert
\right)  \cdot\log\left(  \left\vert B_{j}\right\vert \right)  \right)  .
\end{equation}
We omit the details.

Analogously to Theorem \ref{improvedalg}, we can also improve the running time
of \texttt{Est}$_{k}$\ in the case of perfect completeness.

\begin{theorem}
\label{improvedkmerlin}Given a $k$-player free game $G=\left(  Y_{1}%
,\ldots,Y_{k},B_{1},\ldots,B_{k},V\right)  $, we can decide whether
$\omega\left(  G\right)  =1$\ or $\omega\left(  G\right)  <1-\varepsilon
$\ (promised that one of those is the case) using a deterministic algorithm
that runs in time $n^{O(\varepsilon^{-1}k^{2}\log n)}$, where $n=\left\vert
Y_{1}\right\vert \left\vert B_{1}\right\vert \cdots\left\vert Y_{k}\right\vert
\left\vert B_{k}\right\vert $\ is the input size. \ (In more detail, in both
running time bounds of Theorem \ref{kmerlinalg}, we can improve the factor of
$k^{2}/\varepsilon^{2}$\ in the exponent to $k^{2}/\varepsilon$.)
\end{theorem}

\begin{proof}
[Proof Sketch]As in Theorem \ref{improvedalg}, the key observation is that, if
we only care about distinguishing $\omega\left(  G\right)  =1$\ from
$\omega\left(  G\right)  <1-\varepsilon$, then it suffices to set%
\begin{equation}
\kappa_{\ell}:=\frac{C}{\varepsilon/k^{2}}\sum_{i=1}^{\ell-1}\log\left(
\left\vert Y_{i}\right\vert \left\vert B_{i}\right\vert \right)  .
\end{equation}
The reason is this: we still need to limit the new error introduced at each
level of the recursion to $\delta=\varepsilon/k$. \ However, if $\omega\left(
G\right)  =1$, then the total error will \textit{never} exceed $k\left(
\varepsilon/k\right)  =\varepsilon$, given optimal responses to the question
sets $S_{2},\ldots,S_{k}$ chosen at each level of the recursion, assuming that
$S_{2},\ldots,S_{k}$ are good. \ And it is known that, if a $\left[
0,1\right]  $\ random variable has expectation at most $\varepsilon$, then we
can estimate it to within additive error $\pm\delta$\ with high probability
using only $O(\varepsilon/\delta^{2})$\ samples (see for example
\cite[Appendix 6]{aa:struc}). \ The improved running time bounds follow
directly from the improvement to $\kappa_{\ell}$.
\end{proof}

Theorem \ref{kmerlinalg}\ readily implies an upper bound on $\mathsf{AM}%
\left(  k\right)  $.

\begin{corollary}
\label{amkexpcor}$\mathsf{AM}\left(  k\right)  \subseteq\mathsf{EXP}$ for all
polynomials $k=k\left(  n\right)  $.
\end{corollary}

\begin{proof}
Let $L\in\mathsf{AM}\left(  k\right)  $. \ Then given an input $x\in\left\{
0,1\right\}  ^{n}$, the $\mathsf{AM}\left(  k\right)  $\ protocol for checking
whether $x\in L$\ can be represented as a $k$-player free game $G=(\left(
Y_{i}\right)  _{i\in\left[  k\right]  },\left(  B_{i}\right)  _{i\in\left[
k\right]  },V)$, where $Y_{i}=B_{i}=\left\{  0,1\right\}  ^{p\left(  n\right)
}$ for all $i$ (for some polynomial $p$), and where Arthur's verification
function $V$ is computable in $\operatorname*{poly}\left(  n\right)  $\ time
using $r\left(  n\right)  =\operatorname*{poly}\left(  n\right)  $\ bits of
randomness. \ Now by Theorem \ref{kmerlinalg}, we can estimate $\omega\left(
G\right)  $\ to additive error (say) $\varepsilon=1/10$\ by a deterministic
algorithm that makes%
\begin{equation}
\exp\left(  \frac{k^{2}}{\varepsilon^{2}}\sum_{i<j}p\left(  n\right)
^{2}\right)  =\exp\left(  k^{4}p\left(  n\right)  ^{2}\right)
\end{equation}
evaluations of $V$. \ Furthermore, each $V$\ evaluation can be performed in
deterministic time $2^{r\left(  n\right)  }\operatorname*{poly}\left(
n\right)  $\ (or in randomized time $\operatorname*{poly}\left(  n\right)  $,
even allowing for amplification to exponentially small error probability).
\ But this lets us decide whether $\omega\left(  G\right)  \geq2/3$\ or
$\omega\left(  G\right)  \leq1/3$, and hence whether $x\in L$.
\end{proof}

A second corollary of Theorem \ref{kmerlinalg} is that, assuming the ETH,
there is a hard $\Omega(n^{1/4})$ limit on the amount of communication needed
in any constant-soundness $\mathsf{AM}\left(  k\right)  $\ protocol for
\textsc{3Sat}, regardless of $k=k\left(  n\right)  $. \ Furthermore, if
$k=n^{o\left(  1\right)  }$,\ then $n^{1/2-o\left(  1\right)  }$%
\ communication is needed. \ (Later, in Section \ref{KSUBSAMP},\ we will
improve this to show that $\Omega(\sqrt{n})$\ communication is needed
regardless of $k$.)

\begin{corollary}
\label{3satcor2}Assuming the Randomized ETH, any $\mathsf{AM}\left(  k\right)
$\ protocol for \textsc{3Sat{}{}}\ with a $1$ vs.\ $1-\varepsilon$
completeness/soundness gap\ must use $\Omega(k+\sqrt{\varepsilon n}%
/k)=\Omega(\left(  \varepsilon n\right)  ^{1/4})$\ bits of communication in
total. \ (Also, if Arthur's verification procedure is deterministic, then it
suffices to assume the ordinary ETH.)
\end{corollary}

\begin{proof}
Assume for simplicity that $\varepsilon=1/2$. \ Consider an $\mathsf{AM}%
\left(  k\right)  $\ protocol that uses $q\left(  n\right)  $\ bits of
communication in total. \ We can assume $q\left(  n\right)  \geq k$, since
otherwise we could eliminate one of the Merlins and reduce to the
$\mathsf{AM}\left(  k-1\right)  $\ case. \ Now suppose that for all
$i\in\left[  k\right]  $, Arthur sends an $s_{i}$-bit message to Merlin$_{i}$
and receives a $t_{i}$-bit response. \ Then by Theorem \ref{kmerlinalg}, we
can simulate the protocol to within constant error by an algorithm that makes%
\begin{equation}
\exp\left(  k^{2}\sum_{i<j}\left(  s_{i}+t_{i}\right)  \left(  s_{j}%
+t_{j}\right)  \right)  \leq\exp\left(  \frac{k^{2}}{2}\left(  \sum_{i=1}%
^{k}\left(  s_{i}+t_{i}\right)  \right)  ^{2}\right)  \leq\exp\left(
k^{2}q(n)^{2}\right)
\end{equation}
evaluations of Arthur's verification procedure $V$. \ Furthermore, each
$V$\ evaluation can be performed in randomized $\operatorname*{poly}\left(
n\right)  $\ time (even allowing for amplification to exponentially small
error probability). \ So if \textsc{3Sat{}{}}\ requires $2^{\Omega\left(
n\right)  }$\ randomized time, then $k^{2}q(n)^{2}=\Omega(n)$\ and $q\left(
n\right)  =\Omega(\sqrt{n}/k)$. \ Combining with $q\left(  n\right)  \geq
k$\ then yields $q\left(  n\right)  =\Omega(n^{1/4})$.

For general $\varepsilon>0$, we simply need to use Theorem
\ref{improvedkmerlin}\ rather than Theorem \ref{kmerlinalg}. \ For the last
part, we note that if $V$ is deterministic then so is our \textsc{3Sat{}{}}\ algorithm.
\end{proof}

\subsection{\label{KSUBSAMP}Subsampling with $k$ Merlins}

Finally, let us show that $\mathsf{AM}\left(  k\right)  =\mathsf{AM}$\ for all
$k=\operatorname*{poly}\left(  n\right)  $. \ The first step is to generalize
Theorem \ref{subfg}, the subsampling theorem for $2$-player free games,\ to
$k$ players for arbitrary $k$. \ For technical reasons---related to the
definition of \textquotedblleft denseness\textquotedblright\ in the statement
of Theorem\ \ref{bhhsthm}---doing this will require reducing a free game to a
$k$-CSP\ in a different way than we did in the proof of Theorem \ref{subfg}%
.\footnote{In more detail, suppose we tried to encode a $k$-player free game
$G$\ as a $k$-CSP in the \textquotedblleft obvious\textquotedblright\ way.
Then among all possible $k$-tuples of variables, the fraction that were
related by a nontrivial constraint would decrease like $k!/k^{k}\approx
e^{-k}$, simply because any such $k$-tuple must involve exactly one variable
for each of the $k$ players, with no \textquotedblleft
collisions.\textquotedblright\ \ But this, in turn, would mean that we could
only get the conclusion $\mathsf{AM}\left(  k\right)  =\mathsf{AM}$\ when
$k=O\left(  \log n\right)  $: for larger $k$, our $k$-CSP simply wouldn't be
\textquotedblleft dense\textquotedblright\ enough for Theorem \ref{bhhsthm}%
\ to give what we want. \ To get around this problem, we use a different
encoding of $G$ as a $k$-CSP: one in which every variable, individually,
involves questions to all $k$ of the players.}

\begin{theorem}
[Subsampling of $k$-Player Free Games]\label{subk}Given a $k$-player free game%
\begin{equation}
G=\left(  Y_{1},\ldots,Y_{k},B_{1},\ldots,B_{k},V\right)
\end{equation}
and $\varepsilon>0$, let $\kappa:=\varepsilon^{-\Lambda}\log\left(  \left\vert
B_{1}\right\vert \cdots\left\vert B_{k}\right\vert \right)  $\ (for some
suitable constant $\Lambda$), and assume $\kappa\leq\min\left\{  \left\vert
Y_{1}\right\vert ,\ldots,\left\vert Y_{k}\right\vert \right\}  $. \ For each
$i\in\left[  k\right]  $, choose a subset $S_{i}\subseteq Y_{i}$ of
Merlin$_{i}$ questions uniformly at random subject to $\left\vert
S_{i}\right\vert =\kappa$, let $S:=S_{1}\times\cdots\times S_{k}$, and let
$G_{S}$\ be the subgame of $G$ with Merlin$_{i}$'s questions restricted to
$S_{i}$. \ Then%
\begin{equation}
\operatorname*{E}_{S}\left[  \omega\left(  G_{S}\right)  \right]  \leq
\omega\left(  G\right)  +\varepsilon.
\end{equation}

\end{theorem}

\begin{proof}
We define a $k$-CSP $\varphi$\ as follows. \ Let%
\begin{align}
\mathbf{Y}  &  :=Y_{1}\times\cdots\times Y_{k},\\
\mathbf{B}  &  :=B_{1}\times\cdots\times B_{k}.
\end{align}
Then there is one variable, of the form $\mathbf{b}\left(  \mathbf{y}\right)
\in\mathbf{B}$, for every $k$-tuple $\mathbf{y}\in\mathbf{Y}$. \ (Thus, an
assignment $\mathbf{b}:\mathbf{Y}\rightarrow\mathbf{B}$\ to $\varphi$\ will be
a fairly large object, mapping $k$-tuples of questions to $k$-tuples of
answers.) \ There is also a $\left[  0,1\right]  $-valued constraint,
$C_{\mathbf{R}}$, for every subset $\mathbf{R}=\left\{  \mathbf{y}_{1}%
,\ldots,\mathbf{y}_{k}\right\}  \subseteq\mathbf{Y}$ of size $k$. \ Let
$\left(  \mathbf{y}\right)  _{i}\in Y_{i}$\ denote the $i^{th}$ component of
the $k$-tuple $\mathbf{y}\in\mathbf{Y}$, and likewise let $\left(
\mathbf{b}\right)  _{i}\in B_{i}$\ denote the $i^{th}$\ component of
$\mathbf{b}\in\mathbf{B}$. \ Then the constraint $C_{\mathbf{R}}$\ has the
following satisfaction value:%
\begin{equation}
C_{\mathbf{R}}\left(  \mathbf{b}\left(  \mathbf{y}_{1}\right)  ,\ldots
,\mathbf{b}\left(  \mathbf{y}_{k}\right)  \right)  :=\operatorname*{E}%
_{\sigma\in S_{k}}\left[  V\left(  \left(  \mathbf{y}_{1}\right)
_{\sigma\left(  1\right)  },\ldots,\left(  \mathbf{y}_{k}\right)
_{\sigma\left(  k\right)  },\left(  \mathbf{b}\left(  \mathbf{y}_{1}\right)
\right)  _{\sigma\left(  1\right)  },\ldots,\left(  \mathbf{b}\left(
\mathbf{y}_{k}\right)  \right)  _{\sigma\left(  k\right)  }\right)  \right]  ,
\end{equation}
where we fix some ordering of the $\mathbf{y}_{i}$'s, like $\mathbf{y}%
_{1}<\cdots<\mathbf{y}_{k}$. \ In words, we can think of $C_{\mathbf{R}}$\ as
an algorithm that first randomly permutes the $k$-tuples\ $\mathbf{y}%
_{1},\ldots,\mathbf{y}_{k}$ and $\mathbf{b}\left(  \mathbf{y}_{1}\right)
,\ldots,\mathbf{b}\left(  \mathbf{y}_{k}\right)  $, and that then checks
\textquotedblleft satisfaction of $V$ along the diagonal\textquotedblright:
i.e., does Arthur accept if, for each $i\in\left[  k\right]  $, Merlin$_{i}$
is asked the $i^{th}$ question in $\mathbf{y}_{i}$\ and responds with the
$i^{th}$\ answer in $\mathbf{b}\left(  \mathbf{y}_{i}\right)  $?

In this way, we ensure the following four properties:

\begin{enumerate}
\item[(1)] $\varphi$ has density $\alpha=1$\ in the sense of Theorem
\ref{bhhsthm}, since it includes a constraint for every possible subset of $k$ variables.

\item[(2)] $\varphi$ has alphabet size $\left\vert \Sigma\right\vert
=\left\vert \mathbf{B}\right\vert =\left\vert B_{1}\right\vert \cdots
\left\vert B_{k}\right\vert $.

\item[(3)] $\operatorname*{SAT}\left(  \varphi\right)  \geq\omega\left(
G\right)  $. \ To see this: given any strategy $\left(  b_{i}:Y_{i}\rightarrow
B_{i}\right)  _{i\in\left[  k\right]  }$\ for $G$ that achieves value $\omega
$, we can easily construct an assignment $\mathbf{b}:\mathbf{Y}\rightarrow
\mathbf{B}$\ to $\varphi$\ that achieves value $\omega$, by setting%
\begin{equation}
\mathbf{b}\left(  \mathbf{y}\right)  :=\left(  b_{1}\left(  \left(
\mathbf{y}\right)  _{1}\right)  ,\ldots,b_{k}\left(  \left(  \mathbf{y}%
\right)  _{k}\right)  \right)
\end{equation}
for all $\mathbf{y}\in\mathbf{Y}$.

\item[(4)] $\operatorname*{SAT}\left(  \varphi\right)  \leq\omega\left(
G\right)  $ (so in fact $\operatorname*{SAT}\left(  \varphi\right)
=\omega\left(  G\right)  $). \ To see this: fix any assignment $\mathbf{b}%
:\mathbf{Y}\rightarrow\mathbf{B}$. \ Then for each $i\in\left[  k\right]  $,
let $\mathcal{D}_{i}$\ be the probability distribution over functions
$b_{i}:Y_{i}\rightarrow B_{i}$\ obtained by first choosing $y_{j}\in Y_{j}%
$\ uniformly at random for all $j\neq i$, and then considering the
function\ $b_{i}\left(  y_{i}\right)  :=\left(  \mathbf{b}\left(  y_{1}%
,\ldots,y_{k}\right)  \right)  _{i}$. \ Then%
\begin{align}
\operatorname*{SAT}\left(  \varphi\right)   &  =\operatorname*{E}%
_{\mathbf{y}_{1},\ldots,\mathbf{y}_{k}\in\mathbf{Y}}\left[  V\left(  \left(
\mathbf{y}_{1}\right)  _{1},\ldots,\left(  \mathbf{y}_{k}\right)  _{k},\left(
\mathbf{b}\left(  \mathbf{y}_{1}\right)  \right)  _{1},\ldots,\left(
\mathbf{b}\left(  \mathbf{y}_{k}\right)  \right)  _{k}\right)  \right] \\
&  =\operatorname*{E}_{y_{1}\in Y_{1},\ldots,y_{k}\in Y_{k},b_{1}%
\thicksim\mathcal{D}_{1},\ldots,b_{k}\thicksim\mathcal{D}_{k}}\left[  V\left(
y_{1},\ldots,y_{k},b_{1}\left(  y_{1}\right)  ,\ldots,b_{k}\left(
y_{k}\right)  \right)  \right] \\
&  \leq\omega\left(  G\right)  ,
\end{align}
where the last line used convexity.
\end{enumerate}

Now suppose we choose a random subset $\mathbf{I}\subseteq\mathbf{Y}$ of size%
\begin{equation}
\kappa=\varepsilon^{-\Lambda}\log\left\vert \Sigma\right\vert =\varepsilon
^{-\Lambda}\log\left(  \left\vert B_{1}\right\vert \cdots\left\vert
B_{k}\right\vert \right)  ,
\end{equation}
and consider a restriction $\varphi_{\mathbf{I}}$\ of $\varphi$\ to the subset
of variables $\left\{  \mathbf{b}\left(  \mathbf{y}\right)  \right\}
_{\mathbf{y}\in\mathbf{I}}$. \ Then by Theorem \ref{bhhsthm}, together with
properties (1) and (2) above, we have%
\begin{equation}
\operatorname*{E}_{\mathbf{I}}\left[  \operatorname*{SAT}\left(
\varphi_{\mathbf{I}}\right)  \right]  \leq\operatorname*{SAT}\left(
\varphi\right)  +\varepsilon.
\end{equation}
Furthermore, for each $i\in\left[  k\right]  $, let $S_{i}\subseteq Y_{i}$ be
chosen uniformly at random subject to $\left\vert S_{i}\right\vert =\kappa$,
and let $S_{i}=\left\{  y_{i1},\ldots,y_{i\kappa}\right\}  $, fixing a
uniformly-random ordering of $y_{i1},\ldots,y_{i\kappa}$. \ Also let
$S:=S_{1}\times\cdots\times S_{k}$. \ Then for each $j\in\left[
\kappa\right]  $, let $\mathbf{y}_{j}:=\left(  y_{1j},\ldots,y_{kj}\right)  $,
and let $\mathbf{I}:=\left\{  \mathbf{y}_{1},\ldots,\mathbf{y}_{\kappa
}\right\}  $. \ Then reusing the same argument from property (3) above, we
have $\omega\left(  G_{S}\right)  \leq\operatorname*{SAT}\left(
\varphi_{\mathbf{I}}\right)  $ for every $S$. \ But the uniform distribution
over $S$'s (and over the orderings of the elements in each $S_{i}$) induces
the uniform distribution over $\mathbf{I}$'s. \ It follows that%
\begin{equation}
\operatorname*{E}_{S}\left[  \omega\left(  G_{S}\right)  \right]
\leq\operatorname*{E}_{\mathbf{I}}\left[  \operatorname*{SAT}\left(
\varphi_{\mathbf{I}}\right)  \right]  .
\end{equation}
Finally, by property (4) we have $\operatorname*{SAT}\left(  \varphi\right)
\leq\omega\left(  G\right)  $. \ Combining, we get%
\begin{equation}
\operatorname*{E}_{S}\left[  \omega\left(  G_{S}\right)  \right]
\leq\operatorname*{E}_{\mathbf{I}}\left[  \operatorname*{SAT}\left(
\varphi_{\mathbf{I}}\right)  \right]  \leq\operatorname*{SAT}\left(
\varphi\right)  +\varepsilon\leq\omega\left(  G\right)  +\varepsilon,
\end{equation}
which is what we wanted to show.
\end{proof}

We are now ready to prove that $\mathsf{AM}\left(  k\right)  =\mathsf{AM}$.

\begin{theorem}
\label{amkam}$\mathsf{AM}\left(  k\right)  =\mathsf{AM}$ for all
$k=\operatorname*{poly}\left(  n\right)  $.
\end{theorem}

\begin{proof}
Let $L\in\mathsf{AM}\left(  k\right)  $. \ Then just like in Theorem
\ref{am2am}, an $\mathsf{AM}\left(  k\right)  $\ protocol for checking whether
a string is in $L$\ can be represented as a $k$-player\ free game $G=\left(
Y_{1},\ldots,Y_{k},B_{1},\ldots,B_{k},V\right)  $, where $Y_{i}=B_{i}=\left\{
0,1\right\}  ^{p\left(  n\right)  }$ for all $i\in\left[  k\right]  $ and some
polynomial $p$, and $V$ is computable in randomized $\operatorname*{poly}%
\left(  n\right)  $\ time.

Let $\varepsilon:=\frac{1}{12}$\ and $\kappa:=\varepsilon^{-\Lambda}p\left(
n\right)  $. \ Suppose we choose $S_{i}\subseteq Y_{i}$ uniformly at random
subject to $\left\vert S_{i}\right\vert =\kappa$ for all $i\in\left[
k\right]  $, then set $S:=S_{1}\times\cdots\times S_{k}$.\ \ Then by Theorem
\ref{subk},%
\begin{equation}
\operatorname*{E}_{S}\left[  \omega\left(  G_{S}\right)  \right]  \leq
\omega\left(  G\right)  +\frac{1}{12}.
\end{equation}
But this immediately gives us our $\mathsf{AM}$\ simulation, as follows.
\ First Arthur chooses $S_{1},\ldots,S_{k}\subseteq\left\{  0,1\right\}
^{p\left(  n\right)  }$ uniformly at random, subject as above to $\left\vert
S_{1}\right\vert =\cdots=\left\vert S_{k}\right\vert =\kappa$, and lets
$S=S_{1}\times\cdots\times S_{k}$. \ He then sends descriptions of
$S_{1},\ldots,S_{k}$ to Merlin, using $k\kappa\cdot p\left(  n\right)
=O(k\cdot p\left(  n\right)  ^{2})$ bits. \ Next Merlin replies with a
$k$-tuple of strategies $\left(  b_{i}:S_{i}\rightarrow B_{i}\right)
_{i\in\left[  k\right]  }$, which again takes $O(k\cdot p\left(  n\right)
^{2})$\ bits. \ Let%
\begin{equation}
\omega_{S}:=\operatorname*{E}_{y_{1}\in S_{1},\ldots,y_{k}\in S_{k}}\left[
V\left(  y_{1},\ldots,y_{k},b_{1}\left(  y_{1}\right)  ,\ldots,b_{k}\left(
y_{k}\right)  \right)  \right]
\end{equation}
be the subsampled success probability; notice that $\omega_{S}\leq
\omega\left(  G_{S}\right)  $ for all $S$. \ Then finally, Arthur computes an
estimate $\widetilde{\omega}_{S}$\ such that%
\begin{equation}
\Pr\left[  \left\vert \widetilde{\omega}_{S}-\omega_{S}\right\vert
>0.01\right]  \leq\exp\left(  -\kappa\right)
\end{equation}
which he can do in randomized $\operatorname*{poly}\left(  n\right)  $ time,
and accepts if and only if $\widetilde{\omega}_{S}\geq0.51$. \ (One small
difference from Theorem \ref{am2am} is that, even if $V$ is deterministic, in
general Arthur will still need to estimate $\omega_{S}$\ rather than computing
it exactly. \ The reason is that $\omega_{S}$\ is\ an average of $\left\vert
S_{1}\right\vert \cdots\left\vert S_{k}\right\vert =\kappa^{k}$ terms, and
$\kappa^{k}$ is more than polynomial whenever $k$ is more than constant.)

The completeness and soundness arguments are precisely the same as in Theorem
\ref{am2am}.
\end{proof}

Let us also show how, by using Theorem \ref{subk}, we can go back and tighten
Corollaries \ref{amkexpcor} and \ref{3satcor2}\ from Section \ref{KMERLIN}.

\begin{corollary}
\label{betteramkexp}Let $G$ be a $k$-player free game, with question sets
$Y_{1},\ldots,Y_{k}$\ and answer sets $B_{1},\ldots,B_{k}$ (assume $\left\vert
B_{i}\right\vert \geq2$\ for all $i\in\left[  k\right]  $). \ There exists a
deterministic algorithm that approximates $\omega\left(  G\right)  $\ to
within additive error $\pm\varepsilon$, in time%
\begin{equation}
\exp\left(  \varepsilon^{-O\left(  1\right)  }\log\left(  \left\vert
Y_{1}\right\vert \cdots\left\vert Y_{k}\right\vert \right)  \log\left(
\left\vert B_{1}\right\vert \cdots\left\vert B_{k}\right\vert \right)
\right)  =n^{\varepsilon^{-O\left(  1\right)  }\log n},
\end{equation}
where $n=\left\vert Y_{1}\right\vert \left\vert B_{1}\right\vert
\cdots\left\vert Y_{k}\right\vert \left\vert B_{k}\right\vert $\ is the input size.
\end{corollary}

\begin{proof}
Let $\kappa:=\varepsilon^{-\Lambda}\log\left(  \left\vert B_{1}\right\vert
\cdots\left\vert B_{k}\right\vert \right)  $. \ Then we simply need to loop
over all possible subsets%
\begin{equation}
S=S_{1}\times\cdots\times S_{k}\subseteq Y_{1}\times\cdots\times Y_{k}%
\end{equation}
such that $\left\vert S_{i}\right\vert =\kappa$ for all $i\in\left[  k\right]
$. \ For each one, we compute the value $\omega\left(  G_{S}\right)  $\ via a
brute-force search over all possible strategy $k$-tuples $\left(  b_{i}%
:S_{i}\rightarrow B_{i}\right)  _{i\in\left[  k\right]  }$. \ Then we output
$\widetilde{\omega}:=\operatorname*{E}_{S}\left[  \omega\left(  G_{S}\right)
\right]  $\ as our estimate for $\omega\left(  G\right)  $.

The correctness of this algorithm---i.e., the fact that $\left\vert
\widetilde{\omega}-\omega\right\vert \leq\varepsilon$---follows from Theorem
\ref{subk}. \ For the running time, note that the number of possible subsets
$S$\ is%
\begin{equation}
\binom{\left\vert Y_{1}\right\vert }{\kappa}\cdots\binom{\left\vert
Y_{k}\right\vert }{\kappa}\leq\left(  \left\vert Y_{1}\right\vert
\cdots\left\vert Y_{k}\right\vert \right)  ^{\kappa}\leq n^{\varepsilon
^{-O\left(  1\right)  }\log n}.
\end{equation}
Also, for each $S$, the number of possible strategy $k$-tuples is $\left\vert
B_{1}\right\vert ^{\kappa}\cdots\left\vert B_{k}\right\vert ^{\kappa}\leq
n^{\varepsilon^{-O\left(  1\right)  }\log n}$. \ Hence the total running time
is $n^{\varepsilon^{-O\left(  1\right)  }\log n}$\ as well.
\end{proof}

Corollary \ref{betteramkexp}, in turn, has the following further corollary.

\begin{corollary}
\label{better3satcor}Assuming the Randomized ETH, any $\mathsf{AM}\left(
k\right)  $\ protocol for \textsc{3Sat{}{}}\ with a constant
completeness/soundness gap\ must use $\Omega(\sqrt{n})$\ bits of communication
in total. \ (Also, if Arthur's verification procedure is deterministic, then
it suffices to assume the ordinary ETH.)
\end{corollary}

\begin{proof}
Suppose there existed an $\mathsf{AM}\left(  k\right)  $\ protocol for
\textsc{3Sat{}{}}, which used $q\left(  n\right)  =n^{O\left(  1\right)  }%
$\ bits of communication in total, and which had a completeness/soundness gap
of, say, $2/3$\ versus $1/3$\ (the exact constants will be irrelevant). \ Set
$\varepsilon:=1/10$. \ Then by Corollary \ref{betteramkexp}, we can
approximate the Merlins' maximum winning probability $\omega$\ to within
$\pm\varepsilon$\ by a deterministic algorithm that makes $q\left(  n\right)
^{\varepsilon^{-O\left(  1\right)  }\log q\left(  n\right)  }=2^{O(\log
^{2}q\left(  n\right)  )}$\ evaluations of Arthur's verification function $V$.
\ Furthermore, each $V$ evaluation takes $\operatorname*{poly}\left(
n\right)  $\ time by a randomized algorithm if $V$ is randomized (even
counting the time needed to amplify to $\exp(-q\left(  n\right)  ^{2})$\ error
probability), or $\operatorname*{poly}\left(  n\right)  $\ time by a
deterministic algorithm if $V$ is deterministic. \ Thus, the algorithm's total
running time is $2^{O(\log^{2}q\left(  n\right)  )}\operatorname*{poly}\left(
n\right)  $. \ Moreover, the algorithm lets us decide whether $\omega\geq
2/3$\ or $\omega\leq1/3$, and hence whether our original \textsc{3Sat{}{}%
}\ instance was satisfiable. \ On the other hand, \textsc{3Sat{}{}}\ must take
$2^{\Omega\left(  n\right)  }$\ time assuming the ETH. \ Combining, we obtain
$q\left(  n\right)  =\Omega(\sqrt{n})$.
\end{proof}

\section{Conclusions and Open Problems\label{OPEN}}

In this paper, we saw how a deceptively simple problem---understanding the
power of $\mathsf{AM}\left(  2\right)  $\ protocols, and the complexity of
approximating free games---hides a wealth of interesting phenomena. \ On the
one hand, the fact that a two-prover game $G$ is free leads to a
quasipolynomial-time approximation algorithm for $\omega\left(  G\right)  $,
and even a proof of $\mathsf{AM}\left(  2\right)  =\mathsf{AM}$. \ On the
other hand, the fact that the Merlins still can't communicate leads to
quasipolynomial-time \textit{hardness} (assuming the ETH), and to an
$\widetilde{O}(\sqrt{n})$-communication $\mathsf{AM}\left(  2\right)
$\ protocol for \textsc{3Sat}.

While we managed to give nearly-matching upper and lower bounds for the
complexity of \textsc{FreeGame}, numerous open problems remain, both about
free games themselves, and about the applicability of our techniques to other
problems. \ We now list twelve.

\begin{enumerate}
\item[(1)] Can we improve our result $\mathsf{NTIME}\left[  n\right]
\subseteq\mathsf{AM}_{n^{1/2+o\left(  1\right)  }}\left(  2\right)  $\ to
$\mathsf{NTIME}\left[  n\right]  \subseteq\mathsf{AM}_{\widetilde{O}(\sqrt
{n})}\left(  2\right)  $? \ This would follow if, for example, we could get
the \textquotedblleft best of both worlds\textquotedblright\ between the two
PCP theorems of Dinur \cite{dinur}\ and Moshkovitz and Raz \cite{mr}, and
achieve $n\operatorname*{polylog}n$\ size together with a $1$ vs.\ $\delta
$\ completeness/soundness gap.

\item[(2)] Assuming the ETH, can we completely close the gap between our
$n^{O(\varepsilon^{-2}\log n)}$\ upper bound and $n^{\widetilde{\Omega
}(\varepsilon^{-1}\log n)}$\ lower bound on the complexity of
\textsc{FreeGame}$_{\varepsilon}$? \ What is the right dependence on
$\varepsilon$? \ Also, given a PCP $\phi$\ of size $N$, is there an
$\mathsf{AM}\left(  2\right)  $\ protocol for verifying $\phi$'s
satisfiability that uses $O(\sqrt{N})$\ communication rather than $O(\sqrt
{N}\log N)$? \ (In other words, in our hardness result, can we at least
eliminate the $\log$\ factor that comes from the birthday game, if not the
$\log$\ or larger factors from the PCP reduction?)

\item[(3)] We gave two different algorithms for approximating the value of a
$k$-player free game with $k\geq3$: one that took $n^{O(\varepsilon^{-2}%
k^{2}\log n)}$\ time (using a recursive reduction to $\left(  k-1\right)
$-player games),\ and one that took $n^{\varepsilon^{-O\left(  1\right)  }\log
n}$\ time (using subsampling). \ Can we get the \textquotedblleft best of both
worlds,\textquotedblright\ and give an algorithm that takes $n^{O(\varepsilon
^{-2}\log n)}$\ time? \ If so, this would imply that, assuming the ETH, any
$\mathsf{AM}\left(  k\right)  $\ protocol for \textsc{3Sat}\ with a $1$\ vs.
$1-\varepsilon$\ completeness/soundness gap requires $\Omega(\sqrt{\varepsilon
n})$\ total communication, regardless of $k$.

\item[(4)] Can we prove a \textquotedblleft Birthday Repetition
Theorem\textquotedblright\ for the birthday game $G_{\phi}^{k\times\ell}$?
\ In other words, can we show that the Merlins' cheating probability
$\omega(G_{\phi}^{k\times\ell})$ continues to decrease as $\exp\left(
-k\ell/N\right)  $, if the product $k\ell$\ is larger than $N$? \ If not, then
can we give some other $\mathsf{AM}\left(  k\right)  $\ protocol for
\textsc{3Sat}\ that has a $1$ vs.\ $\delta$\ completeness/soundness gap for
arbitrary $\delta=\delta(n)>0$, and that uses $n^{1/2+o\left(  1\right)
}\operatorname*{polylog}\left(  1/\delta\right)  $\ communication, rather than
$n^{1/2+o\left(  1\right)  }\operatorname*{poly}\left(  1/\delta\right)  $?
\ Directly related to that, given a free game $G$, can we show that deciding
whether $\omega\left(  G\right)  =1$\ or $\omega\left(  G\right)  <\delta$
requires $n^{\widetilde{\Omega}\left(  \frac{\log n}{\log1/\delta}\right)  }%
$\ time, assuming the ETH? \ Recall that Theorem \ref{improvedalg} gave
an\ $n^{O\left(  1+\frac{\log n}{\log1/\delta}\right)  }$ algorithm for that
problem, while Theorem \ref{freegamecor3}\ gave an\ $n^{\operatorname*{poly}%
\left(  \delta\right)  \cdot\left(  \log n\right)  ^{1-o\left(  1\right)  }}%
$\ lower bound assuming the ETH. \ Between these, we conjecture that the upper
bound is tight, but the PCP and parallel-repetition\ machinery that currently
exists seems insufficient to show this.

\item[(5)] Given an \textit{arbitrary} two-prover game $G$ and positive
integers $k$ and $\ell$, what are the necessary and sufficient conditions on
$G,k,\ell$\ for us to have $\omega(G^{k\times\ell})\leq\omega(G^{1\times
1})^{\Omega\left(  k\ell\right)  }$? \ In other words, when exactly does
birthday repetition work? \ Recall from Section \ref{3SATINT}\ that, if
$\omega(G^{1\times1})=1-\varepsilon$, then we can only ever hope to do
birthday repetition when $k=O(\frac{1}{\varepsilon}\log\left\vert B\right\vert
)$ and $\ell=O(\frac{1}{\varepsilon}\log\left\vert A\right\vert )$. \ Can we
at least do birthday repetition up to that limit?

\item[(6)] Can we generalize the Parallel Repetition Theorem, as well as Rao's
concentration bound (Theorem \ref{raothm}), to $k$-player free games for
arbitrary $k$? \ This would let us amplify $\mathsf{AM}\left(  k\right)  $
protocols for $k>2$, though as usual with a polynomial blowup in communication cost.

\item[(7)] Is our result that $\mathsf{NTIME}\left[  n\right]  \subseteq
\mathsf{AM}_{n^{1/2+o\left(  1\right)  }}\left(  2\right)  $---that is, the
existence of our \textsc{3Sat}\ protocol---non-algebrizing in the sense of
Aaronson and Wigderson \cite{awig}? \ (Recall from Proposition \ref{oraclesep}%
\ that the result is non-relativizing.)

\item[(8)] Can we give \textquotedblleft direct\textquotedblright\ proofs that
$\mathsf{AM}\left(  k\right)  =\mathsf{AM}\left(  2\right)  $ for all $k>2$,
and that any $\mathsf{AM}\left(  k\right)  $\ protocol can be made public-coin
and perfect-completeness (where \textquotedblleft direct\textquotedblright%
\ means, without using the full power of $\mathsf{AM}\left(  k\right)
=\mathsf{AM}$)?

\item[(9)] How far can we improve our approximation algorithms for free games,
if we assume that the game is also a \textit{projection game} or a
\textit{unique game}? \ Conversely, what hardness results can we prove under
those restrictions?

\item[(10)] Let $\mathsf{AM}^{\ast}\left(  2\right)  $\ be defined the same
way as $\mathsf{AM}\left(  2\right)  $, except that now the Merlins can share
an unlimited amount of quantum entanglement. \ (Their communication with
Arthur is still classical.) \ What can we say about this class? \ Does our
\textsc{3Sat}\ protocol become unsound? \ If so, then can we somehow
\textquotedblleft immunize\textquotedblright\ it\ against entangled
provers---as the spectacular work of Ito and Vidick \cite{itovidick}\ (see
also Vidick \cite{vidick:egame}) recently managed to do for the original BFL
protocol? \ In the other direction, it's currently a notorious open problem to
prove \textit{any upper bound whatsoever} on the class $\mathsf{MIP}^{\ast}%
$\ (that is, $\mathsf{MIP}$ with entangled provers): even the set of
computable languages! \ The issue is that we don't have any \textit{a priori}
upper bound on the amount of entanglement the provers might need; and the more
entanglement they use, the longer it could take to simulate them. \ Does this
problem become more tractable if we restrict attention to $\mathsf{AM}^{\ast}$
protocols: that is, to protocols with uncorrelated questions?

\item[(11)] Can we use our hardness result for \textsc{FreeGame}---or more
generally, the idea of birthday repetition---as a starting point for proving
$n^{\Omega\left(  \log n\right)  }$ hardness results for \textit{other}
problems? \ One problem of particular interest is approximate Nash
equilibrium. \ For that problem, Lipton, Markakis, and Mehta \cite{lmm}\ gave
an $n^{O(\varepsilon^{-2}\log n)}$\ approximation algorithm---indeed, one
strikingly reminiscent of our algorithm from Theorem \ref{freegamealg}---while
Hazan and Krauthgamer \cite{hazank}\ recently showed\ $n^{\Omega\left(  \log
n\right)  }$ hardness, assuming $n^{\Omega\left(  \log n\right)  }$\ hardness
for the planted clique problem.\footnote{Similarly, while this reduction is
arguably weaker than the one we give, it is not hard to show that
\textsc{FreeGame}$_{\varepsilon}$\ requires $n^{\Omega\left(  \log n\right)
}$\ time for constant $\varepsilon$, under the assumption that the planted
clique problem requires $n^{\Omega\left(  \log n\right)  }$\ time. \ We thank
Oded Regev for this observation.} \ We conjecture that, using birthday
repetition of \textsc{3Sat}, one could show $n^{\widetilde{\Omega}%
(\varepsilon^{-1}\log n)}$\ hardness for approximate Nash equilibrium assuming
only the ETH. \ This would solve an open problem explicitly raised by Hazan
and Krauthgamer.\footnote{Technically, Hazan and Krauthgamer asked for a proof
that approximate Nash equilibrium is not in $\mathsf{P}$, assuming \textsc{Max
Clique} requires $2^{\omega\left(  \sqrt{n}\right)  }$\ time. \ But this is
extremely similar to assuming the ETH.}

\item[(12)] What can we say about $\mathsf{QMA}\left(  2\right)  $, the class
that originally motivated our study of $\mathsf{AM}\left(  2\right)  $? \ Is
$\mathsf{QMA}\left(  2\right)  \subseteq\mathsf{EXP}$? \ Are the
$\widetilde{O}(\sqrt{n})$-qubit protocols for\ \textsc{3Sat}, due to Aaronson
et al.\ \cite{abdfs}\ and Harrow and Montanaro \cite{harrowmontanaro}, optimal
assuming the ETH? \ Is the \textsc{BSS}$_{\varepsilon}$\ problem from Section
\ref{QUANTUM} solvable in $n^{O\left(  \varepsilon^{-2}\log n\right)  }$ time,
as \textsc{FreeGame}$_{\varepsilon}$\ is?
\end{enumerate}

\section{Acknowledgments}

We thank Boaz Barak, Oded Regev, Avi Wigderson, and other participants at the
2013 Banff Complexity Theory workshop for helpful discussions. \ We especially
thank Peter Shor for early discussions, Ryan O'Donnell for pointing us to
\cite{avkk}\ and \cite{bhhs}, Anup Rao for clarifications about parallel
repetition, and Aram Harrow for goading us to write this paper up after a
four-year delay.

\bibliographystyle{plain}
\bibliography{thesis}

\end{document}